\documentclass[11pt]{article}
\usepackage{natbib,amsmath,amssymb,amsfonts,amsthm,dsfont,enumitem,float,bbm}
\RequirePackage[colorlinks,citecolor=blue,urlcolor=blue,breaklinks]{hyperref}
\usepackage[normalem]{ulem}
\usepackage{epsfig,longtable,xcolor,algpseudocode}
\usepackage[linesnumbered, ruled, noline]{algorithm2e}
\usepackage{wrapfig}
\usepackage{multirow}
\usepackage{subfigure}
\usepackage{caption}

\usepackage{array}
\usepackage{listings} 
\usepackage{xcolor} 
\usepackage{graphicx}

\newcommand{\bfm}[1]{\ensuremath{\mathbf{#1}}}
     \def\bA{\bfm A}     \def\cA{{\cal  A}}     
     \def\bB{\bfm B}     \def\cB{{\cal  B}}     
               
          \def\cD{{\cal  D}}     
\def\be{\bfm e}          \def\cE{{\cal  E}}     
         \def\cF{{\cal  F}}     
\def\bg{\bfm g}               
\def\bh{\bfm h}          \def\cH{{\cal  H}}     
     \def\bI{\bfm I}

          \def\cM{{\cal  M}}     
          \def\cN{{\cal  N}}

     \def\bQ{\bfm Q}          
     \def\bR{\bfm R}          
\def\bs{\bfm s}     \def\bS{\bfm S}     \def\cS{{\cal  S}}     
          \def\cT{{\cal  T}}     
\def\bu{{\bfm u}}               
\def\bv{\bfm v}               
\def\bw{\bfm w}               
\def\bx{\bfm x}     \def\bX{\bfm X}     \def\cX{{\cal  X}}     
\def\by{\bfm y}     \def\bY{{\bfm Y}}     \def\cY{{\cal  Y}}     
\def\bz{\bfm z}               
\def\bzero{\bfm 0}

\newcommand{\bfsym}[1]{\ensuremath{\boldsymbol{#1}}}
       \def \bbeta    {\bfsym{\beta}}
       \def \bdelta   {\bfsym{\xi}}

      \def \bmu      {\bfsym{\mu}}

       \def \bDelta   {\bfsym{\Delta}}
\def \bTheta   {\bfsym{\Theta}}       
          
\def \bSigma   {\bfsym{\Sigma}}

\DeclareMathOperator*{\argmin}{argmin}

\DeclareMathOperator{\var}{var}
\DeclareMathOperator{\cov}{cov}
\DeclareMathOperator{\diag}{diag}

\def \bone   {\bfsym{1}}
\def \RR	{\mathbb{R}}
\def \PP {\mathbb{P}}

\def \NN {\mathbb{N}}

\def \rb {\mathrm{rb}}
\newcommand{\EE}{\mathbb{E}}

\newcommand{\beq}  {\begin{equation}}
\newcommand{\eeq}  {\end{equation}}
\newcommand{\beqn} {\begin{eqnarray}}
\newcommand{\eeqn} {\end{eqnarray}}
\newcommand{\beqnn}{\begin{eqnarray*}}
\newcommand{\eeqnn}{\end{eqnarray*}}

\theoremstyle{definition}

\theoremstyle{plain}
\newtheorem{lem}{Lemma}
\newtheorem{cor}{Corollary}
\newtheorem{con}{Condition}
\newtheorem{prop}{Proposition}
\newtheorem{thm}{Theorem}

\newcommand{\ltwonorm}[1]{\lVert#1\rVert_2}

\newcommand{\fnorm}[1]{\|#1\|_{\mathrm{F}}}

\providecommand{\keywords}[1]
{	
  \small
  \textbf{Keywords:} #1
}

\title{ReBoot: Distributed statistical learning via refitting \\bootstrap samples}
\author{Yumeng Wang$^\dagger$, Ziwei Zhu$^{\dagger}$, Xuming He$^{\ddagger}$ \\
 \normalsize
 $^\dagger$Department of Statistics, University of Michigan \\ 
 $^{\ddagger}$Department of Statistics and Data Science, Washington University in St Louis
\date{\today} }

\begin{document}

\maketitle

\begin{abstract}
In this paper, we propose a one-shot distributed learning algorithm via \underline{re}fitting \underline{boot}strap samples, which we refer to as ReBoot. ReBoot refits a new model to mini-batches of bootstrap samples that are continuously drawn from each of the locally fitted models. It requires only one round of communication of model parameters without much memory. Theoretically, we analyze the statistical error rate of ReBoot for generalized linear models (GLM) and noisy phase retrieval, which represent convex and non-convex problems, respectively. In both cases, ReBoot provably achieves the full-sample statistical rate.
In particular, we show that the systematic bias of ReBoot, the error that is independent of the number of subsamples (i.e., the number of sites), is $O(n ^ {-2})$ in GLM, where $n$ is the subsample size (the sample size of each local site). This rate is sharper than that of model parameter averaging and its variants, implying the higher tolerance of ReBoot with respect to data splits to maintain the full-sample rate. Our simulation study demonstrates the statistical advantage of ReBoot over competing methods.
Finally, we propose FedReBoot, an iterative version of ReBoot, to aggregate convolutional neural networks for image classification. FedReBoot exhibits substantial superiority over Federated Averaging (FedAvg) within early rounds of communication.
\end{abstract}

\normalsize

\section{Introduction}
\label{sec:Introduction}

Apace with the data explosion in the digital era, it is common for modern data to be distributed across multiple or even a large number of sites. One example is the data that are continuously generated by edge devices such as mobile phones, personal computers, and smart watches. Such data, if accessible, can be used to train models that underpin modern AI applications and services. For instance, browsing history data can help recommendation systems learn customer preference and produce personalized recommendations. Another example is regarding health records from multiple clinical sites. Aggregating these datasets or the derived models can improve the learning accuracy and enhance the power of statistical tests of interest, lending further statistical support to knowledge discovery. However, there are two salient challenges of analyzing decentralized data: (a) communication of large-scale data between sites is expensive and inefficient, {because of the limitation of network bandwidth, among other things}; (b) such sensitive data as internet browsing history or health records may not be allowed to be shared for privacy or legal reasons. It is thus time-pressing to develop a new generation of statistical learning methods that can address these challenges.

Perhaps the most straightforward strategy to handle distributed datasets is the one-shot aggregation (see the left panel of Figure \ref{fig:one_shot}). Suppose the data are distributed across $m$ sites, forming $m$ sub-datasets $\{\cD ^ {(k)}\}_{k=1} ^ m$. The one-shot aggregation framework first calculates local statistics $\widehat \bbeta ^ {(k)}$ on each sub-dataset $\cD ^ {(k)}$ and then combines all the local statistics to form an aggregated statistic $\widetilde \bbeta$. This strategy requires only one round of communication of subsample-based statistics, thereby requiring low communication cost and preventing privacy leakages due to transmission of raw data.

\begin{figure*}[t]
    \centering
    \includegraphics[scale=0.5]{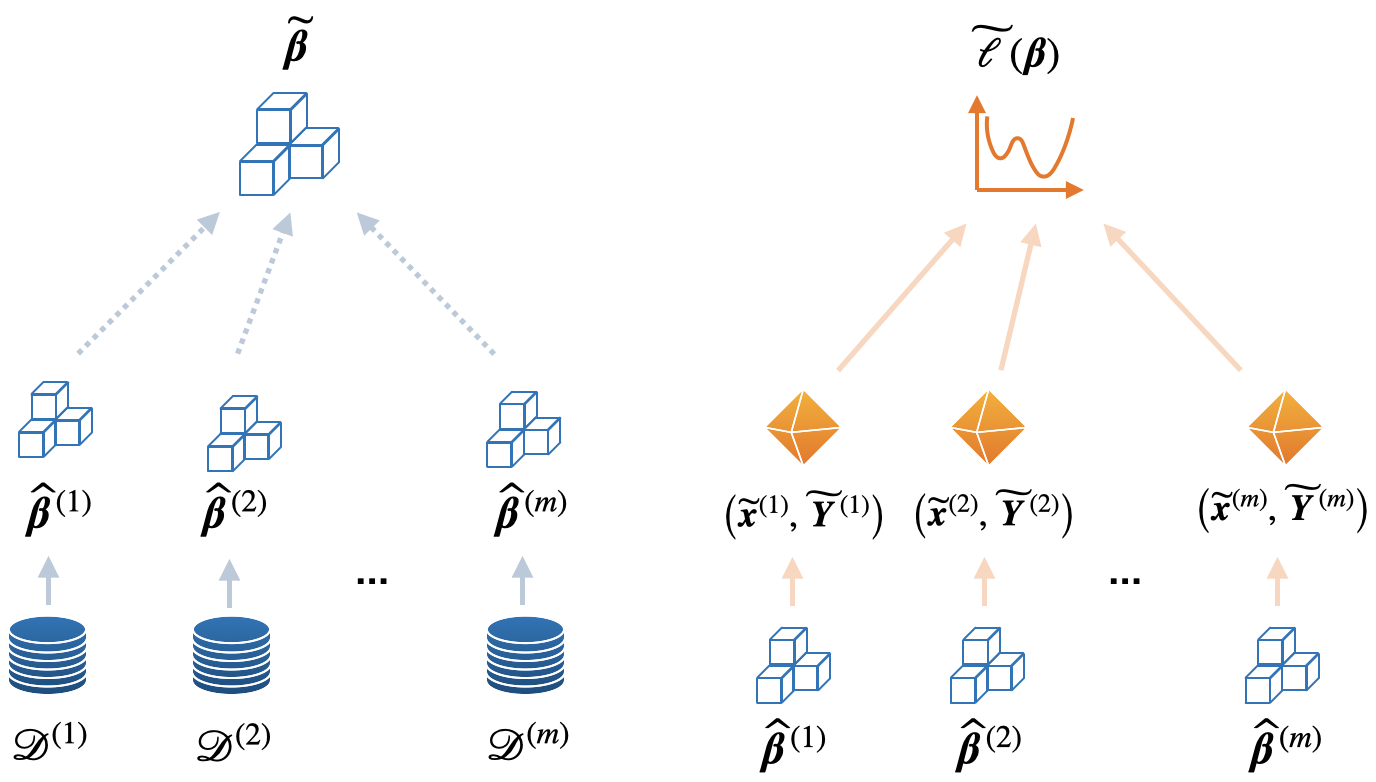} 
    \begin{tabular}{c}
    	One-shot framework \kern12em  ReBoot \kern3em
    \end{tabular}
    \caption{\textit{Left:} Illustration of one-shot aggregation. Each site trains a local model on sub-dataset $\cD ^ {(k)}$ and obtains local statistics $\bbeta ^ {(k)}$, $k \in [m]$. A central server then combines these local statistics to form an aggregated statistic $\widetilde \bbeta$. 
    \textit{Right:} Overview of ReBoot framework. The central server generates a bootstrap sample $\big(\widetilde \bx ^ {(k)}, \widetilde Y ^ {(k)}\big)$ based on the model parameterized by $\widehat\bbeta ^ {(k)}$ for each $k \in [m]$, and then pools them to evaluate the ReBoot loss function.
    }
    \label{fig:one_shot}
\end{figure*}

Much work has been published over the past decade on one-shot parameter averaging or its variants for a wide range of statistical learning problems.  A partial list of the related literature includes \cite{Zhang2012Comunication}, \cite{chen2014split}, \cite{rosenblatt2016optimality}, \cite{lee2017communication}, \cite{battey2018distributed}, \cite{banerjee2019divide}, \cite{dobriban2020wonder}, \cite{fan2019distributed}, etc. Specifically, in the low-dimensional setup, \cite{Zhang2012Comunication} and \cite{rosenblatt2016optimality} analyzed the mean squared error (MSE) of the na\"ive averaging estimator for general empirical risk minimization (ERM) problems. They found that the averaging estimator can achieve the same statistical rate as the full-sample estimator, i.e., the standard estimator derived from the entire data, provided that the number of machines and the parameter dimension are not large. \cite{Zhang2012Comunication} also proposed a new approach called subsampled average mixture (SAVGM) that averages debiased local estimators, which incurs less bias than the na\"ive averaging and thus allows more data splits while preserving the full-sample rate. \cite{liu2014distributed} proposed a KL-divergence-based averaging to aggregate local MLEs in point estimation problems, which provably yields the minimax optimal rate. \cite{banerjee2019divide} studied the averaging approach in non-standard problems where the statistical rate is typically slower than root-$n$, instantiated by the least squares estimator in isotonic regression. Their work unveiled a super-efficiency phenomenon: the averaging estimator \emph{outperforms} the full-sample estimator in pointwise inference, while in a uniform sense, the averaging estimator's performance worsens as the number of subsamples {(the number of sites)} increases. In the high-dimensional setup, \cite{chen2014split} proposed a split-and-conquer approach for penalized regression problems in high-dimensional generalized linear models (GLM). It uses majority voting across different machines to estimate the support of the true model and then applies a weighted average of the local estimators within the selected model to obtain the final distributed estimator. \cite{lee2017communication} and \cite{battey2018distributed} proposed to average local debiased lasso estimators or test statistics in high-dimensional sparse (generalized) linear models. They showed that the proposed distributed test and estimator can achieve full-sample efficiency and accuracy if the number of subsamples is not large. \cite{dobriban2020wonder} studied distributed ridge regression in a high-dimensional random-effects model and investigated the MSE of a weighted-average of local ridge regression estimators. \cite{Panigrahi2024} proposed a data aggregation scheme with model selection performed at each site.  With a great variety of the problem setups, a common finding in the literature is that one-shot averaging of model parameters can achieve full-sample statistical accuracy whenever the subsample size is sufficiently large. 

In this paper, we focus on a different one-shot aggregation method based on  \underline{re}fitting \underline{boot}strap samples from local models. We refer to this method as ReBoot. At a high level, ReBoot consists of three main steps:
\begin{enumerate}
    \item Train a local model on each subsample $\cD ^ {(k)}$ and send the model parameter estimate $\widehat\bbeta ^ {(k)}$ to the central server; 
    \item For each $k \in [m]$, the central server generates a bootstrap sample $\big(\widetilde \bx ^ {(k)}, \widetilde Y ^ {(k)}\big)$ based on the model parameterized by $\widehat\bbeta ^ {(k)}$; 
    \item The central server minimizes the ReBoot loss function:  $$\widetilde \ell \big(\bbeta \big) := \frac{1}{m} \sum_{k = 1} ^ m \EE \big\{\ell \big(\bbeta \, ;\,  \big(\widetilde \bx ^ {(k)}, \widetilde Y ^ {(k)}\big)\big) |\, \widehat\bbeta ^ {(k)}\big\},$$ where the expectation is taken with respect to all the bootstrap samples conditional on the local estimators, and $\ell (\cdot)$ is a loss function which will be specified later. 
\end{enumerate}
In practice, we apply mini-batch stochastic gradient descent (SGD) to minimize $\widetilde \ell(\bbeta)$, which avoids generating and storing large bootstrap samples in memory. We provide the implementation details of ReBoot in Algorithm \ref{alg:reboot}. The right panel of Figure \ref{fig:one_shot} illustrates the ReBoot framework. The motivation behind ReBoot is that a linear combination of the loss functions of the subsamples can recover the full-sample loss function, while a linear combination of the subsample-based model parameters cannot recover the global parameter. In other words, it is the \emph{loss functions}, rather than the model parameters, that are additive with respect to subsamples. By drawing bootstrap samples from local models and then pooling them to evaluate the loss function, ReBoot aims to reconstruct the global loss function and then minimizes it to obtain the aggregated estimator. 
Given the design of ReBoot, one can see that it enjoys at least three desirable properties in distributed learning setups: (i) privacy protection: ReBoot communicates only the local estimators or models, thereby avoiding leakages of instance-specific information; (ii) communication efficiency: ReBoot requires only one round communication of $m$ sets of model parameters, whose size is often much smaller than that of the raw data themselves; (iii) wide applicability: ReBoot is a generic rather than model-specific algorithm for aggregating local models. 

The most important advantage of ReBoot is statistical: ReBoot tolerates \emph{smaller subsample sizes}, or equivalently, \emph{more data splits} than averaging methods to retain the full-sample error rate. To demonstrate this, we rigorously derive the statistical rate of ReBoot under distributed GLM and noisy phase retrieval, which represent convex and non-convex problems, respectively. Suppose the data are uniformly distributed across $m$ sites and each site has $n$ subjects. Under distributed GLM, we show that the systematic bias of ReBoot, which is independent of the number of sites $m$, is of order $O(n ^ {-2})$, while those of na\"ive averaging and SAVGM are respectively $O(n ^ {-1})$ and $O(n ^ {-1.5})$. With the total sample size fixed, the sharper bias allows ReBoot to tolerate a larger number of sample splits while retaining the desirable full-sample rate. Under distributed noisy phase retrieval, we show that the systematic bias of ReBoot is of order $O(n ^ {-1})$; to the best of our knowledge, we have not seen any competing result in the same context. 

The idea of bootstrap aggregation for distributed learning already appeared in \cite{liu2014distributed} as a practical implementation of the KL-averaging to aggregate the MLEs for the exponential family. \cite{liu2014distributed} and \cite{han2016bootstrap} analyzed the asymptotic properties of this approach, and \cite{han2016bootstrap} proposed several variance reduction techniques to reduce bootstrap noise and thus relax the requirement of the bootstrap sample sizes to obtain the full-sample efficiency. We clarify the main differences between our work and theirs in the following aspects: 
\begin{enumerate}
    \item We focus on general supervised learning setups, while \cite{liu2014distributed} and \cite{han2016bootstrap} mainly focus on unsupervised learning problems. 
    \item Our analysis accommodates general loss functions and data distributions, while \cite{liu2014distributed} and \cite{han2016bootstrap} focus on the MLE problems under the exponential family. 
    \item Our main theoretical interest is to understand the systematic bias of ReBoot, which is independent of the number of data splits $m$ and thus cannot be reduced by increasing $m$. With the total sample size $N$ fixed, this bias determines the largest number of data splits one can have to maintain the full-sample efficiency and is widely acknowledged as a crucial statistical criteria to evaluate a one-shot distributed learning algorithm. To rigorously characterize this bias, we take a non-asymptotic approach and explicitly track the dependence of $m$ in the statistical rate. In contrast, the analysis of \cite{liu2014distributed} and \cite{han2016bootstrap} is asymptotic: they mainly focus on the first-order term with respect to $N$ and do not track $m$ in the second-order term. The rate of the systematic bias thus remains unclear therein, and so does the maximum data splits allowed to retain the full-sample efficiency. 
\end{enumerate}

Finally, we discuss recent development on multi-round communication algorithms. Multi-round communication has been found remarkably useful to alleviate the aforementioned restriction of the one-shot approaches on the subsample size to retain the full-sample efficiency. A natural multi-round approach, free of the subsample restriction, is performing gradient descent in a distributed manner, where the global gradient can be re-assembled by averaging all the local gradients. The problem with this proposal, however, lies in its communication efficiency: the number of communication rounds can scale polynomially with the sample size \citep{shamir2014communication}. To reduce the communication cost, \cite{shamir2014communication} proposed a novel distributed optimization framework called DANE, short for Distributed Approximate NEwton. In each iteration, DANE asks local machines to take an approximate Newton step based on global gradient and local Hessian and then transmit the updates to the central machine for averaging. For quadratic objectives, DANE provably enjoys a linear convergence rate that improves with the subsample size under reasonable assumptions, which implies that it can take just a constant number of iterations to reach the optimum.   \cite{jordan2018communication} and \cite{wang2017efficient} further applied this approximate Newton strategy to design new distributed algorithms for high-dimensional setups. \cite{jordan2018communication} referred to their algorithm as CSL, short for Communication-efficient Surrogate Likelihood. They also applied CSL to Bayesian inference for regular parametric models. There have also been recent works on new multi-round distributed algorithms for problems that are not amenable to CSL or DANE, including support vector machines \citep{wang2019distributed}, principal component analysis \citep{chen2021distributed}, quantile regression \citep{battey2021communication}, etc. Regarding the relationship between multi-round approaches and one-shot ones, we view a one-shot approach as the cornerstone for a multi-round approach, because it is natural to derive a multi-round algorithm from a one-shot one. Similarly to deriving FedAvg \citep{mcmahan2017communication} from one-shot averaging, in the real data analysis section, we propose FedReBoot, a multi-round extension of ReBoot, to train a convolutional neural network on decentralized data. We find that the CNN trained by FedReBoot enjoys higher prediction accuracy than that by FedAvg in the early rounds of communication, which we attribute to the higher statistical efficiency of ReBoot than averaging.

The rest of this paper is organized as follows. In Section \ref{sec:Problem setup}, we introduce GLM and noisy phase retrieval under the context of decentralized data. In Section \ref{sec:Methodology}, we elucidate the ReBoot algorithm. In Section \ref{sec:Statistical analysis}, we present the theoretical guarantee for ReBoot in the two aforementioned problems. In Section \ref{sec:Implementation}, we provide a practical implantation of ReBoot algorithm. In Section \ref{sec:Numerical studies}, we numerically compare ReBoot with existing methods via simulation. In Section \ref{sec:Real data}, we apply ReBoot to the Fashion-MNIST dataset \citep{xiao2017fashion} to learn a convolutional neural network (CNN) in a distributed fashion. We also propose FedReBoot and apply it to the Fashion-MNIST dataset in Section \ref{sec:Real data}. All the proofs are given in supplementary material.

\subsection{Notation}

We first introduce the notation that is used throughout this paper. By convention, we use regular letters for scalars, bold lower-case letters for vectors and bold capital letters for both matrices and tensors of order three or higher. We use $[n]$ to denote the set $\{1, \ldots, n\}$ for any positive integer $n$. We use $|\cdot|$ to denote  absolute value or cardinality of a set. Given $a, b \in \RR$, let $a \vee b$ denote the maximum of $a$ and $b$. For any function $f: \RR \rightarrow \RR$, we put primes in its superscript to denote its derivative, and the number of primes refers to the order of the derivative. For instance, $f''''$ is the fourth-order derivative function of $f$. For any $p$-dimensional vector $\bx = (x_1 \ldots x_p)^\top$ and $q \in[1, \infty)$, we define $\|\bx\|_q := (\sum_{i=1}^{p} |x_i|^q) ^ {1/q}$ and $\|\bx\|_\infty := \max _{i \in [p]} |x_i|$. 

Let $\be_j$ denote the unit vector with the $j$th element equal to one and other elements equal to zeros, and let $\bone_p$ denote the $p$-dimensional all-one vector. For any matrix $\bX \in \RR^{n_1 \times n_2}$, we use $\|\bX\|_2$ and $\fnorm{\bX}$ to denote the operator norm and the Frobenius norm of $\bX$ respectively. For any $q_1, q_2 \in[1, \infty]$, we use $\|\bX\|_{q_1 \rightarrow q_2} := \sup _{\|\bu\|_{q_1}=1}\|\bX \bu\|_{q_2}$ to denote its $q_1$-to-$q_2$ operator norm, where $\bX$ is viewed as a representation of a linear map from $(\RR^{n_1},\|\cdot\|_{q_1})$ to $(\RR^{n_2},\|\cdot\|_{q_2})$. For a symmetric matrix $\bX$, we use $\lambda_j(\bX)$ to denote the $j$th largest eigenvalue of $\bX$. For convenience, we also use $\lambda_{\max}(\bX)$ and $\lambda_{\min}(\bX)$ to denote the maximum and minimum eigenvalue of $\bX$. Let $\bI_p$ denote the $p \times p$ identity matrix. Given $\bx = (x_1 \ldots x_p)^\top$, we use $\diag(\bx)$ to denote the $p \times p$ diagonal matrix whose $j$th diagonal entry is $x_j$ for any $j \in [p]$. Let $\otimes$ denote the outer product. Given a $k$th-order symmetric tensor $\bA \in \RR ^ {p ^ k}$, for any $\bx \in \RR ^ p$, define the tensor product $\bA(\underbrace{\bx \otimes \ldots \otimes \bx}_{k - 1})$, which is in $\RR ^ p$, such that 
\[
    [\bA(\underbrace{\bx \otimes \ldots \otimes \bx}_{k - 1})]_i = \sum_{j_1, j_2, \ldots j_{k - 1} \in [p]} A_{i, j_1, \ldots, j_{k - 1}} \Pi_{t = 1}^{k - 1} x_{j_t},~\forall i \in [p],
\]
and define the operator norm of $\bA$ as
\[
    \|\bA\|_2 := \sup _{\|\bu\|_2 = 1, \bu \in \RR ^ p}\|\bA \underbrace{(\bu \otimes \bu \otimes \ldots \otimes \bu)}_{k-1}\|_2. 
\]

For two scalar sequences $\{a_n\}_{n \ge 1}$ and $\{b_n\}_{n \ge 1}$, we say $a_n \gtrsim b_n (a_n \lesssim b_n)$ if there exists a universal constant $C > 0$ such that $a_n \ge C b_n (a_n \le C b_n)$ for all $n \ge 1$. For any random variable $X$ valued in $\RR$ and any $r \in \NN$, define its (Orlicz) $\psi_r$-norm as
\[
    \|X\|_{\psi_r} := \inf \bigl\{k > 0: \EE \exp \{(|X| / k) ^ r\} \leq 2\bigr \}. 
\]
Similarly, for any random vector $\bx$ valued in $\RR^p$, define its (Orlicz) $\psi_r$-norm by
\[
	\|\bx\|_{\psi_r} := \sup _{\bu \in \cS ^ {p-1}} \|\bu ^ \top \bx\|_{\psi_r}, 
\]
where $\cS^{p-1}$ denotes the unit sphere in $\RR^p$. 

Define $\cB(\bbeta ^ \ast, r)$ to be the Euclidean ball of radius $r$ centered at $\bbeta ^ *$.

\section{Problem setups}
\label{sec:Problem setup}

Suppose we have $N$ independent observations $\cD := \{(\bx_i, Y_i)\}_{i=1}^{N}$ of $(\bx, Y)$ valued in $\RR ^ p \times \RR$. In matrix forms, write $\bX = (\bx_1, \ldots, \bx_N) ^ \top$ and $\by = (Y_1, \ldots, \allowbreak Y_N) ^ \top$. Denote the probability density or mass function of $\bx$ by $f_{\bx}$ and the conditional probability density or mass function of $Y$ given $\bx$ by $f_{Y | \bx}(\cdot| \bx; \bbeta^ *)$, where $\bbeta ^ * \in \RR ^ p$ parametrizes $f_{Y | \bx}$ and is of our interest. In the decentralized data setup, $\cD$ is distributed over $m$ sites, across which the communication is highly restricted. For simplicity, we assume that {$\cD$ is split uniformly at random over $m$ sites} so that each site has $n := N / m$ observations. For each $k \in [m]$, let $\cD^{(k)} := \big\{\big(\bx_1^{(k)},Y_1^{(k)}\big), \ldots, \big(\bx_n^{(k)}, Y_n^{(k)}\big)\big\}$ denote the subsample at the $k$th site. Similarly, write $\bX^{(k)} = \big(\bx_1^{(k)}, \ldots,  \bx_n^{(k)}\big)^\top$ and $\by^{(k)} = \big(Y_1^{(k)}, \ldots, Y_n^{(k)}\big)^\top$.

Our paper focuses on distributed estimation of $\bbeta ^ *$ through one-shot communication of local estimators. Under the centralized setup where the full sample $\cD$ is accessible, one often estimates $\bbeta ^ *$ by solving an empirical risk minimization problem as follows:  
\beq
    \label{eq:mle_full}
    \widehat \bbeta ^ {\mathrm{full}} \in \argmin_{\bbeta \in \cT} \frac{1}{N}\sum_{i = 1} ^ N \ell(\bbeta; (\bx_i, Y_i)). 
\eeq
Here $\cT$ is the parameter space, and $\ell: \cT \times \RR ^ p \times \RR \to \RR$ is a differentiable loss function. However, when the full data are decentralized, they are hard to access, thereby making it difficult to evaluate the global loss function above. To avoid massive data transfers, we instead communicate and aggregate local estimators $\{\widehat\bbeta ^ {(k)}\}_{k = 1} ^ m$. Formally, for any dataset $\cA$, a finite subset of $\RR ^ p \times \RR$, define $\ell_{\cA}(\bbeta) :=\frac{1}{|\cA|} \sum_{(\bx, Y) \in \cA} \ell(\bbeta; (\bx, Y))$. For any $k \in [m]$, the $k$th site computes 
\beq
    \label{eq:erm_local}
    \widehat{\bbeta} ^ {(k)} \in \underset{\bbeta \in \cT}{\argmin} ~\ell_{\cD ^ {(k)}} (\bbeta) =: \underset{\bbeta \in \cT}{\argmin} ~\ell ^ {(k)}(\bbeta). 
\eeq
Next, a central server collects all these local estimators $\{\widehat\bbeta ^ {(k)}\}_{k \in [m]}$ and aggregates them to generate an enhanced estimator, which we expect to enjoy comparable statistical accuracy as the full-sample estimator $\widehat \bbeta ^ {\mathrm{full}}$ \citep{Zhang2012Comunication, rosenblatt2016optimality}.

We consider two specific problem setups with decentralized data. The first setup is distributed estimation of the coefficients of a GLM with canonical link. There the conditional probability density function (PDF) of $Y$ given $\bx$ is defined as
\begin{equation}
\label{eq:pdf_glm}
    f_{Y|\bx}(y | \bx ; \bbeta ^ *)
    =c(y) \exp \bigg(\frac{y(\bx ^ \top \bbeta^{*})-b(\bx ^ \top \bbeta ^ *)}{\phi}\bigg)
    =c(y) \exp \bigg(\frac{y \eta-b(\eta)}{\phi}\bigg).      
\end{equation}
Here $\eta = \bx ^ {\top}\bbeta ^ *$ is the linear predictor, $\phi$ is the dispersion parameter, and $b: \RR \to \RR$ and $c: \RR \to \RR$ are known functions. Some algebra yields that $\EE (Y | \bx) = b'(\eta)$ and that $\var(Y | \bx) = \phi b''(\eta)$. The moment generating function of $Y$ is given by
\beq
    \label{eq:mgf_glm}
    \begin{aligned}
        M_Y(t) := \EE(e ^ {tY})  
        &= \int_{-\infty}^{\infty} c(y)\exp\biggl(\frac{(\eta + \phi t)y - b(\eta)}{\phi}\biggr)dy \\
        & = \int_{-\infty}^{\infty} c(y)\exp\biggl(\frac{(\eta + \phi t)y - b(\eta + \phi t) + b(\eta + \phi t) - b(\eta)}{\phi}\biggr)dy \\
        & = \exp[\phi ^ {-1}\{b(\eta + \phi t) - b(\eta)\}], 
    \end{aligned}
\eeq
where the last equation is due to the fact that for any $\eta \in \RR$, 
\[
    \int_{y = -\infty} ^ {\infty} c(y)\exp\bigg(\frac{y\eta - b(\eta)}{\phi}\bigg)dy = 1. 
\]
To estimate $\bbeta ^ *$, we choose the loss function in \eqref{eq:erm_local} to be negative log-likelihood, i.e., 
\beq
    \label{eq:nll}
    \ell(\bbeta; (\bx, Y)) = -Y\bx ^ \top \bbeta + b(\bx ^ \top \bbeta).
\eeq
Under the one-shot distributed learning framework, this means that the central server needs to aggregate $m$ local maximum likelihood estimators (MLEs) to learn $\bbeta ^ *$. 

The second problem that we consider is the noisy phase retrieval problem, which, unlike solving for the MLE under GLMs, is a non-convex problem. Basically, the problem of phase retrieval aims to recover the phase of a signal from the magnitude of its Fourier transformation. It has wide applications in X-ray crystallography, microscopy, optical imaging, etc. We refer the readers to \cite{shechtman2015phase} for more details on how the mathematical formulation of the problem is derived from the physical setting in optical imaging. Formally, consider the sensing vector $\bx$ and response $Y$ that conform to the following phase retrieval model with noise:
\begin{equation}
	\label{eq:noisy_phase_retrieval}
	Y = (\bx^\top\bbeta^ *)^2 + \varepsilon, 
\end{equation}
where $\varepsilon$ is the noise term that is independent of $\bx$. It is noteworthy that this model has an identifiability issue regarding $\bbeta ^ *$: flipping the sign of $\bbeta ^ *$ does not change the model at all! Therefore, in order to gauge the statistical error of an estimator of $\bbeta ^ *$, say $\widehat \bbeta$, we take the minimum of the distances between $\widehat \bbeta$ and $\bbeta ^ *$ and between $\widehat \bbeta$ and $-\bbeta ^ *$. Following \cite{candes2015phase}, \cite{ma2018implicit}, we choose the loss function to be the square loss in \eqref{eq:erm_local} to estimate $\bbeta ^ *$, i.e., $\ell(\bbeta; (\bx, Y)) = \{Y - (\bx ^ \top \bbeta) ^ 2\} ^ 2$. While the resulting least-squares problem is non-convex, it is solvable through a two-stage approach, exemplified by the Wirtinger Flow algorithm (Algorithm \ref{alg:wirtinger_flow}). This approach first derives a plausible initial estimator of $\bbeta ^ *$, say through spectral methods, and then refines this initial estimator through solving a local least squares problem around it. In the distributed learning context, all the sites perform this two-stage approach locally and send the least squares estimators to the central server for aggregation.

\section{Core methodology}
\label{sec:Methodology}

In this section, we introduce ReBoot, a one-shot distributed learning framework based on refitting bootstrap samples drawn from local models. The ReBoot framework consists of three main steps:
\begin{enumerate}
    \item The central server collects all the local estimators $\big\{\widehat\bbeta ^ {(k)}\big\}_{k = 1} ^ m$.
    \item For each $k \in [m]$, the central server generates a bootstrap sample $\big(\widetilde \bx ^ {(k)}, \widetilde Y ^ {(k)}\big)$ by first drawing a vector $\widetilde \bx ^ {(k)}$ from the feature distribution $f_{\bx}(\cdot)$, followed by drawing a response $\widetilde Y ^ {(k)}$ from the conditional distribution $f_{Y|\bx} \big(\cdot|\,\widetilde \bx ^ {(k)}, \widehat \bbeta ^ {(k)}\big)$. Let $\widetilde \bz ^ {(k)}$ denote the bootstrap sample $\big(\widetilde \bx ^ {(k)}, \widetilde Y ^ {(k)}\big)$.
    \item The central server aggregates all the bootstrap samples $\big\{\widetilde \bz ^ {(k)}\big\}_{k = 1} ^ m$ and forms the ReBoot loss function:
\begin{equation}
\label{eq:reboot_loss}
	\widetilde \ell (\bbeta) 
    := \frac{1}{m} \sum_{k = 1} ^ m \EE \Big\{\ell \big(\bbeta \, ;\,  \widetilde \bz ^ {(k)}\big) \big|\, \widehat\bbeta ^ {(k)}\Big\}
    = \frac{1}{m} \sum_{k = 1} ^ m \int_{\RR ^ {p + 1}} \ell \big(\bbeta \, ;\,  \widetilde \bz ^ {(k)}\big) f_{\widetilde \bz ^ {(k)}} \big(\widetilde \bz ^ {(k)}\,\big|\,  \widehat \bbeta ^ {(k)}\big) d  \widetilde \bz ^ {(k)},
\end{equation}
where $f_{\widetilde \bz ^ {(k)}} \big(\cdot|\, \widehat \bbeta ^ {(k)}\big)$ denotes the conditional distribution of $\widetilde \bz ^ {(k)}$ given $\widehat\bbeta ^ {(k)}$, and where $\EE \big(\cdot |\, \widehat\bbeta ^ {(k)}\big)$ denotes the conditional expectation given $\widehat \bbeta ^ {(k)}$. Here we denote the ReBoot loss by simply $\widetilde \ell (\bbeta)$ for notational convenience, though it is conditional on all the local estimators $\big\{\widehat\bbeta ^ {(k)}\big\}_{k = 1} ^ m$.
\end{enumerate}
By minimizing $\widetilde \ell (\bbeta)$ given the parameter space $\cT$, we derive the ReBoot estimator $\widehat \bbeta ^ {\mathrm{rb}}$ as
\begin{equation}
\label{eq:reboot_def}
    \widehat\bbeta ^ {\mathrm{rb}} 
    := \underset{\bbeta \in \cT}{\argmin} \ \widetilde \ell (\bbeta).
\end{equation}
Our framework assumes that the feature distribution $f_\bx(\cdot)$ is known when drawing the bootstrap samples $\{\widetilde \bx ^ {(k)}\}_{k = 1} ^ m$. However, this assumption may not hold in real-world circumstances. In such cases, we instead use an estimated feature distribution, denoted as $f_{\widetilde \bx}(\cdot)$, as a practical substitute for $f_\bx(\cdot)$ to generate bootstrap features. In the following sections, we will provide theoretical guarantee for relaxing this assumption and show robust empirical performances when using an estimated feature distribution.

It is worth emphasis that the methodology we present is conceptual; after all, the ReBoot loss function $\widetilde{\ell}(\bbeta)$ cannot be computed directly in practice. Nevertheless, minimizing this loss function does not require its exact form. In Section \ref{sec:Implementation}, we provide a practical implementation for the ReBoot algorithm by mini-batch stochastic gradient descent.

\section{Statistical analysis}
\label{sec:Statistical analysis}

In this section, we analyze the statistical error of the ReBoot estimator $\widehat \bbeta ^ {\mathrm{rb}}$ as defined in \eqref{eq:reboot_def} under the GLM and the noisy phase retrieval models. 

\subsection{Generalized linear models}

In this subsection, we analyze the statistical error of ReBoot under the GLM. Consider the pair $(\bx, Y)$ of the feature vector and response that satisfies the GLM \eqref{eq:pdf_glm}. We further assume the following conditions for the GLM. 

\begin{con}
\label{con:distribution}
	Suppose the feature vector $\bx$ satisfies that $\EE \bx = \bzero$ and that $\|\bx\|_{\psi_2} \le K$ with $K \geq 1$. Besides, $\lambda_{\min} \{\EE (\bx \bx^\top)\} \geq \kappa_0 > 0$.	
\end{con}

\begin{con}
\label{con:b_double_prime}
	(i) There exists $\tau: \RR_+ \to \RR_+$ such that for any $\eta \in \RR$ and any $\omega > 0$, $b'' (\eta) \geq \tau (\omega) > 0$ whenever $|\eta| \leq \omega$; (ii) $\forall \eta \in \RR$, $b''(\eta) \leq M$ with $1 \leq M < \infty$.
\end{con}

{\begin{con}
\label{con:b_four_prime}
	$\forall \eta \in \RR$, $|b''' (\eta)| \leq M$ and $|b'''' (\eta)| \leq M$. 
\end{con}}

Condition \ref{con:distribution} assumes that $\bx$ is centered and sub-Gaussian with covariance matrix positive definite. Condition \ref{con:b_double_prime} guarantees that the response is sub-Gaussian and non-degenerate when $\eta$ is bounded. To see this, by \eqref{eq:mgf_glm}, when $b''(\eta) \le M$ for any $\eta$, 
\[
    \EE\bigl[\exp\{t(Y - b'(\eta))\} | \bx\bigr] 
    = \exp\bigg(\frac{b(\eta + \phi t) - b(\eta) - \phi t b'(\eta)}{\phi}\bigg)
    \le \exp\biggl(\frac{\phi M t ^ 2}{2}\biggr), 
\]
which implies that $\|Y - b'(\eta)\|_{\psi_2} \lesssim (\phi M) ^ {1 / 2}$. Besides, $\var(Y | \bx) = \phi b''(\eta) \ge \phi \tau(\omega) > 0$ when $|\eta| \le \omega$; $Y$ is thus non-degenerate given $\bx$. In particular, in logistic regression, we can choose $\tau(\omega) = (3 + e ^ \omega) ^ {-1}$. 
Note that $\nabla ^ 2 \ell (\bbeta; (\bx, Y)) = b''(\bx ^ \top \bbeta) \bx \bx ^ \top$. Therefore, one implication of the conditions above is that $C_1 \preceq \nabla ^ 2 \ell_{\cD ^ {(k)}} (\bbeta) \preceq C_2$ for some positive constants $C_1$ and $C_2$ with high probability for any $k \in [m]$. \cite{han2016bootstrap} required similar conditions for the Hessian matrix of the log-likelihood to establish the statistical guarantee for their KL-averaging estimator. 
Finally, Condition \ref{con:b_four_prime} requires $b''' (\eta)$ and $b''''(\eta)$ to be bounded; this guarantees that the third and fourth derivatives of the empirical negative log-likelihood function enjoy fast concentration rates (see Lemmas \ref{lem:b_three_prime} and \ref{lem:b_four_prime} as well as their proof).
Similar assumptions appear in the statistical analysis of the na\"ive averaging estimator in \cite{rosenblatt2016optimality}, which required the sixth-order derivative of the loss function to be bounded when $p / n \rightarrow 0$.

Now we introduce more notation to facilitate the presentation. Recall that we choose the loss function to be the negative log-likelihood function as per \eqref{eq:nll}. The gradient and Hessian of the loss function $\ell_{\cD}(\bbeta)$ are respectively
\[
\begin{aligned}
	\nabla \ell_{\cD}(\bbeta) = \frac{1}{N} \sum_{i=1} ^ N \{b'(\bx_i ^ \top \bbeta) - Y_i\} \bx_i
	\quad \text{and}\quad
	\nabla ^ 2 \ell_{\cD}(\bbeta) = \frac{1}{N} \sum_{i=1} ^ N b''(\bx_i ^ \top \bbeta) \bx_i \bx_i ^ \top.    
\end{aligned}
\]
For simplicity, write the negative log-likelihood on the $k$th subsample $\ell_{\cD ^ {(k)}} (\bbeta)$ as $\ell ^{(k)}(\bbeta)$. 

Given that $\widehat\bbeta ^ \rb$ is derived from minimizing $\widetilde\ell(\bbeta)$, a standard approach to bound its statistical error is to first establish the local strong convexity of $\widetilde \ell(\bbeta)$ around $\bbeta ^ *$ and then bound $\ltwonorm{\nabla \widetilde \ell(\bbeta ^ *)}$ \citep{negahban2012unified, fan2018lamm, zhu2021taming}. These two ingredients are established in Proposition \ref{cor:_reboot_glr_lrsc} and Theorem \ref{thm:glm_reboot_gradient} respectively. 

We start with establishing the local strong convexity of $\widetilde \ell(\bbeta)$ around $\bbeta ^ *$. Towards this end, for any differentiable map $\ell: \RR ^ p \rightarrow \RR$, we define the first-order Taylor remainder of $\ell(\bbeta)$ at $\bbeta_0$ as
\[
	\delta \ell(\bbeta; \bbeta_0) := \ell(\bbeta) - \ell(\bbeta_0) - \nabla \ell(\bbeta_0) ^ \top (\bbeta - \bbeta_0).
\]

\begin{prop}
\label{cor:_reboot_glr_lrsc}
Let $\alpha := 2\log(64K ^ 2 / \kappa_0)$ and $\kappa := \kappa_0 \tau(K \alpha ^ {1/2} + K \alpha ^ {1/2} \ltwonorm{\bbeta ^ *}) / 4$. Under Conditions \ref{con:distribution} and \ref{con:b_double_prime}, we have
\begin{equation}
	\delta \widetilde \ell (\bbeta; \bbeta ^ *)	
	\ge \kappa \|\bbeta - \bbeta ^ *\|_2 ^ 2
\end{equation}
for any $\bbeta \in \cB(\bbeta ^ *, 1)$.
\end{prop}

Then we are establish the statistical rate of $\ltwonorm{\nabla\widetilde \ell(\bbeta ^ *)}$.

\begin{thm}
\label{thm:glm_reboot_gradient}
Suppose $n \ge C \max(\kappa_0 ^ {-2} K^ 4\alpha ^ 2 \log n, \kappa_0 ^ {-2} K^ 4 \alpha p, p ^ 2)$ where $C$ is a universal constant. Under Conditions \ref{con:distribution}, \ref{con:b_double_prime} and \ref{con:b_four_prime}, we have with probability at least $1 - (12 m + 14) n ^ {-4}$ that
\begin{equation}
\begin{aligned}
    \|\nabla \widetilde \ell(\bbeta ^ *) \|_2
	\lesssim C_{\kappa, \phi, M, K, \bSigma ^ {-1}}' \bigg(\frac{p \vee \log n}{mn}\bigg) ^ {1 / 2}
    + C_{\kappa, \phi, M, K, \bSigma ^ {-1}}'' \bigg(\frac{p \vee \log n}{n}\bigg) ^ {2}, 
\end{aligned}
\end{equation}
for some polynomial function $C_{\kappa, \phi, M, K, \bSigma ^ {-1}}'$ and $C_{\kappa, \phi, M, K, \bSigma ^ {-1}}''$ of $\kappa, \phi, M, K, \ltwonorm{\bSigma ^ {-1}}$.
\end{thm}

Simply speaking, Theorem \ref{thm:glm_reboot_gradient} shows that under appropriate assumptions, $\ltwonorm{\nabla \widetilde \ell (\bbeta ^ *)} = O_{\mathbb{P}}\big\{\big(\frac{p \vee \log n}{mn}\big) ^ {1 / 2} + \big(\frac{p \vee \log n}{n}\big) ^ 2\big\}$. The first term is a concentration term that corresponds to $\ltwonorm{\nabla\widetilde \ell(\bbeta ^ *) - \EE \nabla\widetilde \ell(\bbeta ^ *)}$, and the second term is a bias term that corresponds to $\ltwonorm{\EE \nabla\widetilde \ell(\bbeta ^ *)}$. Finally, we combine Proposition \ref{cor:_reboot_glr_lrsc} and Theorem \ref{thm:glm_reboot_gradient} to achieve the following statistical rate of the ReBoot estimator.

\begin{thm} 
\label{thm:glm_reboot_mse}
Under the same conditions as in Theorem \ref{thm:glm_reboot_gradient}, we have with probability at least $1 - (12 m + 14) n ^ {-4}$ that
\begin{equation}
    \|\widehat \bbeta ^ \rb - \bbeta^ *\|_2
	\lesssim \kappa ^ {-1} \bigg\{C_{\kappa, \phi, M, K, \bSigma ^ {-1}}' \bigg(\frac{p \vee \log n}{mn}\bigg) ^ {1 / 2}
    + C_{\kappa, \phi, M, K, \bSigma ^ {-1}}'' \bigg(\frac{p \vee \log n}{n}\bigg) ^ {2}\bigg\}, 
\end{equation}
where $C_{\kappa, \phi, M, K, \bSigma ^ {-1}}'$ and $C_{\kappa, \phi, M, K, \bSigma ^ {-1}}''$ are the same as in Theorem \ref{thm:glm_reboot_gradient}.
\end{thm}

The most salient advantage of the ReBoot estimator relative to the other one-shot distributed approaches is its sharper rate of systematic bias, i.e., $O\{(\frac{\max(p, \log n)}{n}) ^ 2\}$. Here the systematic bias refers to the error that is independent of the number of subsamples $m$; it thus persists however many subsamples we have and can be viewed as the statistical bottleneck of a distributed estimator. When the dimension $p$ is fixed, the systematic bias of the na\"ive averaging approach is well known to be $O(n ^ {-1})$ \citep{Zhang2012Comunication, battey2018distributed, rosenblatt2016optimality}. Besides, \cite{Zhang2012Comunication} proposed the SAVGM estimator based on bootstrap subsampling with systematic bias of order $O(n ^ {-3 / 2})$. Our ReBoot estimator further sharpens this rate to be $O(n ^ {-2}\log n)$. 
An important benefit of  small systematic bias is that it allows more data splits to maintain the full-sample statistical accuracy. For instance, with $p$ fixed, some algebra yields that whenever $m = O\big\{\big(\frac{N}{\log N}\big) ^ {3 / 4}\big\} $, $\ltwonorm{\widehat\bbeta ^ \mathrm{rb} - \bbeta ^ *} = O_{\mathbb{P}}\big\{\big(\frac{\log n}{N}\big) ^ {1 / 2}\big\}$, which is nearly the full-sample rate (up to a logarithmic factor). In contrast, the na\"ive averaging estimator and the SAVGM estimator require $m = O(N ^ {1 / 2})$ and $m = O(N ^ {2 / 3})$ respectively to yield the full-sample statistical rate; both requirements are more restrictive than that of ReBoot. 

Finally, we point out that the high-order bias of ReBoot can also be achieved by other distributed algorithms with, however, higher communication cost. For instance, \cite{huang2019distributed} proposed a distributed estimator that incurs two rounds of communication:  the local servers first send the local estimators to the central server for averaging, and after receiving the averaged estimator $\overline\bbeta$, they send the local gradient and Hessian matrices evaluated at $\overline\bbeta$ to the central server to perform a global Newton-Raphson step from $\overline\bbeta$. \cite{huang2019distributed} showed that the systematic bias of the resulting estimator is of order $O(n ^ {-2})$. The communication cost of this algorithm is of order $O(m(p + p ^ 2))$, while that of ReBoot is of order $O(mp)$.

We now relax the assumption that the feature distribution $f_\bx (\cdot)$ is known when we draw the bootstrap samples $\{\widetilde \bx ^ {(k)}\}_{k = 1} ^ m$. Suppose that we draw the bootstrap sample from a misspecified or estimated feature distribution $f_{\widetilde \bx}(\cdot)$, which deviates from the unknown true feature distribution $f_\bx (\cdot)$. In the following, we investigate the statistical consequence.

\begin{cor}
\label{thm:misspecified}
Suppose we draw bootstrap feature vectors from $\cN(\widetilde \bmu, \widetilde \bS)$, while the true feature vector follow $\cN(\bmu, \bS)$ with $\|\bmu - \widetilde \bmu\|_2 \leq 1$ and $\|\bS - \widetilde \bS\|_2 \leq 1$. Under the same conditions as in Theorem \ref{thm:glm_reboot_gradient}, we have with probability at least $1 - 16 (m + 1) n ^ {-4}$ that
\begin{equation}
\begin{aligned}
	\|\widehat \bbeta ^ \rb - \bbeta^ *\|_2
	\lesssim \kappa ^ {-1} \bigg\{&C_{\kappa, \phi, M, K, \bSigma ^ {-1}}' \bigg(\frac{p \vee \log n}{mn}\bigg) ^ {1 / 2} + C_{\kappa, \phi, M, K, \bSigma ^ {-1}}'' \bigg(\frac{p \vee \log n}{n}\bigg) ^ {2} \\
	&\quad\quad+ C_{\kappa, \phi, M, K, \bbeta ^ *, \bSigma ^ {-1}} \big(\|\widetilde \bS -\bS\|_2 ^ {1/2} + \|\widetilde \bmu - \bmu\|_2\big) \bigg(\frac{p \vee \log n}{n}\bigg) \bigg\}, 
\end{aligned}
\end{equation}
for some polynomial function $C_{\kappa, \phi, M, K, \bbeta ^ *, \bSigma ^ {-1}}$ of $\kappa$, $\phi$, $M$, $K$, $\|\bbeta ^ *\|_2$, $\ltwonorm{\bSigma ^ {-1}}$, and $C_{\kappa, \phi, M, K, \bSigma ^ {-1}}'$ and $C_{\kappa, \phi, M, K, \bSigma ^ {-1}}''$ are the same as in Theorem \ref{thm:glm_reboot_gradient}.
\end{cor}

Corollary \ref{thm:misspecified} demonstrates that the bias rate of ReBoot remains sharp if the true feature distribution and the one used for bootstrap are not far from each other. For instance, when the feature vector $\bx$ follows normal distribution $\cN(\bmu, \bSigma)$ with unknown mean $\bmu$ and covariance $\bSigma $, we use the average estimators $\bar \bmu$ and $\bar \bSigma$ to construct $f_{\widetilde \bx}$, where $\|\bar \bmu  - \bmu\|_2 = O_{\PP}((p/N) ^ {-1/2})$ and $\|\bar \bSigma  - \bSigma\|_2 = O_{\PP}((p/N) ^ {-1/2})$. Under the same conditions as in Theorem \ref{thm:glm_reboot_gradient}, we can still achieve that  $\ltwonorm{\widehat \bbeta ^ \rb - \bbeta^ *} = O_{\mathbb{P}}\big\{\big(\frac{p \vee \log n}{mn}\big) ^ {1 / 2} + \big(\frac{p \vee \log n}{n}\big) ^ 2\big\}$.

\subsection{Noisy phase retrieval}

\begin{figure}[t]
	\centering
	\includegraphics[scale=.36]{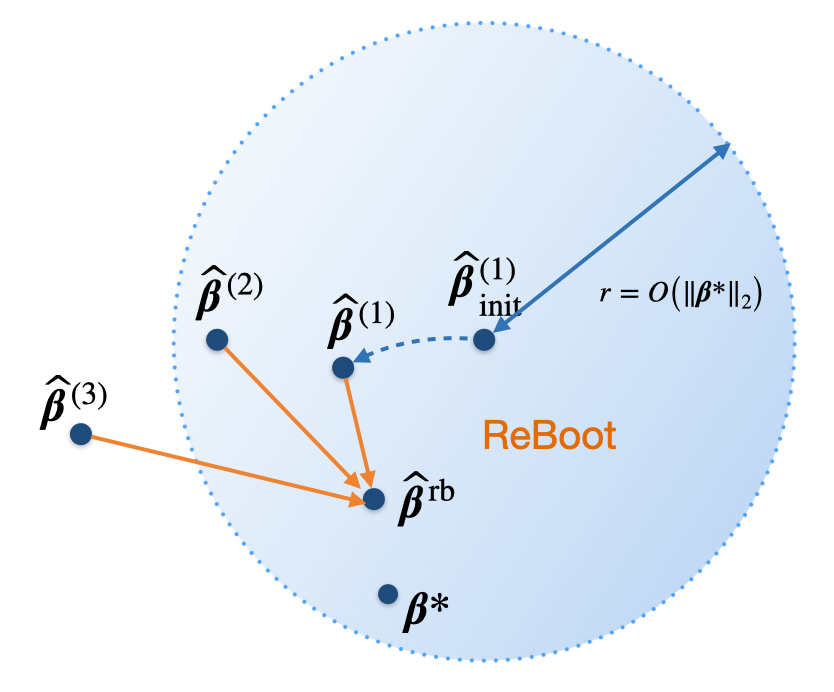}
	\caption{ The ReBoot procedure for the noisy phase retrieval problem. The blue circle represents the Euclidean ball within which we minimize the ReBoot loss $\widetilde \ell(\bbeta)$. The dashed blue arrow corresponds to the step of refining $\widehat\bbeta ^ {(1)}_{\mathrm{init}}$, implemented by gradient descent starting from $\widehat\bbeta ^ {(k)}_{\mathrm{init}}$ without any specification of the local neighborhood (see Algorithm \ref{alg:wirtinger_flow} for details). The orange arrows correspond to the ReBoot procedure.} 
	\label{fig: pr_reboot} 
\end{figure} 

In this subsection, we analyze ReBoot for the noisy phase retrieval problem. The distribution of the sensing vector $\bx$ is determined in advance. Therefore, we do not consider the misspecification of feature distribution in this section. Let $(\bx, Y)$ follow the phase retrieval model \eqref{eq:noisy_phase_retrieval} with the following condition. 
\begin{con}
\label{con:gaussian_pr}
    $\bx \sim \cN(\bzero_p, \bI_p)$ and $\varepsilon \sim \cN(0, 1)$. 
\end{con}
Condition \ref{con:gaussian_pr} assumes that the design vector $\bx$ follows the standard Gaussian distribution. The Gaussian tail of  $\varepsilon$ is imposed for technical simplicity and can be generalized to sub-Gaussian tails. As mentioned in Section \ref{sec:Problem setup}, we choose the loss function to be the square loss, that is, 
\[
\ell(\bbeta; (\bx, Y)) = \{Y - (\bx ^ \top \bbeta) ^ 2\} ^ 2.
\]

Now we introduce a two-stage approach to estimate $\bbeta ^ *$ for local sites. On each subsample $\cD ^ {(k)}$, we first use the spectral method to obtain an initial estimator $\widehat\bbeta_{\mathrm{init}} ^ {(k)}$, i.e., 
\[
    \widehat{\bbeta}_{\rm{init}} ^ {(k)} := (\widehat{\lambda} ^ {(k)} / 3)^{1 / 2} \widehat{\bv} ^ {(k)},
\]
where $\widehat{\lambda} ^ {(k)}$ and $\widehat{\bv} ^ {(k)}$ are the leading eigenvalue and eigenvector of $\frac{1}{n} \sum_{i=1}^{n} y ^ {(k)} _i \bx_i ^ {(k)} \bx_{i}^{(k)\top}$. We then refine $\widehat\bbeta_{\mathrm{init}} ^ {(k)}$ by solving a constrained least squares problem within a small Euclidean ball around $\widehat\bbeta_{\mathrm{init}} ^ {(k)}$. Specifically, let 
\begin{equation}
\label{eq:pr_local_refine}
	\widehat \bbeta ^ {(k)} := \argmin _{\bbeta \in \cB(\widehat\bbeta ^ {(k)} _ {\rm{init}}, \|\bbeta ^ *\|_2/26)} \ell ^ {(k)} (\bbeta). 
\end{equation}
Once we obtain the local estimators $\big\{\widehat \bbeta ^ {(k)}\big\}_{k \in [m]}$, we derive the ReBoot estimator $\widehat \bbeta ^ \rb$ as follows:
\begin{equation}
    \label{eq:pr_rb}
	\widehat \bbeta ^ \rb := \argmin _{\bbeta \in \cB(\widehat \bbeta ^ {(1)}, \|\bbeta ^ *\|_2/26)} \widetilde \ell(\bbeta). 
\end{equation}
Figure \ref{fig: pr_reboot} illustrates the local refinement step \eqref{eq:pr_local_refine} as well as the ReBoot aggregation step \eqref{eq:pr_rb}. Note that the local neighborhood radius $\ltwonorm{\bbeta ^ *} / 26$ in these two steps is set only to facilitate theoretical analysis and does not need to be specified in practice. 

The distance between the initial estimator $\widehat \bbeta_{\rm{init}}^{(k)}$ and the true parameter $\bbeta^*$ is shown to be reasonably small in \cite{candes2015phase} and \cite{ma2018implicit}, justifying the validity of the refinement step in \eqref{eq:pr_local_refine}. We also provide the statistical rate of $\widehat \bbeta_{\rm{init}} ^ {(k)}$ in Supplementary material.
Then we establish the root-$n$ rate of the restricted least squares estimator $\widehat \bbeta ^ {(k)}$ as follows.
\begin{prop}
\label{prop:pr2}
Under Condition \ref{con:gaussian_pr}, there exists a universal positive constant $C$ such that whenever $n \ge C\max(p ^ 2, \log n)$, for any $k \in [m]$, we have
\[
	\|\widehat \bbeta ^ {(k)} - \bbeta ^ *\| _ 2 \lesssim \|\bbeta ^ *\|_2 ^ {-1} \bigg(\frac{p \vee \log n}{n}\bigg) ^ {1 / 2}
\]
 with probability at least $1 - 42 n ^ {-2}$. Moreover, $\EE \|\widehat \bbeta ^ {(k)} - \bbeta ^ *\|_2 ^ 2 
		\lesssim \|\bbeta ^ *\| _ 2 ^ {-2} (p/n)$.
\end{prop}

Finally, we establish the statistical rate of the ReBoot estimator $\widehat \bbeta ^ \rb$.
\begin{thm}
\label{thm:reboot_phase_retrieval}
	 Suppose that $n \ge C \max(p ^ 2, \log ^ 5 n)$ for some universal positive constant $C$. Then under Condition \ref{con:gaussian_pr}, we have
\begin{equation}
    \|\widehat \bbeta ^ \rb - \bbeta ^ *\|_2
	\lesssim \|\bbeta ^ *\|_2 ^ {-1} \bigg(\frac{p \vee \log n}{mn}\bigg) ^ {1/2} 
    + \big(\|\bbeta ^ *\|_2 ^ {-1}+ \|\bbeta ^ *\|_2 ^ {-3}\big)\bigg(\frac{p \vee \log n}{n}\bigg) ,
\end{equation}
with probability at least $1 - (46 m + 6)n ^ {-2}$.
\end{thm}

Theorem \ref{thm:reboot_phase_retrieval} implies that whenever $m = O\big\{\big(\frac{N}{p \vee \log N}\big) ^ {1/2}\big\} $, $\ltwonorm{\widehat\bbeta ^ \mathrm{rb} - \bbeta ^ *} = O_{\mathbb{P}}\big\{\big(\frac{p \vee \log n}{N}\big) ^ {1/2}\big\}$, which achieves nearly the full-sample rate. The proof strategy is analogous to that for Theorem \ref{thm:glm_reboot_mse}. We first show that under Condition \ref{con:gaussian_pr}, $\ltwonorm{\nabla \widetilde \ell (\bbeta ^ *)} = O_{\PP}\big\{\big(\frac{p}{mn}\big) ^ {1 / 2} + \big(\frac{p \vee \log n}{n}\big)\big\}$. The root-$N$ rate corresponds to the concentration term $\ltwonorm{\nabla\widetilde \ell(\bbeta ^ *) - \EE \nabla\widetilde \ell(\bbeta ^ *)}$, resembling the counterpart rate in Theorem \ref{thm:glm_reboot_gradient}. The $O(\frac{p \vee \log n}{n})$ term corresponds to the bias term $\ltwonorm{\EE \nabla \widetilde \ell(\bbeta ^ *)}$. Note that this bias rate is slower than that in the GLM, which is due to the difference in the Hessian structure between the two setups. Unlike the GLM setup, the local Hessian $\nabla ^ 2 \ell ^ {(k)} (\bbeta)$ in the noisy phase retrieval problem depends on the responses $(y^{(k)}_i)_{i \in [n]}$, so that some term in the high-order decomposition of $\nabla \widetilde \ell(\bbeta ^ *)$ is not unbiased any more. The new bias then gives the rate $O\big(\frac{p \vee \log n}{n}\big)$ of $\ltwonorm{\EE \nabla \widetilde \ell(\bbeta ^ *)}$.

\section{Implementation of ReBoot}
\label{sec:Implementation}

In this section, we provide a practical implementation of the ReBoot algorithm using mini-batch stochastic gradient descent (SGD) \citep{robbins1951stochastic, bottou2010large}. At its core, the ReBoot framework first constructs the loss function $\widetilde \ell (\bbeta)$ based on the bootstrap samples drawn from all local models, and then obtains the ReBoot estimator $\widehat \bbeta ^ {\mathrm{rb}}$ by minimizing this loss function. It is noteworthy that
\[
    \frac{1}{m \widetilde n} \sum_{k = 1} ^ m \sum_{i = 1} ^ {\widetilde n} \ell \big(\bbeta \, ;\,  \big(\widetilde \bx_{i} ^ {(k)}, \widetilde Y_{i} ^ {(k)}\big)\big) 
    \rightarrow \widetilde \ell (\bbeta),
    \ \ \text{as }
    \widetilde n \rightarrow \infty,
\]
where $\widetilde \bx_i ^ {(k)}$ is drawn from $f_{\bx}(\cdot)$, and $\widetilde Y_i ^ {(k)}$ is subsequently drawn from $f_{Y|\bx} \big(\cdot|\,\widetilde \bx_i ^ {(k)}, \widehat \bbeta ^ {(k)}\big)$. In practice, we directly apply SGD to minimize the ReBoot loss, instead of generating a large bootstrap sample at once to form the ReBoot loss. Specifically, in each iteration, we compute a mini-batch stochastic gradient based on small bootstrap samples drawn from the local models, perform the gradient descent and then discard these bootstrap samples. The memory thus only needs to hold the bootstrap samples generated at the current iteration. 

\IncMargin{1em}
\begin{algorithm}[t]
    \SetAlgoNlRelativeSize{-1}
    \SetNlSty{textbf}{}{:}
    \caption{Refitting bootstrap samples (ReBoot) via mini-batch SGD}   
    \label{alg:reboot}
    
    \Indentp{-1.35em}
    \KwIn{the local estimators $\big\{\widehat{\bbeta} ^ {(k)}\big\}_{k=1} ^ m$, the initial estimator $\bbeta_0$, the iteration number $T$, the fraction $L$, the local batch size $B$, the learning rate $\mu$}
    \Indentp{1.35em}
    
    \For{$t = 0, \ldots, T - 1$}{
    	Draw a random set $\cM_t \subseteq [m]$ with $|\cM_t| = \lfloor Lm\rfloor$\;
    	\For{$k \in \cM_t$}{
    		Draw samples $\widetilde{\bx}_{i,t} ^ {(k)}$ from $f_{\bx}(\cdot), i \in [B]$\;
    		Draw responses $\widetilde{Y}_{i,t} ^ {(k)}$ from $f_{Y | \bx} \big(\cdot \big|\widetilde{\bx}_{i,t} ^ {(k)}, \widehat{\bbeta} ^ {(k)}\big), i \in [B]$\;
    		$\widetilde \cD_t ^ {(k)} \leftarrow \big\{\big(\widetilde \bx_{i,t} ^ {(k)}, \widetilde Y_{i,t} ^ {(k)}\big)\big\}_{i \in [B]}$\;
    	}
    	$\widetilde \cD_t \leftarrow \bigcup_{k \in \cM_t} \widetilde \cD_t ^ {(k)}$\;
    	$\bbeta_{t + 1} \leftarrow \bbeta_t - \mu \nabla \ell_{\widetilde \cD_t} (\bbeta_t)$\;
    }
    \Indentp{-1.35em}
    \KwOut{$\bbeta_T$}
    \Indentp{1.35em}
\end{algorithm} 
\DecMargin{1em}

In Algorithm \ref{alg:reboot}, we present the pseudocode of the ReBoot algorithm. During the $t$-th iteration, the ReBoot algorithm begins by selecting a random subset $\cM_t \subseteq [m]$ with a size of $|\cM_t| = \lfloor Lm\rfloor$ (line 2). For each $k \in \cM_t$, the central server generates a bootstrap sample, $\widetilde \cD_t ^ {(k)} = \big\{\big(\widetilde \bx_{i, t} ^ {(k)}, \widetilde Y_{i, t} ^ {(k)}\big)\big\}_{i=1} ^ B$, of size $B$, given the feature distribution $f_{\bx}$, the conditional distribution $f_{Y|\bx}$, and the local estimator $\widehat\bbeta^{(k)}$ (lines 4 - 6). Subsequently, it pools all bootstrap samples $\{\widetilde \cD_t ^ {(k)}\}_{k \in \cM_t}$ together to form a mixed bootstrap sample $\widetilde \cD_t$ (line 8). Then it updates the current estimator $\bbeta_t$ using the gradient of the loss function evaluated on $\widetilde \cD_t$ (line 9).

\section{Numerical studies}
\label{sec:Numerical studies}

\begin{figure*}[t]
    \centering
    \includegraphics[scale = 0.71]{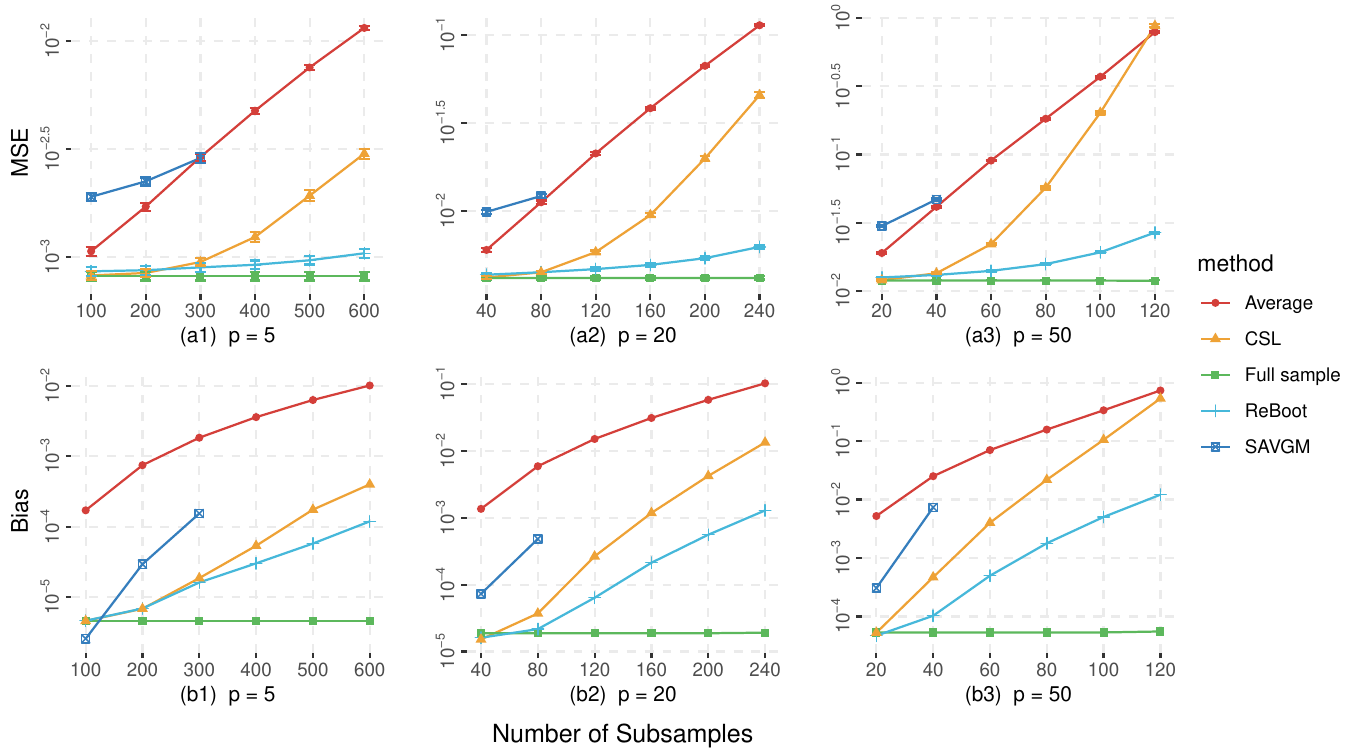}
    \caption{
    MSE (panels (a1), (a2) and (a3)) and bias (panels (b1), (b2) and (b3)) versus the subsample number $m$ under logistic regression. In panels (a1) and (b1), $p = 5$; in panels (a2) and (b2), $p = 20$; in panels (a3) and (b3), $p = 50$.}
    \label{fig:1}
\end{figure*}

In this section, we conduct simulation and real data analysis to illustrate the performance of our ReBoot method. We first compare the estimation MSE and bias of ReBoot with those of the na\"ive averaging, SAVGM \citep{Zhang2012Comunication} and CSL \citep{jordan2018communication} with one round of gradient communication under logistic regression, Poisson regression and noisy phase retrieval. Then we investigate the sensitivity of ReBoot with respect to misspecification of the design distribution under logistic regression, which is inevitable in practice. Finally, we compare the averaging method and ReBoot in terms of aggregating multiple subsample-based convolutional neural networks (CNNs) on the Fashion-MNIST \citep{xiao2017fashion} dataset. 

\subsection{MSE and bias comparison}

In this section, we focus on comparing the MSE and bias of ReBoot, averaging, SAVGM, CSL under logistic regression, Poisson regression and noisy phase retrieval. 

\subsubsection{Logistic regression}

In each Monte Carlo experiment, we first generate $N = 24,000$ independent observations $\{(\bx_i, Y_i)\}_{i \in [N]}$ of $(\bx, Y)$ satisfying that $\bx \sim \cN(\mathbf{0}_{p}, \bI_p)$ and that $Y | \bx \sim \mathrm{Ber}\big((1 + e ^ {-\bx ^ \top \bbeta ^ *}) ^ {-1}\big)$, where $\bbeta^ * = 0.2\times\bone_p$. We then divide $N$ observations into $m$ subsamples, each of which has $n = N / m$ observations, and solve for a local MLE of $\bbeta ^ *$ on each subsample. Finally, we apply ReBoot, averaging, SAVGM and CSL with one rounds of gradient communication to estimate $\bbeta ^ *$. In ReBoot, we correctly specify $f_{\bx}$ to be the probability density function (PDF) of $\cN(\bzero_p, \bI_p)$ and set $L = 1$, $B = 1$, $T = 1000$ and $\mu = 0.1$. In SAVGM, we choose the subsampling rate $r = 0.5$; a smaller $r$ causes failure of convergence of the MLE procedure, while a larger $r$ gives worse estimation accuracy. Figure \ref{fig:1} compares the MSE and bias in terms of estimating $\bbeta ^ *$ of all these methods based on $200$ independent Monte Carlo experiments with $N$ fixed and $m$ growing. We have the following observations: 
\begin{enumerate}   
    \item The MSE and bias of all the investigated approaches tend to increase as $m$ increases.
    \item When $m$ is large, ReBoot yields significantly less bias and MSE than CSL. In particular, their performance gap increases as $m$ grows. This is consistent with the high-order dependence of the MSE and bias of ReBoot on the subsample size $n$ in Theorem \ref{thm:glm_reboot_mse}.
    \item {When $p$ grows from $5$ to $50$}, the superiority of ReBoot becomes more substantial in terms of both bias and MSE.  
    \item Averaging and SAVGM are much worse than CSL and ReBoot in all the cases.
\end{enumerate}

\subsubsection{Poisson regression}

\begin{figure*}[t]
    \centering
    \includegraphics[scale = 0.71]{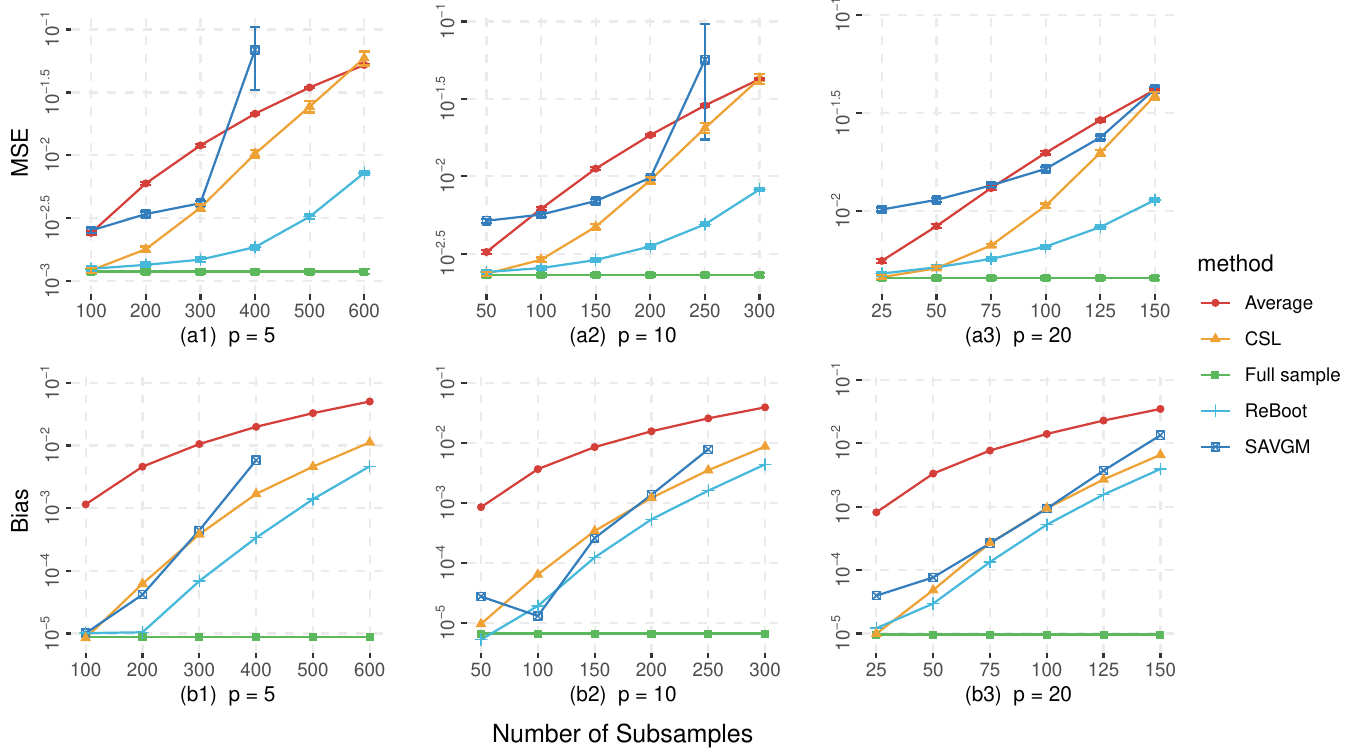}
    \caption{
    MSE (panels (a1), (a2) and (a3)) and bias (panels (b1), (b2) and (b3)) versus the subsample number $m$ under Poisson regression. In panels (a1) and (b1), $p = 5$; in panels (a2) and (b2), $p = 10$; in panels (a3) and (b3), $p = 20$.}
    \label{fig:2}
\end{figure*}

In each Monte Carlo experiment, we first generate $N = 12,000$ independent observations $\{(\bx_i, Y_i)\}_{i \in [N]}$ of $(\bx, Y)$, where the first element of $\bx$ uniformly distributed over $[0, 1]$ and the other elements uniformly distributed over $[-1, 1] ^ {p - 1}$, and $Y | \bx \sim \mathrm{Poisson}\big(e ^ {\bx ^ \top \bbeta ^ *}\big)$ where $\bbeta^ * = 0.1\times\bone_p$. Similarly to the previous section, we divide $N$ observations into $m$ subsamples, each of size $n = N / m$ observations, and compute a local MLE of $\bbeta ^ *$ on each subsample. Finally, we apply ReBoot, averaging, SAVGM and CSL to estimate $\bbeta ^ *$. In ReBoot, we correctly specify $f_{\bx}$ and set $L = 1$, $B = 1$, $T = 1000$ and $\mu = 0.01$. In SAVGM, we choose the subsampling rate $r = 0.5$, an oracular choice that minimizes MSE. Figure \ref{fig:2} compares the MSE and bias of all these approaches based on $200$ independent Monte Carlo experiments when $N$ is fixed and $m$ grows. We have essentially similar observations as in the case of logistic regression. Specifically, 
\begin{enumerate}
    \item The MSE and bias of all the approaches tend to increase as $m$ increases.
    \item {When $m$ is large (the subsample size is small)}, ReBoot yields significantly less MSE than CSL. 
    \item Averaging and SAVGM yield much higher MSE than ReBoot {in all the cases}.
    \item When $p$ grows from $5$ to $20$, the MSE of CSL and averaging grows substantially, while ReBoot maintains relatively low MSE.      
\end{enumerate}

\begin{figure*}[!t]
    \centering
    \includegraphics[scale = 0.71]{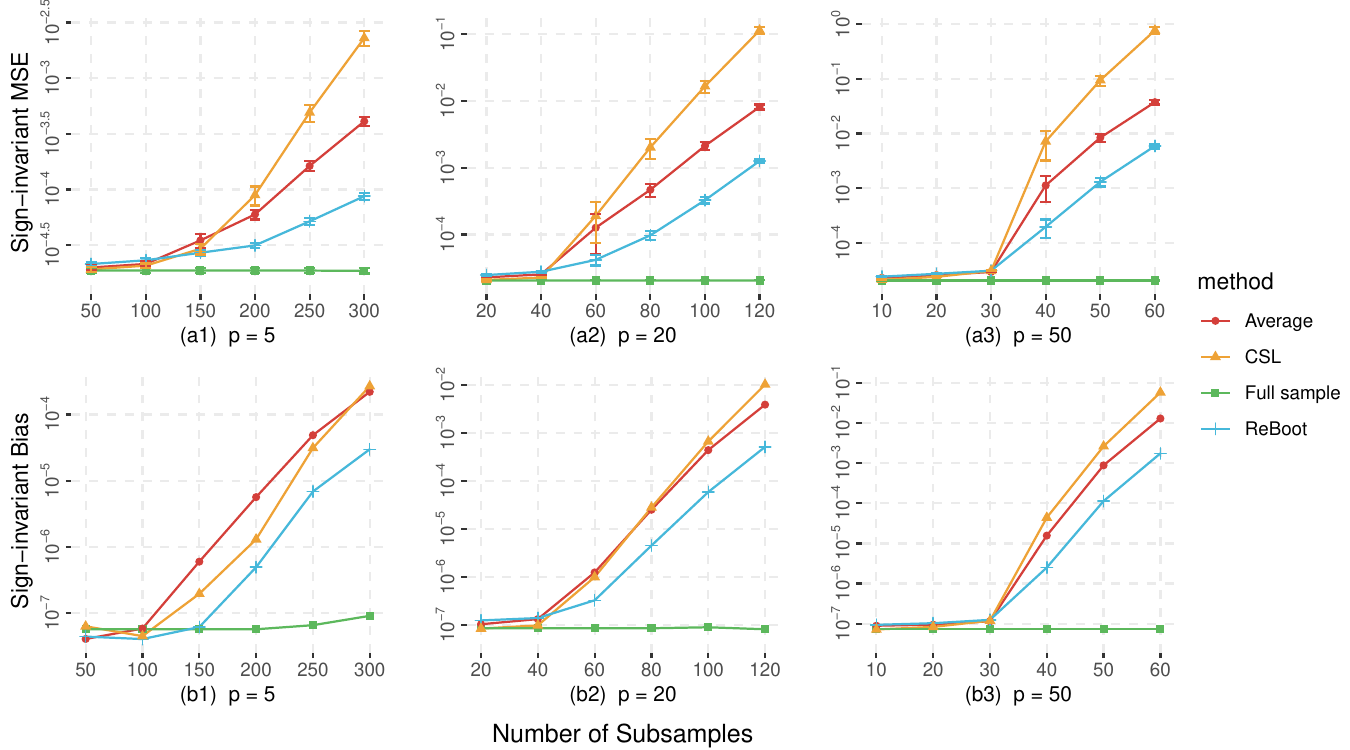}
    \caption{
    MSE (panels (a1), (a2) and (a3)) and bias (panels (b1), (b2) and (b3)) versus the subsample number $m$ under noisy phase retrieval. In panels (a1) and (b1), $p = 5$; in panels (a2) and (b2), $p = 20$; in panels (a3) and (b3), $p = 50$.}
    \label{fig:3}
\end{figure*}

\subsubsection{Noisy phase retrieval}

In our simulation, we first generate $N = 12,000$ independent observations $\{(\bx_i, Y_i)\}_{i \in [N]}$ of $(\bx, Y)$ from the noisy phase retrieval model (\ref{eq:noisy_phase_retrieval}) with $\bx \sim \cN(\mathbf{0}_{p}, \bI_p)$, $\varepsilon \sim \cN(0,1)$ and $\bbeta^ * = \bone_p$. Similarly, we uniformly split $N$ observations into $m$ subsamples, each having $n = N / m$ observations. On each subsample, we use the Wirtinger Flow algorithm (Algorithm \ref{alg:wirtinger_flow}), which is essentially a combination of spectral initialization and gradient descent, to derive a local estimator of $\bbeta ^ *$. We set $T = 10,000$, $\mu = 0.005$ for $p = 5$, $\mu = 0.001$ for $p = 20$ and $\mu = 0.0001$ for $p = 50$ in Wirtinger Flow algorithm. Finally, we apply ReBoot, averaging and CSL on these subsamples to estimate $\bbeta ^ *$. In ReBoot, we set $f_{\bx}$ to be the PDF of $\cN(\bzero_p, \bI_p)$, and set $L = 1$, $B = 1$, $T = 1000$, and $\mu = 0.01$ for $p = 5$, $\mu = 0.005$ for $p = 20$ and $\mu = 0.001$ for $p = 50$. In averaging, to ensure the sign consistency across all the local estimators for averaging, we calibrate all the local estimators to have the same sign in their first entries.
Given the identifiability issue of model \eqref{eq:noisy_phase_retrieval} due to the sign of $\bbeta ^ *$, we consider the sign-invariant versions of MSE and bias: $$\mathrm{MSE} ^ \dagger(\widehat \bbeta) := \EE \left\{\min\left(\ltwonorm{\widehat\bbeta - \bbeta ^ *} ^ 2, \ltwonorm{\widehat\bbeta + \bbeta ^ *} ^ 2\right)\right\},$$ $$\mathrm{bias} ^ \dagger(\widehat \bbeta) := \Big\|\EE \Big\{\Big(2 \times 1_{\{\ltwonorm{\widehat\bbeta  - \bbeta ^ *} < \ltwonorm{\widehat\bbeta  - \bbeta ^ *}\}} - 1\Big)\widehat\bbeta \Big\} -  \bbeta ^ *\Big\|_2.$$    
These new definitions always adjust the sign of $\widehat \bbeta$ to better align $\widehat \bbeta$ with $\bbeta ^ *$, thereby being invariant with respect to the sign of $\widehat\bbeta$. 
Figure \ref{fig:3} compares the $\mathrm{MSE} ^ {\dagger}$ and $\mathrm{bias} ^ {\dagger}$ of all these approaches based on $200$ independent Monte Carlo experiments with $N$ fixed and $m$ growing. We have the following observations: 
\begin{enumerate}
    \item ReBoot is overall the best estimator in terms of $\mathrm{MSE} ^ {\dagger}$ and $\mathrm{bias} ^ {\dagger}$ among all the investigated methods, especially when $m$ is large;
    \item The performance gap between ReBoot and averaging is substantially smaller than that in the previous GLMs. This is consistent with Theorem \ref{thm:reboot_phase_retrieval} that ReBoot does not yield as a sharp bias rate as in the GLM setup. 
\end{enumerate}

\IncMargin{1em}
\begin{algorithm}[!t]
    \SetAlgoNlRelativeSize{-1}
    \SetNlSty{textbf}{}{:}
    \caption{Wirtinger Flow Algorithm \cite{candes2013phaselift}}  
    \label{alg:wirtinger_flow}
    \Indentp{-1.35em}
    \KwIn{the dataset $\cD$, the iteration number $T$, the learning rate $\mu$}
    \Indentp{1.35em}
  	Calculate the leading eigenvalue $\widehat\lambda$ and eigenvector $\widehat \bv$ of $\bY = \frac{1}{N} \sum_{i=1}^{N} y_i \bx_i\bx_i^\top$\;
    $\widehat \bbeta _ {\rm{init}} \leftarrow (\widehat \lambda/3) ^ {1 / 2} \widehat\bv$\;
    $\bbeta_0 \leftarrow \widehat\bbeta _ {\mathrm{init}}$\;
  	\For{$t = 0,1, \ldots, T-1$}{
  	    $\bg_t \leftarrow \frac{1}{N} \sum_{i=1}^{N}\big\{\big(\bx_i^\top \bbeta_{t}\big)^2 - y_i\big\}\big(\bx_i^\top\bbeta_{t}\bigr) \bx_i$\;
    	$\bbeta_{t + 1} \leftarrow \bbeta_{t} - \mu\bg_t$\;
    }
    
    \Indentp{-1.35em}
    \KwOut{$\bbeta_T$}
    \Indentp{1.35em}
\end{algorithm} 
\DecMargin{1em}

\subsection{Misspecified or estimated feature distribution in ReBoot}

In this section, we assess the performance of ReBoot when the feature distribution $f_\bx$ is unknown and thus estimated or misspecified. In each Monte Carlo experiment, we first draw an independent sample $\{(\bx_i, Y_i)\}_{i \in [N]}$ of size $N = 12, 000$ of $(\bx, Y)$ that follows a logistic regression with autoregressive design. Specifically, we set $p = 20$, $\bbeta^ * = 0.2 \times \bone_p$, $Y | \bx \sim \mathrm{Bern}\big((1 + e ^ {-\bx ^ \top \bbeta ^ *}) ^ {-1}\big)$, $\bx \sim \cN(\mathbf{0}_{p}, \bSigma)$ in (a1), (a2), (a3) or $\bx \sim t_3(\mathbf{0}_{p}, \bSigma)$ in (b1), (b2), (b3)  with $\Sigma_{i j} = \rho ^ {|i-j|}$ for any $i, j \in [p]$. We consider $\rho \in\{0.1, 0.3, 0.5\}$ and set $L = 1$, $B = 1$, $T = 2000$ and $\mu = 0.05$ in ReBoot algorithm. We then split $N$ observations into $m$ subsamples of size $n = N/m$ observations each. To investigate the effect of error of estimating $f_\bx$, we apply two versions of ReBoot on the simulated data:
\begin{enumerate}
	\item Parametric ReBoot: The ReBoot algorithm generates bootstrap samples of the features from $\cN(\widetilde \bmu, \widetilde \bS)$, where $\widetilde \bmu$ and $ \widetilde \bS$ are the averages of locally estimated means and covariances of $\bx$.
	\item Isotropic ReBoot: The ReBoot algorithm generates bootstrap samples of the features from $\cN(\mathbf{0}_p, \bI_p)$, which misspecifies the true feature distribution.
\end{enumerate}
We also run averaging and CSL for performance comparison. Figure \ref{fig: effect} compares the MSE in terms of estimating $\bbeta ^ *$ of all these methods based on $200$ independent Monte Carlo experiments with $N$ fixed and $m$ growing.
We have the following observations: 
\begin{enumerate}
    \item In (a1), (a2), and (a3), when the feature distribution family is correctly identified, Parametric ReBoot and Isotropic ReBoot outperform both averaging and CSL. This aligns with the theoretical finding presented in Corollary \ref{thm:misspecified}.
    \item In (b1), (b2) and (b3), when the feature distribution family is misspecified, Parametric ReBoot and Isotropic ReBoot demonstrate robustness in comparison with averaging and CSL as $m$ increases.
    \item When the covariance of $\bx$ is relatively small and easily estimated, Parametric ReBoot yields smaller MSE compared to Isotropic ReBoot, as illustrated in (a1), (a2), (a3), and (b1). Conversely, Isotropic ReBoot outperforms Parametric ReBoot when dealing with a larger covariance. For instance, in (b3), where  $\cov(\bx) = 3 \bSigma$ with $\Sigma_{i j} = 0.5 ^ {|i-j|}$ for any $i, j \in [p]$, Isotropic ReBoot achieves a significantly lower MSE.
    \item The poor performances of averaging may affect the performances of CSL, as observed in (a3) and (b3) when $m = 120$ .
\end{enumerate}

\begin{figure*}[t]
    \centering
    \includegraphics[scale = 0.71]{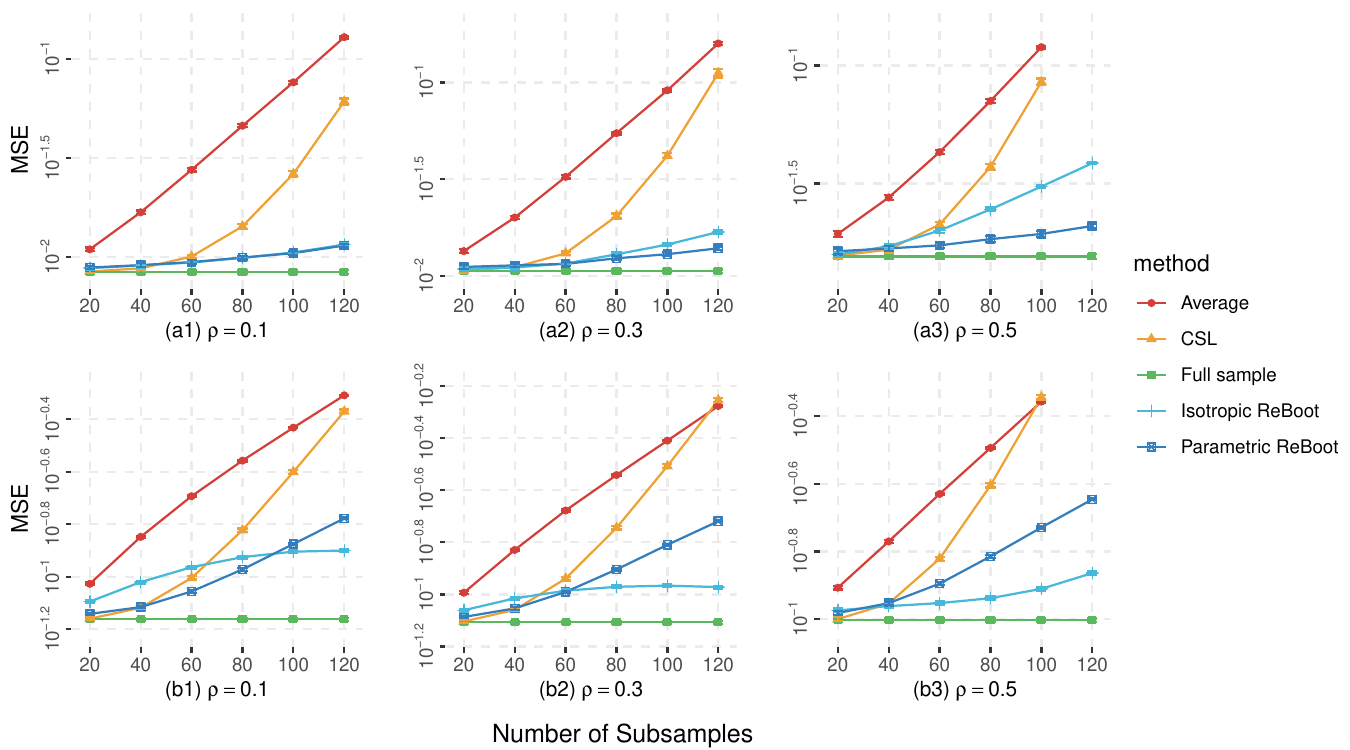}
    \caption{
    MSE versus the subsample number $m$ under logistic regression. In panels (a1), (a2) and (a3), $\bx$ is generated from $\cN(\mathbf{0}_{p}, \bSigma)$; in panels (b1), (b2) and (b3), $\bx$ is generated from $t_3(\mathbf{0}_{p}, \bSigma)$. In panels (a1) and (b1), $\rho = 0.1$; in panels (a2) and (b2), $\rho = 0.3$; in panels (a3) and (b3), $\rho = 0.5$.}
    \label{fig: effect}
\end{figure*}

\section{Real data analysis}
\label{sec:Real data}

In this section, we consider learning a convolutional neural network (CNN) for label prediction on the Fashion-MNIST dataset \citep{xiao2017fashion} in a distributed fashion. The dataset has in total 70,000 images of 10 categories of fashion articles: T-shirt, trouser, pullover, dress, coat, sandal, shirt, sneaker, bag, ankle boot. 

\begin{figure*}[h]
	\centering
	\includegraphics[scale=.48]{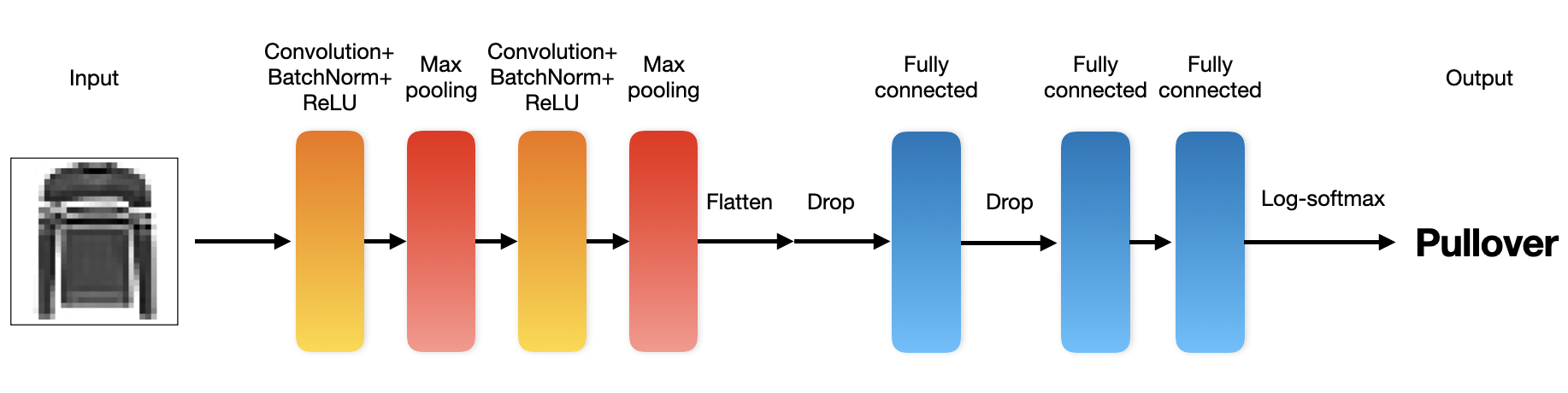}
	\caption{The CNN architecture.}
	\label{fig: cnn_procedure} 
\end{figure*}

We first split the entire dataset into the following four parts: (i) training dataset: 10,000 images; (ii) validation dataset: 5,000 images; (iii) testing dataset: 5,000 images; (iv) auxiliary dataset: 50,000 images with the labels blinded.
Here the unlabeled auxiliary data is reserved for ReBoot to retrain, the details of which are deferred to the end of this paragraph. To simulate the setup of decentralized data, we uniformly split the training dataset into $m=10$ sub-datasets $\big\{\cD_\mathrm{train} ^ {(k)}\big\}_{k \in [10]}$ of size 1,000 each. We do a similar uniform splitting of the validation dataset, yielding $\big\{\cD_{\mathrm{validation}}^ {(k)}\big\}_{k \in [10]}$. For each $k \in [10]$, we independently fit a CNN model of the same architecture (see Figure \ref{fig: cnn_procedure} for details) based on $\cD_{\mathrm{train}} ^  {(k)}$ and $\cD_{\mathrm{validation}} ^ {(k)}$, which serve as the training and validation data respectively. 
We compare two distributed learning algorithms to aggregate the ten subsample-based CNNs: na\"ive averaging and ReBoot. Na\"ive averaging means to average all the parameters across the ten CNNs respectively to construct a new CNN of the same architecture. ReBoot here needs a slight twist: Given the difficulty of modeling such images through a parameterized distribution, we directly use the unlabeled images in the auxiliary dataset to substitute for the bootstrap sample of the feature distribution at each local server. In other words, $\{\widetilde \bx_{i, t} ^{(k)}\}_{i \in [\widetilde n], t \in [T]}$ in Algorithm \ref{alg:reboot} is replaced with the auxiliary dataset for all $k \in [10]$. To summarize, ReBoot asks all the local CNNs to label the images in the auxiliary dataset and refits a CNN of the same architecture based on these labeled data. For each image, all the ten labels from the local CNNs take equal weights in the loss function of the refitting step, regardless of whether they conflict with each other or not.

\begin{table*}[h]
	\centering
	\caption{
	Accuracy(\%) of different methods on the testing dataset.}	
	\begin{tabular}{cccccc}
		\hline\hline
	             & Full-sample & Subsample (mean) & Subsample (max) & Averaging & ReBoot    \\ \hline
		Accuracy & 89.68 & 83.66     & 85.44    & 85.92 & 87.56  \\ \hline\hline
	\end{tabular}
	\label{tab: fashion-MNIST}
\end{table*}

Table \ref{tab: fashion-MNIST} reports the prediction accuracy of the full-sample-based CNN, subsample-based CNNs, averaged CNN and ReBoot CNN on the testing dataset. To characterize the overall performance of the ten subsample-based CNNs, we report their mean and maximum prediction accuracy. Table \ref{tab: fashion-MNIST} shows that averaging and ReBoot can both give CNNs that outperform the best local CNN. More importantly, the ReBoot CNN exhibits superior prediction accuracy over the averaged CNN, suggesting that ReBoot is a more powerful CNN aggregator than na\"ive averaging. 

\subsection{Federated ReBoot}
Motivated by the advantage of ReBoot over averaging in terms of aggregating local CNNs, we further propose the FedReBoot algorithm (Algorithm \ref{alg:fedreboot}), which basically replaces the averaging of model parameters in FedAvg \citep{mcmahan2017communication} with ReBoot. Specifically, FedReBoot allows the local sites to update the estimator iteratively, thus reducing the computation burden on local sites.

\IncMargin{1em}
\begin{algorithm}[h]
    \SetAlgoNlRelativeSize{-1}
    \SetNlSty{textbf}{}{:}
    \caption{Federated ReBoot (FedReBoot)}  
    \label{alg:fedreboot}
    
    \Indentp{-1.35em}
    \KwIn{the iteration number $T$, the number of local epochs $E$}
    \Indentp{1.35em}
    Initialize $\widehat\bbeta ^ {\mathrm{rb}}_0$ with some specific value; \\
    
    \For{$t = 0, \ldots, T-1$}{
        \For{$k = 1, \ldots, m$ {\bf in parallel}}{
            $\widehat\bbeta_t ^ {(k)} \leftarrow \widehat \bbeta ^ {\mathrm{rb}}_t$\;
            Server $k$ runs $E$ epochs of training locally to update $\widehat\bbeta_t ^ {(k)}$ to be $\widehat\bbeta_{t+1} ^ {(k)}$\;
        }
        Transmit these local estimators $\big\{\widehat{\bbeta}_{t+1} ^ {(k)}\big\}_{k=1}^{m}$ to the central server\;
        $\widehat\bbeta^{\mathrm{rb}}_{t+1} \leftarrow \mathrm{ReBoot}\big(\{\widehat\bbeta ^ {(k)}_{t+1}\}_{k=1}^{m}\big)$\;

    }
    \Indentp{-1.35em}
    \KwOut{$\widehat\bbeta^{\mathrm{rb}}_{T}$}
    \Indentp{1.35em}
\end{algorithm} 
\DecMargin{1em}

Similarly, we uniformly split the training dataset into $m=20$ sub-datasets $\big\{\cD_\mathrm{train} ^ {(k)}\big\}_{k \in [20]}$ of size 500 each to simulate the regime of decentralized data. In each local server, we adopt the same model architecture (see Figure \ref{fig: cnn_procedure} for details). In the ReBoot step of Algorithm \ref{alg:fedreboot}, we use the auxiliary dataset as the bootstrap sample of the features for each local server, and retrain a CNN on this dataset with labels given by all the local CNNs. Table \ref{tab: fed-learning} compares the testing accuracy of FedAvg and FedReBoot as the number of communication rounds grows. One can see that FedReBoot consistently outperforms FedAvg, especially when the number of communication rounds is small. This further demonstrates the statistical advantage of ReBoot over averaging in terms of aggregating complicated models. The performance of the two algorithms tends to match when the number of communication rounds is sufficiently large.

\begin{table}[h]
 	\centering
 	\caption{Accuracy(\%) of different methods on the testing dataset.}	
 	\begin{tabular}{ccccccc}
 	    \hline\hline
        \multicolumn{2}{c}{Communication Rounds} & 1 & 2 & 3 & 4 & 5 \\
        \hline
        \multirow{2}*{E=10} & FedAvg & 80.66 & 83.90 & 86.04 & 86.98 & 87.42  \\
                              & FedReBoot & 83.28 & 86.66 & 87.58 & 87.52 & 87.94 \\
        \multirow{2}*{E=20} & FedAvg & 81.40 & 84.92 & 86.52 & 87.32 &  87.82 \\
                              & FedReBoot & 84.02 & 86.94 & 87.84 & 88.66 & 88.28 \\                   
 	    \hline\hline
 	\end{tabular}
 	\label{tab: fed-learning}
 \end{table} 

\section{Discussions}
\label{sec:Discussions}

In this paper, we propose a general one-shot distributed learning algorithm called ReBoot to aggregate subsample-based estimators. Specifically, ReBoot draw bootstrap samples from local models and pool them to evaluate the ReBoot loss function. 
In addition, we extend ReBoot to a multi-round approach named FedReBoot for deep learning on decentralized data. We establish statistical guarantee for ReBoot under distributed GLMs and noisy phase retrieval. Numerical experiments and real data analysis confirm statistical advantages of ReBoot. In the following, we list three important questions that we wish to address in our future research: 
\begin{enumerate}
\item How should ReBoot adapt to the semi-parametric model? 

\item How should ReBoot adapt to the high-dimensional setup? 

\item How should ReBoot handle heterogeneity across subsamples? 
\end{enumerate}

\bibliographystyle{ims}
\bibliography{ref.bib}

\newpage

\appendix
\begin{center} 
    \LARGE Supplementary material for ``ReBoot: Distributed statistical learning via refitting bootstrap samples''
\end{center}

\section{Auxiliary Results for Section \ref{sec:Statistical analysis}}

\subsection{Generalized linear models}

\begin{figure*}[h]
	\centering
	\includegraphics[scale=.52]{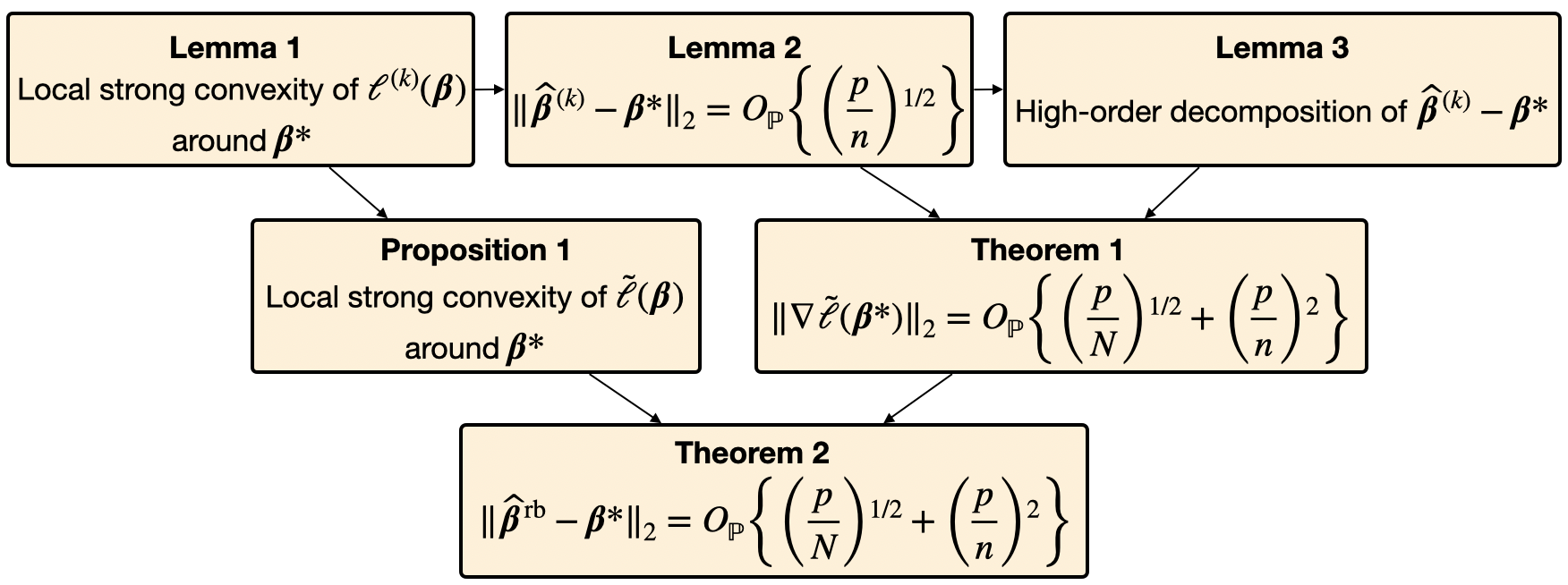}
	\caption{The proof roadmap to derive the statistical error of ReBoot.}
	\label{fig:roadmap_reboot} 
\end{figure*}

Figure \ref{fig:roadmap_reboot} presents the roadmap we follow to establish the statistical rate of $\widehat\bbeta ^ {\mathrm{rb}}$. 
Regarding Proposition \ref{cor:_reboot_glr_lrsc}, a somewhat surprising observation that underpins the proof is that for any $\bbeta \in \RR ^ p$, $\nabla ^ 2 \widetilde \ell(\bbeta)$ is independent of the local MLEs $\{\widehat\bbeta ^ {(k)}\}_{k \in [m]}$: Specifically, 
\begin{equation}
 \label{eq:hessian}
    \nabla ^ 2 \widetilde \ell(\bbeta) := \frac{1}{m} \sum_{k = 1} ^ m b''( \widetilde \bx ^ {(k)\top} \bbeta) \widetilde \bx ^ {(k)}\widetilde \bx ^ {(k)\top}.   
\end{equation}
Therefore, Proposition \ref{cor:_reboot_glr_lrsc} immediately follows Lemma \ref{lem:glr_lrsc} on the local strong convexity of $\ell ^ {(k)}(\bbeta)$, which can be obtained through standard argument. The major technical challenge lies in establishing Theorem \ref{thm:glm_reboot_gradient}: Since the bootstrap response $\widetilde Y_i ^ {(k)}$ is drawn from $f_{Y | \bx}\big(\cdot| \widetilde{\bx}_{i}^{(k)}; \widehat{\bbeta}^{(k)}\big)$ rather than $f_{Y | \bx}\big(\cdot| \widetilde{\bx}_{i}^{(k)}; {\bbeta}^*\big)$, $\EE \nabla \widetilde \ell(\bbeta ^ *)$ is not zero. We show that $\ltwonorm{\EE \nabla \widetilde \ell(\bbeta ^ *)} = O\big\{\big(\frac{p \vee \log n}{n}\big) ^ 2\big\}$, which corresponds to the machine-number-free ($m$-free) term in the rate of $\ltwonorm{\nabla \widetilde \ell(\bbeta ^ *)}$ in Theorem \ref{thm:glm_reboot_gradient}, and which characterizes the bottleneck of ReBoot that cannot be mitigated by increasing $m$. Accomplishing such a bound for $\EE \nabla \widetilde \ell(\bbeta ^ *)$ hinges on a high-order decomposition of the errors of the local estimators (Lemma \ref{prop:glm_decompn}) together with the closeness between the local estimators and the true parameter (Lemma \ref{prop:glm_local_mse}). 

Lemma \ref{lem:glr_lrsc} establishes the local strong convexity of the loss function $\ell ^{(k)}(\bbeta)$ on the $k$th subsample over $\bbeta \in \cB(\bbeta ^ \ast, r)$ with a tolerance term.
\begin{lem}
\label{lem:glr_lrsc}
Let $\alpha := 2\log(64K ^ 2 / \kappa_0)$. Under Conditions \ref{con:distribution} and \ref{con:b_double_prime}, for any $0 < r < 1$ and $t > 0$, we have with probability at least $1 - 2e ^ {-t/8}$ that
\begin{equation}
\begin{aligned}
    \delta \ell ^{(k)} (\bbeta; \bbeta ^ *) 
	&\geq \frac{\tau(K \alpha ^ {1/2} + K \alpha ^ {1/2} \ltwonorm{\bbeta ^ *})}{2}\bigg[\frac{\kappa_0}{2}\|\bbeta - \bbeta ^ *\| _ 2 ^ 2 
    - \underbrace{K^ 2 r ^ 2\bigg\{\alpha\bigg(\frac{t}{n}\bigg) ^ {1/2} + 16\bigg(\frac{2\alpha p}{n}\bigg) ^ {1 / 2}\bigg\}}_{\text{tolerance term}}\bigg],
\end{aligned}
\end{equation}
for any $\bbeta \in \cB(\bbeta ^ *, r)$, where function $\tau$ is defined in Condition \ref{con:b_double_prime}. 
\end{lem}
From the lemma above, one can see that $\ell ^ {(k)}(\bbeta)$ enjoys local strong convexity around $\bbeta ^ *$ when the tolerance term is small. To control the tolerance term, we can let local radius $r$ decay at an appropriate rate. In the proof of Lemma \ref{prop:glm_local_mse}, we apply Lemma \ref{lem:glr_lrsc} with $r$ of order $(p / n) ^ {1 / 2}$ to prevent the tolerance term from contaminating the desired statistical rate of $\widehat\bbeta ^ {(k)}$. 

As illustrated in Figure \ref{fig:roadmap_reboot}, we then derive the statistical rate of the local MLEs and then establish a high-order decomposition of their errors. We present the results in the two Lemmas below. Define $\bSigma := \EE\{b''(\bx ^ \top \bbeta ^ *)\bx\bx ^ \top\}$ and $\bTheta := \EE\{b'''(\bx ^ \top \bbeta ^ *)\bx \otimes \bx \otimes \bx\}$.  

\begin{lem}
\label{prop:glm_local_mse}
Under Conditions \ref{con:distribution} and \ref{con:b_double_prime}, there exists a universal constant $C > 0$ such that whenever $n \ge C\kappa_0 ^ {-2} K^ 4 \max(\alpha ^ 2 \log n, \alpha p)$, for any $k \in [m]$,
\[
	\PP \biggl\{\ltwonorm{\widehat\bbeta ^ {(k)} - \bbeta ^ *} \ge 2 \kappa ^ {-1} (\phi M) ^ {1 / 2} K \biggl(\frac{p \vee \log n}{n}\biggr) ^ {1 / 2}\biggr\} \le 4 n ^ {-4},
\]
where $\kappa = \kappa_0 \tau(K \alpha ^ {1/2} + K \alpha ^ {1/2} \ltwonorm{\bbeta ^ *}) / 4$ and $\alpha$ is the same as in Lemma \ref{lem:glr_lrsc}.
\end{lem}

Lemma \ref{prop:glm_local_mse} is a standard result that establishes the root-$n$ rate of the local MLE under the low-dimensional setup. Next comes the high-order decomposition of the error of the local MLE, which serves as the backbone of the analysis of ReBoot. 

\begin{lem}
\label{prop:glm_decompn}
Under Conditions \ref{con:distribution}, \ref{con:b_double_prime} and \ref{con:b_four_prime}, there exists a universal constant $C > 0$ such that whenever $n \ge C \max(\kappa_0 ^ {-2} K^ 4\alpha ^ 2 \log n, \kappa_0 ^ {-2} K^ 4 \alpha p, p ^ 2)$ with the same $\kappa$ and $\alpha$ in Lemma \ref{prop:glm_local_mse}, we have
\begin{equation}
\begin{aligned}
	\widehat\bbeta ^ {(k)} - \bbeta ^ *
    =&- \bSigma^{-1} \nabla \ell^{(k)} (\bbeta^ *) - \bSigma^{-1} \big(\nabla^2 \ell^{(k)} (\bbeta^ *) - \bSigma\big) \bSigma^{-1} \nabla \ell^{(k)} (\bbeta^ *) \\
	 &- \bSigma^{-1} \bTheta \big(\bSigma^{-1} \nabla \ell^{(k)} (\bbeta^ *)\big) \otimes \big(\bSigma^{-1} \nabla \ell^{(k)} (\bbeta^ *)\big)
		+\be, 
\end{aligned}
\end{equation}
where $\be$ satisfies with probability at least $1 - 12 n ^ {-4}$ that
\begin{equation}
	\|\be\|_2 \lesssim C_{\kappa, \phi, M, K, \bSigma ^ {-1}} \bigg(\frac{p \vee \log n}{n}\bigg) ^ {3/2} 
\end{equation}
for some polynomial function $C_{\kappa, \phi, M, K, \bSigma ^ {-1}}$ of $\kappa, \phi, M, K, \ltwonorm{\bSigma ^ {-1}}$.
\end{lem}

Similar high-order decomposition of the local MLE appears in Lemma 12 of \cite{Zhang2012Comunication}. The difference here is that we explicitly derive the dependence of the high-order error $\ltwonorm{\be}$ on $p$. 

\subsection{Noisy phase retrieval}

The gradient and Hessian of $\ell_\cD (\bbeta)$ are respectively
\begin{equation}
\label{eq:pr_grad_hessian}
\begin{aligned}
	&\nabla \ell_\cD (\bbeta) = \frac{1}{N} \sum_{i=1}^{N} \{(\bx_i ^ \top \bbeta) ^ 2 - y_i\} (\bx_i ^ \top \bbeta)\bx_i, \\
	&\nabla ^ 2 \ell_\cD (\bbeta) = \frac{1}{N} \sum_{i=1}^{N} \{3(\bx_i ^ \top \bbeta) ^ 2 - y_i\} \bx_i \bx_i ^ \top.
\end{aligned}
\end{equation}

Lemma \ref{prop:pr1} shows that the initial estimator $\widehat \bbeta_{\rm{init}} ^ {(k)}$ is reasonably close to $\bbeta ^ *$, justifying the validity of the refinement step in  \eqref{eq:pr_local_refine}. Similar results can be found in \cite{candes2015phase}, \cite{ma2018implicit}.

\begin{lem}
\label{prop:pr1}
Suppose that  $n \geq C p ^2$ for some positive constant $C$. Under Condition \ref{con:gaussian_pr}, for any $k \in [m]$, we have
\begin{equation}
	\|\widehat \bbeta_{\rm{init}} ^ {(k)} - \bbeta ^ *\|_2 \le \frac{1}{13} \|\bbeta ^ *\|_2,
\end{equation}
with probability at least $1 - 18 n ^ {-2}$.
\end{lem}

\section{Proof of technical results}
\numberwithin{lem}{section}

\subsection{Proof of Lemma \ref{lem:glr_lrsc}}
\begin{proof}
For simplicity, we omit ``${(k)}$'' in the superscript in the following proof. Define a contraction map
\[
	\phi(x; \theta) = x ^ 2 \mathbbm{1}_{\{|x| \leq \theta\}} + (x - 2 \theta) ^ 2 \mathbbm{1}_{\{\theta < x \leq 2 \theta\}} + (x + 2 \theta) ^ 2 \mathbbm{1}_{\{-2 \theta \leq x < -\theta\}}.
\]
One can verify that $\phi(x; \theta) \leq x ^ 2$ for any $\theta$. Given any $\bDelta \in \cB (\bzero, r)$, by Taylor's expansion, we can find $v \in (0,1)$ such that for any $\alpha_1, \alpha_2 > 0$, 
\begin{equation}
\label{eq:B.1_1}
	\begin{aligned}
		\delta \ell (\bbeta ^ * + \bDelta; \bbeta ^ *)
		&= \ell (\bbeta ^ * + \bDelta) - \ell (\bbeta ^ *) - \nabla \ell (\bbeta ^ *) ^ \top \bDelta 
		= \int_0 ^ 1 \frac{1}{2} \bDelta ^ \top \nabla ^ 2 \ell (\bbeta ^ * + v \bDelta) \bDelta dv\\
		&= \int_0 ^ 1 \frac{1}{2n} \sum_{i=1}^{n} b'' (\bx_i ^ \top (\bbeta ^ * + v \bDelta)) (\bDelta ^ \top \bx_i) ^ 2 dv\\
		&\geq \int_0 ^ 1 \frac{1}{2n} \sum_{i=1}^{n} b^{\prime \prime} (\bx_i ^ \top(\bbeta ^ * + v \bDelta)) \phi (\bDelta ^ \top \bx_i; \alpha_1 r) \mathbbm{1}_{\{|\bx_i  ^ \top \bbeta ^ *| \leq \alpha_2\}} dv\\
		&\geq \frac{\tau (\omega)}{2 n} \sum_{i=1}^{n} \phi(\bDelta ^ \top \bx_i; \alpha_1 r) \mathbbm{1}_{\{|\bx_i  ^ \top \bbeta ^ *| \leq \alpha_2\}}, 
	\end{aligned}
\end{equation}
where we choose $\omega = \alpha_1 + \alpha_2 > \alpha_1 r + \alpha_2$. For any $i \in [n]$, define two events $\cA_i := \{|\bDelta ^ \top \bx_i| \leq \alpha_1 r\}$ and $\cB_i := \{|\bx_i  ^ \top \bbeta ^ *| \leq \alpha_2\}$. Then we obtain that
\begin{equation}
    \label{eq:E_phi_lower}
	\begin{aligned}		
		\EE \{\phi (\bDelta ^ \top \bx_i; \alpha_1 r) \mathbbm{1}_{\cB_i}\} 
		&\geq \EE \{(\bDelta ^ \top \bx_i) ^ 2 \mathbbm{1}_{\cA_i \cap \cB_i}\} 
		= \bDelta^{\top} \EE (\bx_i \bx_i ^ \top \mathbbm{1}_{\cA_i \cap \cB_i}) \bDelta \\
		&\geq \bDelta ^ \top \EE (\bx_i \bx_i ^ \top) \bDelta - \bDelta ^ \top \EE (\bx_i \bx_i ^ \top \mathbbm{1}_{\cA_i ^ c \cup \cB_i ^ c}) \bDelta \\
		&\geq \kappa_0 \|\bDelta\| _ 2 ^ 2 - \sqrt{\EE\{(\bDelta ^ \top \bx_i) ^ 4\} \{\PP (\cA_i ^ c) + \PP (\cB_i ^ c)\}}.
	\end{aligned}
\end{equation}
Given that $\forall i \in [n], \|\bx_i\|_{\psi_2} \le K$, by Proposition 2.5.2 in \cite{vershynin2010introduction}, we have
\[
\begin{aligned}
	\PP (\cA_i ^ c) \leq \exp\bigg(-\frac{\alpha_1 ^ 2}{K ^ 2}\bigg), ~
	\PP (\cB_i ^ c) \leq \exp \bigg(- \frac{\alpha_2 ^ 2}{K ^ 2  \|\bbeta ^ *\| _ 2 ^ 2 }\bigg)
	~\text{and}~
	[\EE\{(\bDelta ^ \top \bx_i) ^ 4\}] ^ {1 / 4} \le 4K\ltwonorm{\bDelta}. 
\end{aligned}
\]
Choose $\alpha_1 = K \alpha ^ {1/2}$ and $\alpha_2 = K \alpha ^ {1/2} \ltwonorm{\bbeta ^ *}$. We can then deduce from \eqref{eq:E_phi_lower} that 
\begin{equation}
\label{eq:B.1_2}
	\EE \{\phi (\bDelta ^ \top \bx_i; \alpha_1 r) \mathbbm{1}_{\cB_i}\} 
	\geq \frac{\kappa_0}{2} \|\bDelta\| _ 2 ^ 2.
\end{equation}
Define 
\[
    Z_i := \phi(\bDelta^\top \bx_i; \alpha_1 r) \mathbbm{1}_{\cB_i} = \phi(\bDelta^\top \bx_i \mathbbm{1}_{\cB_i}; \alpha_1 r), \quad\forall i \in [n]
\]
and 
\[
    \Gamma_r := \sup _{\|\bDelta\|_2 \leq r} \bigg|n^{-1} \sum_{i=1}^{n} (Z_i - \EE Z_i)\bigg|.
\]
An application of Massart's inequality \citep{massart2000constants} yields that
\begin{equation}
\label{eq:B.1_3}
	\PP \bigg\{|\Gamma_r - \EE \Gamma_r| \geq \alpha_1 ^ 2 r ^ 2 \bigg(\frac{t}{n}\bigg) ^ {1 / 2}\bigg\} \leq 2 \exp\biggl(-\frac{t}{8}\biggr). 
\end{equation}
Next we derive the order of $\EE \Gamma_r$. Note that $|\phi (x_1; \theta) - \phi(x_2; \theta)| \leq 2 \theta |x_1 - x_2|$ for any $x_1, x_2 \in \RR$. By the symmetrization argument and then the Ledoux--Talagrand contraction principle (Theorem 4.12 in  \cite{ledoux2013probability}), for a sequence of independent  Rademacher variables $\{\gamma_i\}_{i=1}^{n}$, 
\[
\begin{aligned}
	\EE \Gamma_r 
	& \leq 2 \EE \sup _{\|\bDelta\|_2 \leq r} \bigg| \frac{1}{n} \sum_{i=1}^{n} \gamma_i Z_i \bigg| 
	\leq 8 \alpha_1 r \EE \sup _{\|\bDelta\|_2 \leq r} \bigg| \bigg\langle \frac{1}{n} \sum_{i=1}^{n} \gamma_i \bx_i \mathbbm{1}_{\cB_i},  \bDelta \bigg\rangle \bigg| \\
	& \leq 8 \alpha_1 r ^ 2 \EE \bigg\|\frac{1}{n} \sum_{i=1}^{n} \gamma_i \bx_i \mathbbm{1}_{\cB_i}\bigg\|_2 
	\leq 8 \alpha_1 r ^ 2 \bigg( \EE \bigg\|\frac{1}{n} \sum_{i=1}^{n} \gamma_i \bx_i \mathbbm{1}_{\cB_i} \bigg\| _ 2 ^ 2\bigg) ^ {1/2} \\
	& \leq 8 \alpha_1 r ^ 2 \bigg(\frac{1}{n^2} \sum_{i=1}^{n} \EE \|\bx_i\| _ 2 ^ 2 \bigg)^{1/2} 
		\le 16\sqrt 2\alpha_1 r ^ 2 K \bigg(\frac{p}{n}\bigg)^{1/2}, 
\end{aligned}
\]
where the penultimate inequality is due to the fact that $\EE(\gamma_i\gamma_j\bx^{\top}_i\bx_j\mathbbm{1}_{\cB_i \cap \cB_j}) = 0, \forall i, j \in [n], i \neq j$, and where the last inequality is due to the fact that $\|\bx_i\|_{\psi_2} \le K, \forall i \in [n]$. Combining this bound with (\ref{eq:B.1_1}), (\ref{eq:B.1_2}) and (\ref{eq:B.1_3}) yields that for any $t > 0$, with probability at least $1 - 2 e ^ {-t/8}$, for all $\bDelta \in \RR ^ p$ such that $\|\bDelta\|_2 \leq r$,
\begin{equation}
\label{eq:B.1_4}
	\delta \ell (\bbeta; \bbeta ^ *) 
	\geq \frac{\tau(\omega)}{2} \bigg\{\frac{\kappa_0}{2}\|\bDelta\| _ 2 ^ 2 - \alpha^2_1 r ^ 2\bigg(\frac{t}{n}\bigg) ^ {1/2} - 16\sqrt 2 K\alpha_1 r ^ 2 \bigg(\frac{p}{n}\bigg) ^ {1/2}\bigg\}.
\end{equation}
\end{proof}

\subsection{Proof of Proposition \ref{cor:_reboot_glr_lrsc}}
\begin{proof}
For simplicity, let $\widetilde \bz ^ {(k)} = (\widetilde \bx ^ {(k)}, \widetilde Y ^ {(k)})$. For any fixed $\bbeta$, the Hessian matrix $\nabla ^ 2 \widetilde \ell (\bbeta)$ only relies on $\bbeta$ and $\{\widetilde \bx ^ {(k)}\}_{k \in [m]}$, thus does not depend on $\{\widetilde Y ^ {(k)}\}_{k \in [m]}$ and $\{\widehat \bbeta ^ {(k)}\}_{k \in [m]}$. Consequently,
\[
\begin{aligned}
	\nabla ^ 2 \widetilde \ell (\bbeta)
	&= \frac{1}{m} \sum_{k = 1} ^ m \int_{\RR ^ {p + 1}} \nabla ^ 2 \ell \big(\bbeta \, ;\,  \widetilde \bz ^ {(k)}\big)\,  f_{\widetilde \bz ^ {(k)}}\big(\widetilde \bz ^ {(k)} \,\big|\,  \widehat \bbeta ^ {(k)}\big) d \, \widetilde \bz ^ {(k)} \\
	&= \frac{1}{m} \sum_{k = 1} ^ m \int_{\RR ^ {p + 1}} \big\{b'' \big(\widetilde \bx ^ {(k)\top} \bbeta\big) \widetilde \bx ^ {(k)} \widetilde \bx ^ {(k)\top}\big\} f_{\widetilde \bz ^ {(k)}}\big(\widetilde \bz ^ {(k)} \,\big|\,  \widehat \bbeta ^ {(k)}\big) d \, \widetilde \bz ^ {(k)} \\
	&= \frac{1}{m} \sum_{k = 1} ^ m \int_{\RR ^ p} \big\{b'' \big(\widetilde \bx ^ {(k)\top} \bbeta\big) \widetilde \bx ^ {(k)} \widetilde \bx ^ {(k)\top}\big\} f_\bx\big(\widetilde \bx ^ {(k)}\big) d \, \widetilde \bx ^ {(k)}
	= \EE \big\{b'' \big(\widetilde \bx ^ {(1)\top} \bbeta\big) \widetilde \bx ^ {(1)} \widetilde \bx ^ {(1)\top}\big\}.
\end{aligned}
\]
Given any $\bDelta \in \cB (\bzero, 1)$, by Taylor's expansion, we can find $v \in (0,1)$ such that
\[
\begin{aligned}
	\widetilde \ell (\bbeta ^ * + \bDelta) - \widetilde \ell (\bbeta ^ *) - \nabla \widetilde \ell (\bbeta ^ *) ^ \top \bDelta 
	&= \frac{1}{2} \int_0 ^ 1 \bDelta ^ \top \nabla ^ 2 \widetilde \ell (\bbeta ^ * + v \bDelta) \bDelta dv \\
	&= \frac{1}{2} \int_0 ^ 1 \EE \big\{b'' \big(\widetilde \bx ^ {(1)\top} \bbeta ^ * + v \widetilde \bx ^ {(1)\top}\bDelta\big) \big(\widetilde \bx ^ {(1)\top}\bDelta\big) ^ 2\big\} dv
\end{aligned}
\] 
Following the same proof strategy in Lemma \ref{lem:glr_lrsc}, we establish a similar lower bound that
\begin{equation}
\label{eq:reboot_glr_lrsc}
	\widetilde \ell (\bbeta ^ * + \bDelta) - \widetilde \ell (\bbeta ^ *) - \nabla \widetilde \ell (\bbeta ^ *) ^ \top \bDelta 
	\ge  \frac{\kappa_0 \tau(K \alpha ^ {1/2} + K \alpha ^ {1/2} \ltwonorm{\bbeta ^ *})}{4} \|\bDelta\|_2 ^ 2 = \kappa \|\bDelta\|_2 ^ 2,
\end{equation}
where $\kappa := \kappa_0 \tau(K \alpha ^ {1/2} + K \alpha ^ {1/2} \ltwonorm{\bbeta ^ *}) / 4$.
\end{proof}

\subsection{Proof of Lemma \ref{prop:glm_local_mse}}
\begin{proof}
For simplicity, we omit ``${(k)}$'' in the superscript in the following proof.
Construct an intermediate estimator $\widehat \bbeta _ \eta$ between $\widehat \bbeta $ and $\bbeta^ *$:
\[
	\widehat \bbeta _ \eta = \bbeta ^ * + \eta \big(\widehat \bbeta - \bbeta ^ *\big), 
\]
where $\eta=1$ if $\|\widehat \bbeta - \bbeta ^ *\|_2 \leq r$ and $\eta = r /\|\widehat \bbeta - \bbeta ^ *\|_2$ if $\|\widehat \bbeta  - \bbeta ^ *\|_2 > r$. Write $\widehat \bbeta _ \eta - \bbeta ^ *$ as $\bDelta _ \eta$. By Lemma \ref{lem:glr_lrsc}, we have with probability at least $1 - 2 e ^ {-t/8}$ that
\[
\begin{aligned}
	\frac{\tau(K \alpha ^ {1/2} + K \alpha ^ {1/2} \ltwonorm{\bbeta ^ *})}{2}\bigg[\frac{\kappa_0}{2}\|\bDelta_{\eta}\| _ 2 ^ 2 & - K^ 2 r ^ 2\bigg\{\alpha\bigg(\frac{t}{n}\bigg) ^ {1/2} + 16\bigg(\frac{2\alpha p}{n}\bigg) ^ {1/2}\bigg\}\bigg]. \\
	&\leq \delta \ell (\bbeta _ \eta; \bbeta ^ *) 
		\leq - \nabla \ell (\bbeta ^ *) ^ \top \bDelta _ \eta 
		\leq \|\nabla \ell (\bbeta ^ *)\|_2 \|\bDelta _ \eta\|_2.
\end{aligned}
\]
 Write $\kappa = \kappa_0\tau(K \alpha ^ {1/2} + K \alpha ^ {1/2} \ltwonorm{\bbeta ^ *}) / 4$. Some algebra yields that
\begin{equation}
\label{eq:grad}
	\|\bDelta _ \eta\| _ 2 
	\leq \frac{\|\nabla \ell(\bbeta ^ *)\|_2}{\kappa} + \frac{2Kr}{\sqrt \kappa_0} \bigg\{\sqrt \alpha \bigg(\frac{t}{n}\bigg) ^ {1/4} + 4\bigg(\frac{2\alpha p}{n}\bigg) ^ {1/4} \bigg\}.
\end{equation}
Now we derive the rate of $\|\nabla \ell(\bbeta ^ *)\|_2$.
\[
\begin{aligned}
	\|\nabla \ell(\bbeta ^ *)\|_2
	&= \bigg\|\frac{1}{n} \sum_{i=1}^{n} \bx_i \big\{ Y_i - b ^ \prime (\bx_i ^ \top \bbeta ^ *) \big\} \bigg\| _ 2 
	= \max_{\|\bu\|_2 = 1} \bigg\langle \frac{1}{n} \sum_{i=1}^{n} \bx_i \big\{ Y_i - b ^ \prime (\bx_i ^ \top \bbeta ^ *) \big\}, \bu \bigg\rangle \\
	&\leq 2 \max_{\|\bu\|_2 \in \cN(1/2)} \bigg\langle \frac{1}{n} \sum_{i=1}^{n} \bx_i \big\{ Y_i - b ^ \prime (\bx_i ^ \top \bbeta ^ *) \big\}, \bu \bigg\rangle.
\end{aligned}
\]
The last inequality holds because
\[
\begin{aligned}
	&\max_{\|\bu\|_2 = 1} \bigg\langle \frac{1}{n} \sum_{i=1}^{n} \bx_i \big\{ Y_i - b ^ \prime (\bx_i ^ \top \bbeta ^ *) \big\}, \bu \bigg\rangle 
	\leq \max_{\bv \in \cN(1/2)} \bigg\langle \frac{1}{n} \sum_{i=1}^{n} \bx_i \big\{ Y_i - b ^ \prime (\bx_i ^ \top \bbeta ^ *) \big\}, \bv \bigg\rangle \\
	&\quad+ \max_{\|\bu\|_2 = 1, \bv \in \cN(1/2)} \bigg\langle \frac{1}{n} \sum_{i=1}^{n} \bx_i \big\{ Y_i - b ^ \prime (\bx_i ^ \top \bbeta ^ *) \big\}, \bu-\bv \bigg\rangle \\
	&\leq \max_{\bv \in \cN(1/2)} \bigg\langle \frac{1}{n} \sum_{i=1}^{n} \bx_i \big\{ Y_i - b ^ \prime (\bx_i ^ \top \bbeta ^ *) \big\}, \bv \bigg\rangle 
	+ \frac{1}{2} \max_{\|\bu\|_2 = 1} \bigg\langle \frac{1}{n} \sum_{i=1}^{n} \bx_i \big\{ Y_i - b ^ \prime (\bx_i ^ \top \bbeta ^ *) \big\}, \bu \bigg\rangle, 
\end{aligned}
\]
where for the first step we choose $\bv \in \cN(1/2)$ that approximates $\bu$ so that $\|\bu - \bv\|_2 \leq 1/2$. 
Lemma 2.7.6 in \cite{vershynin2010introduction} delivers that
\[
	\big\|\big\{Y - b' (\bx ^ \top \bbeta ^ *) \big\} (\bx ^ \top \bu)\big\|_{\psi_1}
	\leq \|Y - b'(\bx ^ \top \bbeta ^ *)\|_{\psi_2} \|\bx_i ^ \top \bu\|_{\psi_2} 
	\lesssim (\phi M) ^ {1 / 2} K.
\]
Thus, by Bernstein's inequality, we obtain that for any $t > 0$,
\[  
    \PP\biggl(\biggl|\frac{1}{n} \sum_{i=1}^{n} \big\{ Y_i - b ^ \prime (\bx_i ^ \top \bbeta ^ *) \big\} (\bx_i ^ \top \bu)\biggr| \ge \gamma\biggr) \le 2\exp\bigg\{-c \min\bigg(\frac{n\gamma ^ 2}{\phi M K ^ 2}, \frac{n\gamma}{(\phi M) ^ {1 / 2}K} \bigg)\bigg\},
\]
where $c > 0$ is a universal constant. Then we deduce that
\[
\begin{aligned}	
	\PP (\|\nabla \ell(\bbeta ^ *)\|_2 \geq \gamma) 
	& \leq \PP \bigg(2 \max_{\bu \in \cN(1/2)} \bigg\langle \frac{1}{n} \sum_{i=1}^{n} \bx_i \big\{ Y_i - b ^ \prime (\bx_i ^ \top \bbeta ^ *) \big\}, \bu \bigg\rangle > \gamma\bigg) \\
	& \leq \sum_{\bu \in \cN(1/2)} \PP \bigg( \frac{1}{n} \sum_{i=1}^{n} \big\{ Y_i - b ^ \prime (\bx_i ^ \top \bbeta ^ *) \big\} (\bx_i ^ \top \bu)  > \gamma / 2\bigg) \\
	& \leq 2 \exp\bigg\{p \log 6 - c \min\bigg(\frac{n \gamma ^ 2}{4 \phi M K ^ 2}, \frac{n\gamma}{2(\phi M) ^ {1 / 2}K} \bigg)\bigg\}.
\end{aligned}
\]
This implies that
\beq
    \label{ineq:grad_upper_bound}
    \PP\biggl\{\ltwonorm{\nabla\ell(\bbeta ^ *)} \ge 2 (\phi M) ^ {1 / 2} K \max\biggl(\biggl(\frac{\gamma}{n}\biggr) ^ {1 / 2}, \frac{\gamma}{n}\biggr)\biggr\} \le 2e ^ {-(\gamma - p\log 6)}. 
\eeq
Let $t = 32\log n$ in \eqref{eq:grad}. When $n \ge \kappa_0 ^ {-2} K^ 4 \max(32 \times 3 ^ 4\alpha ^ 2 \log n, 2 \times 12 ^ 4 \alpha p)$, we deduce from \eqref{eq:grad} that with probability at least $1 - 2 n ^ {-4}$ that
\beq
    \label{ineq:delta_eta}
    \|\bDelta _ \eta\| _ 2 
	\leq \frac{\|\nabla \ell(\bbeta ^ *)\|_2}{\kappa} + \frac{2r}{3}.
\eeq
Choose 
\[
    r = 6\kappa ^ {-1} (\phi M) ^ {1 / 2} K \max\biggl(\biggl(\frac{\gamma}{n}\biggr) ^ {1 / 2}, \frac{\gamma}{n}\biggr). 
\]
Then by \eqref{ineq:grad_upper_bound} and \eqref{ineq:delta_eta}, with probability at least $1 - 2 n ^ {-4} - 2e ^ {-(\gamma - p\log 6)}$, we have $r > \ltwonorm{\bDelta^{(1)}_\eta}$, which further implies that $\bDelta = \bDelta_{\eta}$ according to the construction of $\bDelta_{\eta}$. Substitute $\xi = \gamma - p\log 6$ into the bound with positive $\xi$. Choose $\xi = 4 \log n$ the conclusion thus follows.
\end{proof}

\subsection{Proof of Lemma \ref{prop:glm_decompn}}
\begin{proof}
For simplicity, we omit ``${(k)}$'' in the superscript in the following proof. We write $\bDelta = \widehat \bbeta - \bbeta ^ *$ for convenience. Since $\nabla \ell (\widehat \bbeta) = 0$, by high order Taylor's expansion we have
\[
    \begin{aligned}
        0 &= \nabla \ell (\bbeta ^ *) + \nabla ^ 2 \ell (\bbeta ^ *) \bDelta + \nabla ^ 3 \ell (\bbeta ^ *) (\bDelta \otimes \bDelta) + \bR_4 (\bDelta \otimes \bDelta \otimes \bDelta) \\
        &= \nabla \ell (\bbeta ^ *) 
		 + \bSigma \bDelta
		 + (\nabla ^ 2 \ell (\bbeta ^ *) - \bSigma) \bDelta
         + \nabla ^ 3 \ell (\bbeta ^ *) (\bDelta \otimes \bDelta) 
        + \bR_4 (\bDelta \otimes \bDelta \otimes \bDelta),
    \end{aligned}
\]
where $\bR_4 = \int_0 ^ 1 \nabla ^ 4 \ell (\bbeta ^ *+ v (\widehat \bbeta - \bbeta ^ *)) dv$. Some algebra yields that
\[
\begin{aligned}
	\bDelta
    = &- \bSigma ^ {-1} \nabla \ell(\bbeta ^ *) - \bSigma ^ {-1} (\nabla ^ 2 \ell (\bbeta ^ *) - \bSigma) \bDelta 
    - \bSigma ^ {-1} \nabla ^ 3 \ell (\bbeta ^ *) (\bDelta \otimes \bDelta) 
    -\bSigma ^ {-1} \bR_4 (\bDelta \otimes \bDelta \otimes \bDelta).
\end{aligned}
\]
 Define the event 
 \[
 	\cE := \bigg\{\ltwonorm{\bDelta} \leq C_1 \kappa ^ {-1} (\phi M) ^ {1 / 2} K \bigg(\frac{p \vee \log n}{n}\bigg) ^ {1/2}\bigg\},
 \]
where $C_1$ is a constant. An application of Lemma \ref{prop:glm_local_mse} delivers that $\PP(\cE ^ c) \leq 4 n ^ {-4}$. Then we assume $\cE$ holds, and we will consider the failure probability of $\cE$ at the end of the proof.
First, we claim with probability at least $1 - 4 n ^ {-4}$ that
\begin{equation} 
\label{eq:lem1_claim1}
	\bSigma ^ {-1} (\nabla ^ 2 \ell (\bbeta ^ *) - \bSigma) \bDelta
	= - \bSigma ^ {-1} (\nabla ^ 2 \ell (\bbeta ^ *) - \bSigma) \bSigma ^ {-1} \nabla \ell(\bbeta ^ *) + \be_1,
\end{equation}
where $\be_1$ satisfies that  
\[
\begin{aligned}
	\|\be_1\|_2 
	\lesssim \{\kappa ^ {-2} \phi M ^ 3 K ^ 7  + \kappa ^ {-1} \phi ^ {1 / 2} M ^ {5 / 2} K ^ 5\} \|\bSigma ^ {-1} \|_ 2 ^ 2 \bigg(\frac{p \vee \log n}{n}\bigg) ^ {3/2}.
\end{aligned}
\]
To show this, by Taylor's expansion, we have
\[
	\begin{aligned}
		\nabla \ell (\widehat \bbeta) 
		= \nabla \ell (\bbeta ^ *) + \bR_2 \bDelta 
		= \nabla \ell (\bbeta ^ *) + \bSigma \bDelta + (\bR_2 - \bSigma)  \bDelta=0,
	\end{aligned}
\]
where $\bR_2 = \int_0 ^ 1 \nabla ^ 2 \ell(\bbeta ^ * + v (\widehat \bbeta - \bbeta ^ *)) dv$. This implies that
\begin{equation}
\label{eq:lem1_1}
	\bDelta = - \bSigma ^ {-1} \nabla \ell (\bbeta ^ *) - \bSigma ^ {-1} (\bR_2 - \bSigma) \bDelta.
\end{equation}
Then we have
\[
\begin{aligned}
	\bSigma ^ {-1} (\nabla ^ 2 \ell (\bbeta ^ *) - \bSigma) \bDelta
	= &- \bSigma ^ {-1} (\nabla ^ 2 \ell (\bbeta ^ *) - \bSigma) \bSigma ^ {-1} \nabla \ell (\bbeta ^ *) 
	- \bSigma ^ {-1} (\nabla ^ 2 \ell (\bbeta ^ *) - \bSigma) \bSigma ^ {-1} (\bR_2 - \bSigma) \bDelta.
\end{aligned}
\]
Let $\be_1 := - \bSigma ^ {-1} (\nabla ^ 2 \ell (\bbeta ^ *) - \bSigma) \bSigma ^ {-1} (\bR_2 - \bSigma) \bDelta$. By Lemma \ref{lem:center_hessian}, with probability at least $1 - 4 n ^ {-4}$, $\be_1$ satisfies that 
\[
\begin{aligned}
	\|\be_1\|_2 
	&\leq \|\bSigma ^ {-1}\big\|_ 2 ^ 2 \|\nabla ^ 2 \ell (\bbeta ^ *) - \bSigma\| _ 2 \|\bR_2 - \bSigma\|_2 \|\bDelta\| _ 2 \\
	&\lesssim \kappa ^ {-1} \phi ^ {1 / 2} M ^ {3 / 2} K ^ 3\{\kappa ^ {-1} \phi ^ {1 / 2} M ^ {3 / 2} K ^ 4  + M K ^ 2\} \|\bSigma ^ {-1} \|_ 2 ^ 2 \bigg(\frac{p \vee \log n}{n}\bigg) ^ {3/2}.
\end{aligned}
\]
Then we claim with probability at least $1 - 2 n ^ {-4}$ that
\begin{equation} 
\label{eq:lem1_claim2}
    \bSigma ^ {-1} \nabla ^ 3 \ell (\bbeta ^ *) (\bDelta \otimes \bDelta )
    = \bSigma ^ {-1} \bTheta (\bSigma ^ {-1} \nabla \ell (\bbeta ^ *)) \otimes (\bSigma ^ {-1} \nabla \ell (\bbeta ^ *)) + \be_2,
\end{equation}
where $\be_2$ satisfies that 
\[
\begin{aligned}
	\|\be_2\|_2 
	\lesssim \{\kappa ^ {-2} \phi ^ {3 / 2} M ^ {7 / 2} K ^ 9 + \kappa ^ {-1} \phi M ^ 3 K ^ 7\} \|\bSigma ^ {-1}\|_2 ^ 3 \bigg(\frac{p \vee \log n}{n}\bigg) ^ {3/2}
	+ \kappa ^ {-2} \phi M ^ 2 K ^ 5 \|\bSigma ^ {-1}\|_2 \bigg(\frac{p \vee \log n}{n}\bigg) ^ {3/2}.
\end{aligned}
\]
Some algebra gives that
\[
	\bSigma ^ {-1} \nabla ^ 3 \ell (\bbeta ^ *) (\bDelta \otimes \bDelta) 
    = \bSigma ^ {-1} \bTheta (\bDelta \otimes \bDelta) 
      + \bSigma ^ {-1} (\nabla ^ 3 \ell (\bbeta ^ *) - \bTheta) (\bDelta \otimes \bDelta).
\]
Combining Lemmas \ref{lem:two_delta} and \ref{lem:b_three_prime}, we reach the desired result.
Next, let $\be_3 := \bR_4(\bDelta \otimes \bDelta \otimes \bDelta)$. Lemma \ref{lem:b_four_prime} yields that
\begin{equation}
\label{eq:lem1_claim3}
	\begin{aligned}
		\|\be_3\|_2 \leq \|\bR_4\|_2 \|\bDelta\| _ 2 ^ 3
	 	\lesssim \kappa ^ {-3} \phi ^ {3 / 2} M ^ {5 / 2} K ^ 7 \bigg(\frac{p \vee \log n}{n}\bigg) ^ {3/2},
	\end{aligned}
\end{equation}
with probability at least $1 - 2 n ^ {-4}$. 
Combining the bounds in (\ref{eq:lem1_claim1}), (\ref{eq:lem1_claim2}) and (\ref{eq:lem1_claim3}) and the failure probability of $\cE$, we find that with probability at least $1 - 12 n ^ {-4}$ that
\[
\begin{aligned}
	\bDelta =& - \bSigma ^ {-1} \nabla \ell (\bbeta ^ *) - \bSigma ^ {-1} (\nabla ^ 2 \ell (\bbeta ^ *) - \bSigma) \bSigma ^ {-1} \nabla \ell (\bbeta^ *) 
	- \bSigma ^ {-1} \bTheta (\bSigma ^ {-1} \nabla \ell  (\bbeta ^ *)) \otimes (\bSigma ^ {-1} \nabla \ell (\bbeta ^ *)) + \be,
\end{aligned}
\]
where $\be := \be_1 + \be_2 + \be_3$ satisfying that 
\[
	\|\be\|_2 \leq \|\be_1\|_2 + \|\be_2\|_2 + \|\be_3\|_2 
	\lesssim C_{\kappa, \phi, M, K, \bSigma ^ {-1}} \bigg(\frac{p \vee \log n}{n}\bigg) ^ {3/2},
\]
with 
\[
\begin{aligned}
	C_{\kappa, \phi, M, K, \bSigma ^ {-1}} 
	&:= \kappa ^ {-3} \phi ^ {3 / 2}  M ^ {5 / 2} K ^ 7 
		+ \kappa ^ {-2} \phi M ^ 2 K ^ 5 \|\bSigma ^ {-1}\|_2 \\
	&\quad + \{\kappa ^ {-2} \phi M ^ 3 K ^ 7  + \kappa ^ {-1} \phi ^ {1 / 2} M ^ {5 / 2} K ^ 5\} \|\bSigma ^ {-1} \|_ 2 ^ 2 \\
	&\quad + \{\kappa ^ {-2} \phi ^ {3 / 2} M ^ {7 / 2} K ^ 9 + \kappa ^ {-1} \phi M ^ 3 K ^ 7\} \|\bSigma ^ {-1}\|_2 ^ 3.
\end{aligned}
\]

\end{proof}
\subsection{Proof of Theorem \ref{thm:glm_reboot_gradient}}

\begin{proof}
For simplicity, let $\widetilde \bz ^ {(k)} = (\widetilde \bx ^ {(k)}, \widetilde Y ^ {(k)})$. Applying the fourth-order Taylor expansion of $\nabla \ell (\widehat \bbeta ^ {(k)}; \widetilde \bz ^ {(k)})$ at $\bbeta ^ *$ yields that 
\[ 
	\begin{aligned}
		\nabla \ell (\widehat \bbeta ^ {(k)}; \widetilde \bz ^ {(k)})
		&= \nabla \ell (\bbeta ^ *; \widetilde \bz ^ {(k)}) 
		+ \nabla ^ 2 \ell (\bbeta ^ *; \widetilde \bz ^ {(k)}) \bDelta ^ {(k)}
		+ \nabla ^ 3 \ell (\bbeta ^ *; \widetilde \bz ^ {(k)}) (\bDelta ^ {(k)} \otimes \bDelta ^ {(k)}), \\
		&\quad+ \bigg\{\int_0 ^ 1 \nabla ^ 4 \ell (\bbeta ^ * + v (\widehat \bbeta ^ {(k)} - \bbeta ^ *), \widetilde \bz ^ {(k)}) dv\bigg\} (\bDelta ^ {(k)} \otimes \bDelta ^ {(k)} \otimes \bDelta ^ {(k)}). 
	\end{aligned}
\]
Recall that $\EE \big(\cdot \big|\, \widehat \bbeta ^ {(k)}\big)$ denote the conditional expectation given $\widehat \bbeta ^ {(k)}$. We have
\begin{equation} 
\label{eq:proof1_subtit1}
	\begin{aligned}
		0 &= \EE \big\{\nabla \ell (\widehat \bbeta ^ {(k)} ;\widetilde \bz ^ {(k)}) \big|\, \widehat \bbeta ^ {(k)}\big\} \\
		  &= \EE \big\{\nabla \ell (\bbeta ^ *; \widetilde \bz ^ {(k)}) \big|\, \widehat \bbeta ^ {(k)}\big\}
		  + \bSigma \bDelta ^ {(k)}
		  + \bTheta (\bDelta ^ {(k)} \otimes \bDelta ^ {(k)})+ \widetilde \bR_4 ^ {(k)} (\bDelta ^ {(k)} \otimes \bDelta ^ {(k)} \otimes \bDelta ^ {(k)}),
	\end{aligned}
\end{equation}
where $\widetilde \bR_4 ^ {(k)} = \EE \big\{\int_0 ^ 1 \nabla ^ 4 \ell (\bbeta ^ * + v (\widehat \bbeta ^ {(k)} - \bbeta ^ *), \widetilde \bz ^ {(k)}) dv |\, \widehat \bbeta ^ {(k)}\big\}$. Similarly, by the fourth-order Taylor expansion and the fact that $\nabla \ell ^ {(k)} (\widehat \bbeta ^ {(k)}) = 0$ , we have
\[
\begin{aligned}
	\nabla \ell ^ {(k)} (\widehat \bbeta ^ {(k)})
	&= \nabla \ell ^ {(k)} (\bbeta ^ *) + \bSigma \bDelta ^ {(k)} + \bTheta (\bDelta ^ {(k)} \otimes \bDelta ^ {(k)}) + (\nabla ^ 2 \ell ^ {(k)} (\bbeta ^ *) - \bSigma) \bDelta ^ {(k)} \\
	&\quad + (\nabla ^ 3 \ell^{(k)} (\bbeta ^ *) - \bTheta) (\bDelta ^ {(k)} \otimes \bDelta ^ {(k)})
  		+ \bR_4 ^ {(k)} (\bDelta ^ {(k)} \otimes \bDelta ^ {(k)} \otimes \bDelta ^ {(k)})
  	= 0, 
\end{aligned}
\]
where $\bR_4 ^ {(k)} := \int_0 ^ 1 \nabla ^ 4 \ell ^ {(k)} \bigl(\bbeta ^ * + v (\widehat \bbeta ^ {(k)} - \bbeta ^ *)\bigr) dv$. This can be rearranged as 
\beq
    \label{eq:substi_2}
    \begin{aligned}
        \bSigma \bDelta ^ {(k)} + \bTheta (\bDelta ^ {(k)} \otimes \bDelta ^ {(k)}) 
        &= - \nabla \ell ^ {(k)} (\bbeta ^ *) - (\nabla ^ 2 \ell ^ {(k)} (\bbeta ^ *) - \bSigma) \bDelta ^ {(k)} - (\nabla ^ 3 \ell ^ {(k)} (\bbeta ^ *) - \bTheta)  (\bDelta ^ {(k)} \otimes \bDelta ^ {(k)}) \\
        &\quad - \bR_4 ^ {(k)} (\bDelta ^ {(k)} \otimes \bDelta ^ {(k)} \otimes \bDelta ^ {(k)}).
    \end{aligned}
\eeq
Substituting \eqref{eq:substi_2} into \eqref{eq:proof1_subtit1}, we have
\[
	\begin{aligned}
		0
		&= \EE \big\{\nabla \ell (\bbeta ^ *; \widetilde \bz ^ {(k)}) \big|\, \widehat \bbeta ^ {(k)}\big\} 
		- \nabla \ell ^ {(k)} (\bbeta ^ *)  
		- (\nabla ^ 2 \ell ^ {(k)} (\bbeta ^ *) - \bSigma) \bDelta ^ {(k)}
		- (\nabla ^ 3 \ell ^ {(k)} (\bbeta ^ *) - \bTheta)  (\bDelta ^ {(k)} \otimes \bDelta ^ {(k)}) \\
		&\quad- \bR_4 ^ {(k)} (\bDelta ^ {(k)} \otimes \bDelta ^ {(k)} \otimes \bDelta ^ {(k)})
		+ \widetilde \bR_4 ^ {(k)} (\bDelta ^ {(k)} \otimes \bDelta ^ {(k)} \otimes \bDelta ^ {(k)}).
	\end{aligned}
\]
Note that $\nabla \widetilde \ell(\bbeta ^ *) = \frac{1}{m} \sum_{k=1} ^ m \EE \big\{\nabla \ell (\bbeta ^ *; \widetilde \bz ^ {(k)}) \big|\, \widehat \bbeta ^ {(k)}\big\}$. Then we have the following decomposition: 
\[
\begin{aligned}
	\nabla \widetilde \ell(\bbeta ^ *) 
	= &\underbrace{\nabla \ell (\bbeta ^ *)}_{T_1}
	+ \underbrace{\frac{1}{m} \sum_{k=1} ^ m (\nabla ^ 2 \ell ^ {(k)} (\bbeta ^ *) - \bSigma) \bDelta ^ {(k)}}_{T_2} 
	+ \underbrace{\frac{1}{m} \sum_{k=1} ^ m (\nabla ^ 3 \ell ^ {(k)} (\bbeta ^ *) - \bTheta) (\bDelta ^ {(k)} \otimes \bDelta ^ {(k)})}_{T_3} \\
	&+ \underbrace{\frac{1}{m} \sum_{k=1} ^ m \bR_4 ^ {(k)} (\bDelta ^ {(k)} \otimes \bDelta ^ {(k)} \otimes \bDelta ^ {(k)})}_{T_4} 
	- \underbrace{\frac{1}{m} \sum_{k=1} ^ m \widetilde \bR_4 ^ {(k)} (\bDelta ^ {(k)} \otimes \bDelta ^ {(k)} \otimes \bDelta ^ {(k)})}_{T_5}.
\end{aligned}	
\]
For simplicity, we define $\Upsilon_1 := \{(p \vee \log n) / n\} ^ {1 / 2}$. To study the appropriate threshold, we introduce the following events:
\[
\begin{aligned}
	\cE ^ {(k)} := 
		&\big\{\|\nabla \ell ^ {(k)} (\bbeta ^ *)\|_2
			\leq C_1 (\phi M) ^ {1 / 2} K \Upsilon_1\big\}\cap \\
		&\big\{\|\nabla ^ 2 \ell ^ {(k)} (\bbeta ^ *) - \bSigma\|_2 
			\leq C_2 M K ^ 2 \Upsilon_1\big\}\cap \\
		&\big\{\|\nabla ^ 3 \ell ^ {(k)} (\bbeta ^ *) - \bTheta\|_2
			\leq C_3 M K ^ 3 \Upsilon_1\big\} \cap \\
		&\big\{\|\bR_4 ^ {(k)}\|_2 \leq C_4 M K ^ 4 \big\}, \\
	\cF^{(k)} := 
		&\big\{\ltwonorm{\bDelta ^ {(k)}} \leq C_5 \kappa ^ {-1} (\phi M) ^ {1 / 2} K \Upsilon_1\big\}
\end{aligned}
\]
where $C_1, C_2, \ldots, C_5$ are constants. Define the intersection of all the above events by $\cA$. By Lemma \ref{prop:glm_local_mse}, Lemmas \ref{lem:center_Sigma}, \ref{lem:b_three_prime}, \ref{lem:b_four_prime} and \ref{lem:center_hessian}, we have $\PP (\cA^c) \leq 12 m n ^ {-4}$. We categorize these terms into variance and bias terms: $T_1$ is the variance term, and $T_2, T_3, T_4, T_5$ are the bias terms. Then we work on the bounds for $\{T_i\}_{i=1} ^ 5$ conditional on event $\cA$.

\paragraph*{\bf \underline{Variance terms}}

\paragraph*{Bound for $T_1$}
By \eqref{ineq:grad_upper_bound}, we have
\begin{equation}
\label{ineq:grad_upper_bound_full}
	\PP\biggl\{\ltwonorm{\nabla\ell(\bbeta ^ *)} \ge 2(\phi M) ^ {1/ 2}K \max\biggl(\bigg(\frac{t p}{mn}\bigg) ^ {1 / 2}, \frac{t p}{mn}\bigg)\bigg\} 
    \le 2e ^ {-(t - \log 6)p}. 
\end{equation}
Choose $t$ such that $(t - \log 6)p = 4 \log n$. This implies with probability at least $1 - 2 n ^ {-4}$ that
\[
	\ltwonorm{T_1} 
	\lesssim (\phi M) ^ {1 / 2} K \bigg(\frac{p \vee \log n}{mn}\bigg) ^ {1 / 2}.
\]

\paragraph*{\bf \underline{Bias terms}}

\paragraph*{Bound for $T_2$}

By Lemma \ref{prop:glm_decompn}, we have the decomposition:
\[
\begin{aligned}
	\frac{1}{m} \sum_{k=1} ^ m (\nabla ^ 2 \ell ^ {(k)} (\bbeta ^ *) - \bSigma) \bDelta ^ {(k)} 
	&= -\underbrace{\frac{1}{m} \sum_{k=1} ^ m (\nabla ^ 2 \ell ^ {(k)} (\bbeta ^ *) - \bSigma) \bSigma^{-1} \nabla \ell ^ {(k)} (\bbeta^ *)}_{T_{21}} \\
	&\quad- \underbrace{\frac{1}{m} \sum_{k=1} ^ m (\nabla ^ 2 \ell ^ {(k)} (\bbeta ^ *) - \bSigma) \bSigma^{-1} (\nabla ^ 2 \ell ^ {(k)} (\bbeta^ *) - \bSigma) \bSigma^{-1} \nabla \ell ^ {(k)} (\bbeta^ *)}_{T_{22}}  \\
	&\quad- \underbrace{\frac{1}{m} \sum_{k=1} ^ m \EE \Big\{(\nabla ^ 2 \ell ^ {(k)} (\bbeta ^ *) - \bSigma) \bSigma^{-1} \bTheta (\bSigma^{-1} \nabla \ell ^ {(k)} (\bbeta^ *)) \otimes (\bSigma ^ {-1} \nabla \ell ^ {(k)} (\bbeta^ *))\Big\}}_{T_{23}} \\
	&\quad - \frac{1}{m} \sum_{k=1} ^ m \Big[(\nabla ^ 2 \ell ^ {(k)} (\bbeta ^ *) - \bSigma) \bSigma^{-1} \bTheta (\bSigma^{-1} \nabla \ell ^ {(k)} (\bbeta^ *)) \otimes (\bSigma ^ {-1} \nabla \ell ^ {(k)} (\bbeta^ *)) \\
		&\quad\quad \underbrace{\quad\quad\quad - \EE \Big\{(\nabla ^ 2 \ell ^ {(k)} (\bbeta ^ *) - \bSigma) \bSigma^{-1} \bTheta (\bSigma^{-1} \nabla \ell ^ {(k)} (\bbeta^ *)) \otimes (\bSigma ^ {-1} \nabla \ell ^ {(k)}(\bbeta^ *))\Big\}\Big]}_{T_{24}} \\
	&\quad- \underbrace{\frac{1}{m} \sum_{k=1} ^ m (\nabla ^ 2 \ell ^ {(k)} (\bbeta ^ *) - \bSigma) \be ^ {(k)}}_{T_{25}}.
\end{aligned}	
\]

\noindent \textit{\underline{Bounds for $T_{21}$ and $T_{22}$}.}
Recall that $\EE \{\nabla \ell ^ {(k)} (\bbeta ^ *)\} = 0$. This implies that $T_{21}$ and $T_{22}$ are centered. For any $\bu \in \cS ^ {p - 1}$, we have
\[
\begin{aligned}
	\big|\bu ^ \top (\nabla ^ 2 \ell ^ {(k)} (\bbeta ^ *) - \bSigma) \bSigma^{-1} \nabla \ell ^ {(k)} (\bbeta^ *)\big| 
	&\leq \big\|\bSigma^{-1}\big\|_2 \big\|\nabla ^ 2 \ell ^ {(k)} (\bbeta ^ *) - \bSigma\big\|_2 \big\|\nabla \ell ^ {(k)} (\bbeta^ *)\big\|_2 \\
	&\lesssim \phi ^ {1 / 2} M ^ {3 / 2} K ^ 3  \|\bSigma^{-1}\|_2 \bigg(\frac{p \vee \log n}{n}\bigg).
\end{aligned}
\]
By Hoeffding's inequality, we have with probability at least $1 - 2 n ^ {-4}$ that
\begin{equation}
\begin{aligned}
\label{eq:bound_T31}
	\|T_{21}\|_2 &= \bigg\|\frac{1}{m} \sum_{k=1} ^ m (\nabla ^ 2 \ell ^ {(k)} (\bbeta ^ *) - \bSigma) \bSigma^{-1} \nabla \ell ^ {(k)} (\bbeta^ *)\bigg\|_2 
	\leq 2 \max_{\bu \in \cN(1/2)} \bigg|\frac{1}{m} \sum_{k=1} ^ m \bu ^ \top (\nabla ^ 2 \ell ^ {(k)} (\bbeta ^ *) - \bSigma) \bSigma^{-1} \nabla \ell ^ {(k)} (\bbeta^ *)\bigg| \\
	&\lesssim \phi ^ {1 / 2} M ^ {3 / 2} K ^ 3  \|\bSigma^{-1}\|_2 \bigg(\frac{p \vee \log n}{n}\bigg) \bigg(\frac{p \vee \log n}{m}\bigg) ^ {1/2} 
	\lesssim \phi ^ {1 / 2} M ^ {3 / 2} K ^ 3  \|\bSigma^{-1}\|_2 \bigg(\frac{p \vee \log n}{mn}\bigg) ^ {1/2},
\end{aligned}
\end{equation}
where we apply the fact that $n \geq C\max(\log n, p ^ 2)$ for some constant $C$. Similarly, we obtain that with probability at least $1 - 2 n ^ {-4}$ that
\begin{equation}
\label{eq:bound_T32}
	\|T_{22}\|_2 \lesssim  \phi ^ {1 / 2} M ^ {5 / 2} K ^ 5 \|\bSigma^{-1}\|_2 ^ 2 \bigg(\frac{p \vee \log n}{n}\bigg) ^ {1/2} \bigg(\frac{p \vee \log n}{mn}\bigg) ^ {1/2}.	
\end{equation}
\noindent \textit{\underline{Bound for $T_{23}$.}}
Define functions $U: \cX \rightarrow \RR ^ {p \times p}$ and $V: \cX \times \cY \rightarrow \RR ^ p$ as 
\[
	U(\bx) := b''(\bx ^ \top \bbeta ^ *) \bx\bx ^ \top - \bSigma
	~~\text{and}~~ 
	V(\bx, Y) := -  \bSigma ^ {-1} \bx (Y - b'(\bx ^ \top \bbeta ^ *)). 
\]
Then we have
\[
	\begin{aligned}
		T_{33} &= \EE (\nabla ^ 2 \ell ^ {(1)} (\bbeta ^ *) - \bSigma) \bSigma ^ {-1} \bTheta (\bSigma ^ {-1} \nabla \ell ^ {(1)} (\bbeta ^ *)) \otimes (\bSigma ^ {-1} \nabla \ell ^ {(1)} (\bbeta ^ *)) \\
		&= \frac{1}{n ^ 3} \EE \bigg\{\sum_{i=1}^n U(\bx_{i} ^ {(1)}) \bSigma ^ {-1} \bTheta V(\bx_{i} ^ {(1)}, Y_{i} ^ {(1)}) \otimes   V(\bx_{i} ^ {(1)}, Y_{i} ^ {(1)}) \bigg\} \\
		& = \frac{1}{n ^ 2} \EE \Big\{ U(\bx_1 ^ {(1)})  \bSigma ^ {-1} \bTheta V(\bx_1 ^ {(1)}, Y_1 ^ {(1)}) \otimes   V(\bx_1 ^ {(1)}, Y_1 ^ {(1)}) \Big\}.
	\end{aligned}
\]
By the Cauchy-Schwarz inequality and Jensen's inequality, we have
\[
\begin{aligned}
	\Big\| \EE \Big\{U(\bx_1 ^ {(1)}) \bSigma ^ {-1} \bTheta V(\bx_1 ^ {(1)}, Y_1 ^ {(1)}) \otimes V(\bx_1 ^ {(1)}, Y_1 ^ {(1)})\Big\} \Big\|_2 
	&\leq \EE \Big\| U(\bx_1 ^ {(1)}) \bSigma ^ {-1} \bTheta V(\bx_1 ^ {(1)}, Y_1 ^ {(1)}) \otimes  V(\bx_1 ^ {(1)}, Y_1 ^ {(1)}) \Big\|_2  \\
	&\leq \|\bSigma ^ {-1}\|_2 \|\bTheta\|_2 
			\Big\{\EE \big\|U(\bx_1 ^ {(1)})\big\|_2 ^ 2 
			~\EE \big\|V(\bx_1 ^ {(1)}, Y_1 ^ {(1)})\big\|_2 ^ 4 \Big\} ^ {1/2}\!\!.
\end{aligned}
\]
Note that
\[
	\EE \big\|U(\bx_1 ^ {(1)})\big\|_2 ^ 2 =\EE \big\|b''(\bx_1 ^ {(1)\top}  \bbeta ^ *) \bx_1 ^ {(1)} \bx_1 ^ {(1)\top} - \bSigma \big\|_2 ^ 2 \leq M ^ 2 K ^ 4 p ^ 2,
\]
and
\[
\begin{aligned}
	\EE \big\|V(\bx_1 ^ {(1)}, Y_1 ^ {(1)})\big\|_2  ^ 4
	&\leq \|\bSigma ^ {-1}\|_2 ^ 4 ~ \EE \big\| \bx_1 ^ {(1)} (Y_1 ^ {(1)} - b'(\bx_1 ^ {(1)\top} \bbeta ^ *)) \big\|_2 ^ 4 \\
	&\leq \Big\{\EE \big\|\bx_1 ^ {(1)}\big\|_2  ^ 8 ~\EE \big\| Y_1 ^ {(1)} - b'(\bx_1 ^ {(1)\top} \bbeta ^ *) \big\|_2 ^ 8 \Big\} ^ {1/2}
	\lesssim  (\phi M) ^ 2 K ^ 4 p ^ 2.
\end{aligned}
\]
Combining with the fact that $\|\bTheta\|_2 \lesssim M K^3$, we have
\begin{equation}
\label{eq:bound_T33}
	\|T_{23}\|_2 \lesssim \phi M ^ 3 K ^ 7 \|\bSigma ^ {-1} \|_2 ^ 3 \bigg(\frac{p}{n}\bigg) ^ 2.
\end{equation}
\noindent \textit{\underline{Bound for $T_{24}$.}}
Let 
\[
	\bw_k := (\nabla ^ 2 \ell ^ {(k)} (\bbeta ^ *) - \bSigma) \bSigma^{-1} \bTheta (\bSigma^{-1} \nabla \ell ^ {(k)} (\bbeta^ *)) \otimes (\bSigma ^ {-1} \nabla \ell ^ {(k)} (\bbeta^ *)).
\]
For any $\bu \in \cS ^ {p-1}$, $|\bu ^ \top(\bw_k - \EE \bw_k)| \lesssim \|\bw_k\|_2 \lesssim \phi M ^ 3 K ^ 7 \|\bSigma ^ {-1} \|_2 ^ 3 \Upsilon_1 ^ 3$. Similar to $T_{21}$, we have with probability at least $1 - 2 n ^ {-4}$ that
\begin{equation}
\label{eq:bound_T34}
	\|T_{24}\|_2 \lesssim \phi M ^ 3 K ^ 7 \|\bSigma ^ {-1} \|_2 ^ 3 \bigg(\frac{p \vee \log n}{n}\bigg) ^ {1/2} \bigg(\frac{p \vee \log n}{mn}\bigg) ^ {1 / 2}.
\end{equation}
\noindent \textit{\underline{Bound for $T_{25}$.}}
By the triangle inequality,
\begin{equation}
\label{eq:bound_T35}
	\begin{aligned}
		\|T_{25}\|_2  
		\leq \frac{1}{m} \sum_{k=1} ^ k \big\|(\nabla ^ 2 \ell ^ {(k)} (\bbeta ^ *) - \bSigma) \be ^ {(k)}\big\|_2 
		\lesssim C_{\kappa, \phi, M, K, \bSigma ^ {-1}} M K ^ 2 \bigg(\frac{p \vee \log n}{n}\bigg) ^ 2,
	\end{aligned}
\end{equation}
where $C_{\kappa, \phi, M, K, \bSigma ^ {-1}}$ is defined in Lemma \ref{prop:glm_decompn}.
Combining \eqref{eq:bound_T31}, \eqref{eq:bound_T32}, \eqref{eq:bound_T33} ,\eqref{eq:bound_T34} and \eqref{eq:bound_T35}, we have with the probability at least $1 - 6 n ^ {-4}$ that
\[
\begin{aligned}
		\|T_2\|_2 
		\lesssim &\phi ^ {1 / 2} M ^ {3 / 2} K ^ 3  \|\bSigma^{-1}\|_2 \bigg(\frac{p \vee \log n}{mn}\bigg) ^ {1 / 2} 
		 + C_{\kappa, \phi, M, K, \bSigma ^ {-1}}' \bigg(\frac{p \vee \log n}{n}\bigg) ^ 2,
\end{aligned}
\]
where $C_{\kappa, \phi, M, K, \bSigma ^ {-1}}' = \phi M ^ 3 K ^ 7 \|\bSigma ^ {-1} \|_2 ^ 3 + C_{\kappa, \phi, M, K, \bSigma ^ {-1}} M K ^ 2$.

\paragraph*{Bound for $T_3$}

Lemma \ref{lem:two_delta} delivers that
\[
\begin{aligned}
	\|T_3\|_2 
	\lesssim \bigg\|\frac{1}{m} \sum_{k=1} ^ m (\nabla ^ 3 \ell ^ {(k)} (\bbeta ^ *) - \bTheta) (\bSigma ^ {-1} \nabla \ell ^ {(k)}(\bbeta ^ *)) \otimes (\bSigma ^ {-1} \nabla \ell ^ {(k)}(\bbeta ^ *)) \bigg\|_2\!\!\!\! 
	+ C_{\kappa, \phi, M, K}' M K ^ 3 \|\bSigma ^ {-1}\|_2 ^ 2 \bigg(\frac{p \vee \log n}{n}\bigg) ^ 2 \!\!,
\end{aligned}
\]
where $C_{\kappa, \phi, M, K}' = \kappa ^ {-2} \phi ^ {3 / 2} M ^ {5 / 2} K ^ 6 + \kappa ^ {-1} \phi M ^ 2 K ^ 4$ which is defined in Lemma \ref{lem:two_delta}. Applying the same technique used in $T_{23}$ and $T_{24}$, we have with probability at least $1 - 2 n ^ {-4}$ that
\[
\begin{aligned}
	\bigg\|\frac{1}{m} \sum_{k=1} ^ m (\nabla ^ 3 \ell ^ {(k)} (\bbeta ^ *) &- \bTheta) (\bSigma ^ {-1} \nabla \ell ^ {(k)}(\bbeta ^ *)) \otimes (\bSigma ^ {-1} \nabla \ell ^ {(k)}(\bbeta ^ *)) \bigg\|_2 \\
	&\lesssim \phi M ^ 2 K ^ 5 \|\bSigma ^ {-1} \|_2 ^ 2 \bigg(\frac{p \vee \log n}{n}\bigg) ^ {3 / 2} \bigg(\frac{p \vee \log n}{m}\bigg) ^ {1 / 2}
	+ \phi M ^ 2 K ^ 5 \|\bSigma ^ {-1} \|_2 ^ 2 \bigg(\frac{p}{n}\bigg) ^ 2.
\end{aligned}
\]
Therefore, we have with probability at least $1 - 2 n ^ {-4}$ that
\[
	\|T_3\|_2
	\lesssim \phi M ^ 2 K ^ 5 \|\bSigma ^ {-1} \|_2 ^ 2 \bigg(\frac{p \vee \log n}{n}\bigg) ^ {\!\!1 / 2} \!\! \bigg(\frac{p \vee \log n}{mn}\bigg) ^ {\!\!1 / 2} \!\!\!\!\!\!
		+ C_{\kappa, \phi, M, K, \bSigma ^ {-1}}''  \bigg(\frac{p \vee \log n}{n}\bigg) ^ 2,	
\]
where $C_{\kappa, \phi, M, K, \bSigma ^ {-1}}'' = \phi M ^ 2 K ^ 5 \|\bSigma ^ {-1} \|_2 ^ 2 + C_{\kappa, \phi, M, K}' M K ^ 3 \|\bSigma ^ {-1}\|_2 ^ 2$.

\paragraph*{Bound for $T_4$}
Similarly, Lemma \ref{lem:two_delta} delivers that
\[
\begin{aligned}
	\|T_4\|_2 
	\lesssim &\bigg\|\frac{1}{m} \sum_{k=1} ^ m \bR_4 ^ {(k)} (\bSigma^{-1} \nabla \ell ^ {(k)}(\bbeta ^ *)) \otimes (\bSigma^{-1} \nabla \ell ^ {(k)}(\bbeta ^ *)) \otimes (\bSigma^{-1} \nabla \ell ^ {(k)}(\bbeta ^ *))\bigg\|_2 \\
	&+ C_{\kappa, \phi, M, K}'' M K ^ 4 \|\bSigma ^ {-1}\|_2 ^ 3 \bigg(\frac{p \vee \log n}{n}\bigg) ^ {2},
\end{aligned}
\]
where $C_{\kappa, \phi, M, K}'' = \kappa ^ {-2} \phi ^ 2 M ^ 3 K ^ 7  + \kappa ^ {-1} \phi ^ {3 / 2} M ^ {5 / 2} K ^ 5$ which is defined in Lemma \ref{lem:two_delta}. For convenience, let
\[
	\bs_k := \bR_4 ^ {(k)} (\bSigma^{-1} \nabla \ell ^ {(k)}(\bbeta ^ *)) \otimes (\bSigma^{-1} \nabla \ell ^ {(k)}(\bbeta ^ *))\otimes (\bSigma^{-1} \nabla \ell ^ {(k)}(\bbeta ^ *)).
\]
Applying the same technique used in $T_{24}$, we have with probability at least $1 - 2 n ^ {-4}$ that
\[
\begin{aligned}
	\Big\|\frac{1}{m} \sum_{k = 1} ^ m \bs_k\bigg\|_2
	&\leq \bigg\|\frac{1}{m} \sum_{k = 1} ^ m (\bs_k - \EE \bs_k)\bigg\|_2 + \bigg\|\frac{1}{m} \sum_{k = 1} ^ m \EE \bs_k\bigg\|_2 \\
	&\leq \phi ^ {3 / 2} M ^ {5 / 2} K ^ 7 \|\bSigma ^ {-1}\|_2 ^ 3 \bigg(\frac{p \vee 4 \log n}{n}\bigg) ^ {1 / 2} \bigg(\frac{p \vee 4 \log n}{mn}\bigg) ^ {1 / 2} + \|\EE \bs_1\|_2.
\end{aligned}
\]
Note that
\[
\begin{aligned}
	\|\EE \bs_1\|_2 
	&= \bigg\|\EE \bigg\{\int_0 ^ 1 \nabla ^ 4 \ell ^ {(1)} (\bbeta ^ * + v (\widehat \bbeta ^ {(1)} - \bbeta ^ *)) (\bSigma^{-1} \nabla \ell ^ {(1)}(\bbeta ^ *)) \otimes (\bSigma^{-1} \nabla \ell ^ {(1)}(\bbeta ^ *))\otimes (\bSigma^{-1} \nabla \ell ^ {(1)}(\bbeta ^ *)) dv\bigg\}\bigg\|_2 \\
	&= \sup_{\|\bu\|_2 = 1} \bigg|\EE \bigg\{\int_0 ^ 1 \frac{1}{n} \sum_{i=1} ^ n b''''(\bbeta ^ * + v (\widehat \bbeta ^ {(1)} - \bbeta ^ *)) (\bx_i ^ {(1)\top} \bSigma^{-1} \nabla \ell ^ {(1)}(\bbeta ^ *)) ^ 3 (\bx_i ^ {(1)\top} \bu) dv\bigg\} \bigg| \\
	&= \sup_{\|\bu\|_2 = 1} \bigg|\frac{1}{n} \sum_{i=1} ^ n \EE \bigg\{\int_0 ^ 1 b''''(\bbeta ^ * + v (\widehat \bbeta ^ {(1)} - \bbeta ^ *)) dv (\bx_i ^ {(1)\top} \bSigma^{-1} \nabla \ell ^ {(1)}(\bbeta ^ *)) ^ 3 (\bx_i ^ {(1)\top} \bu)\bigg\} \bigg|.
\end{aligned}
\]
By H\"older's inequality and Jensen's inequality, we have
\[
\begin{aligned}
	\EE \bigg\{\int_0 ^ 1 b''''(\bbeta ^ * &+ v (\widehat \bbeta ^ {(1)} - \bbeta ^ *)) dv (\bx_i ^ {(1)\top} \bSigma^{-1} \nabla \ell ^ {(1)}(\bbeta ^ *)) ^ 3 (\bx_i ^ {(1)\top} \bu)\bigg\} \\
	&\leq \bigg[\EE \bigg\{\int_0 ^ 1 b''''(\bbeta ^ * + v (\widehat \bbeta ^ {(1)} - \bbeta ^ *)) dv (\bx_i ^ {(1)\top} \bSigma^{-1} \nabla \ell ^ {(1)}(\bbeta ^ *)) ^ 3\bigg\} ^ 2\bigg] ^ {1 / 2} \Big\{\EE (\bx_i ^ {(1)\top} \bu) ^ 2\Big\} ^ {1 / 2} \\
	&\leq \bigg[\EE \int_0 ^ 1 \Big\{b''''(\bbeta ^ * + v (\widehat \bbeta ^ {(1)} - \bbeta ^ *))\Big\} ^ 2 dv (\bx_i ^ {(1)\top} \bSigma^{-1} \nabla \ell ^ {(1)}(\bbeta ^ *)) ^ 6\bigg] ^ {1 / 2} \Big\{\EE (\bx_i ^ {(1)\top} \bu) ^ 2\Big\} ^ {1 / 2} \\
	&\leq M \Big\{\EE (\bx_i ^ {(1)\top} \bSigma^{-1} \nabla \ell ^ {(1)}(\bbeta ^ *)) ^ 6\Big\} ^ {1 / 2} \Big\{\EE (\bx_i ^ {(1)\top} \bu) ^ 2\Big\} ^ {1 / 2}.
\end{aligned}
\]
By the Cauchy-Schwarz inequality, we have
\[
	\|\EE \bs_1\|_2
	\leq M\sup_{\|\bu\|_2 = 1} \bigg[\bigg\{\frac{1}{n} \sum_{i=1} ^ n \EE (\bx_i ^ {(1)\top} \bSigma^{-1} \nabla \ell ^ {(1)}(\bbeta ^ *)) ^ 6 \bigg\} \bigg\{\frac{1}{n} \sum_{i=1} ^ n \EE (\bx_i ^ {(1)\top} \bu) ^ 2\bigg\}\bigg] ^ {1 / 2}.
\]
Applying the same technique used in $T_{23}$ and the fact that $\bx ^ \top \bSigma ^ {-1} \bx \leq \|\bSigma ^ {-1}\|_2 \|\bx\|_2 ^ 2$, we have
\[
\begin{aligned}
	\EE (\bx_i ^ {(1)\top} \bSigma^{-1} \nabla \ell ^ {(1)}(\bbeta ^ *)) ^ 6
	&\leq  \frac{\EE \Big\{\bx_1 ^ {(1)\top} \bSigma ^ {-1} \bx_1 ^ {(1)} (Y_1 ^ {(1)} - b'(\bx_1 ^ {(1)\top} \bbeta ^ *))\Big\} ^ 6}{n ^ 5} 
	\lesssim \phi ^ 3 M ^ 3 K ^ {12} \|\bSigma ^ {-1}\|_2 ^ 6 \bigg(\frac{p ^ 6}{n ^ 5}\bigg).
\end{aligned}
\]
By Condition \ref{con:distribution}, we have $\EE (\bx_i ^ \top \bu) ^ 2 \lesssim K ^ 2$. Recall the fact that $n \geq C\max(\log n, p ^ 2)$ for some constant $C$. This implies that
\[
\begin{aligned}
	\|\EE \bs_1\|_2 \lesssim \phi ^ {3 / 2} M ^ {5 / 2} K ^ 7 \|\bSigma ^ {-1}\|_2 ^ 3 \bigg(\frac{p}{n}\bigg) ^ 2.
\end{aligned}
\]
It thus follows that with probability at least $1 - 2 n ^ {-4}$ that
\[
	\|T_4\|_2 
	\lesssim \phi ^ {3 / 2} M ^ {5 / 2} K ^ 7 \|\bSigma ^ {-1}\|_2 ^ 3 \bigg(\frac{p \vee \log n}{n}\bigg) \bigg(\frac{p \vee \log n}{mn}\bigg) ^ {1 / 2}
	+ C_{\kappa, \phi, M, K, \bSigma ^ {-1}}''' \bigg(\frac{p \vee \log n}{n}\bigg) ^ 2,
\]
where $C_{\kappa, \phi, M, K, \bSigma ^ {-1}}''' = \phi ^ {3 / 2} M ^ {5 / 2} K ^ 7 \|\bSigma ^ {-1}\|_2 ^ 3 + C_{\kappa, \phi, M, K}'' M K ^ 4\|\bSigma ^ {-1}\|_2 ^ 3$.

\paragraph*{Bound for $T_5$}
Similar to $T_5$, we derive with probability at least $1 - 2 n ^ {-4}$ that
\[
	\|T_5\|_2
	\lesssim \phi ^ {3 / 2} M ^ {5 / 2} K ^ 7 \|\bSigma ^ {-1}\|_2 ^ 3 \bigg(\frac{p \vee \log n}{n}\bigg) \bigg(\frac{p \vee \log n}{mn}\bigg) ^ {1 / 2}
	+ C_{\kappa, \phi, M, K, \bSigma ^ {-1}}''' \bigg(\frac{p \vee \log n}{n}\bigg) ^ 2.
\]

Combining the bounds for $\{T_i\}_{i=1} ^ 5$ conditional on event $\cA$, we have with probability at least $1 - 14 n ^ {-4}$ that 
\begin{equation}
\label{eq:bound_grad}
\begin{aligned}
	\|\nabla \widetilde \ell(\bbeta ^ *)\|_2
	&\lesssim \check C_{\kappa, \phi, M, K, \bSigma ^ {-1}}' \bigg(\frac{p \vee 4 \log n}{mn}\bigg) ^ {1 / 2}
	+ \check C_{\kappa, \phi, M, K, \bSigma ^ {-1}}'' \bigg(\frac{p \vee 4 \log n}{n}\bigg) ^ 2, 
\end{aligned}
\end{equation}
where $\check C_{\kappa, \phi, M, K, \bSigma ^ {-1}}' = (\phi M) ^ {1 / 2} K + \kappa ^ {-1} \phi ^ {1 / 2} M ^ {3 / 2} K ^ 3 + \phi ^ {1 / 2} M ^ {3 / 2} K ^ 3  \|\bSigma^{-1}\|_2$ and $\check C_{\kappa, \phi, M, K, \bSigma ^ {-1}}'' = C_{\kappa, \phi, M, K, \bSigma ^ {-1}}' + C_{\kappa, \phi, M, K, \bSigma ^ {-1}}'' + C_{\kappa, \phi, M, K, \bSigma ^ {-1}}'''$. Consider the failure probability of $\cA$, we derive the desired result.	
\end{proof}

\subsection{Proof of Theorem \ref{thm:glm_reboot_mse}}

\begin{proof}
For simplicity, we use $\widetilde \bDelta$ to denote $\widehat \bbeta ^ \rb - \bbeta ^ *$. Construct an intermediate estimator $\widetilde \bbeta_\eta$ between $\widetilde \bbeta$ and $\bbeta^ *$:
\begin{equation}
	\widehat \bbeta ^ \rb _ \eta = \bbeta^ * + \eta (\widehat \bbeta ^ \rb - \bbeta^ *)
\end{equation}
where $\eta = 1$ if $\|\widehat \bbeta ^ \rb - \bbeta ^ *\|_2 \leq 1$ and $\eta = 1 /\|\widehat \bbeta ^ \rb - \bbeta ^ *\|_2$ if $\|\widehat \bbeta ^ \rb - \bbeta ^ *\|_2 > 1$. Let $\widetilde \bDelta _ \eta := \widehat \bbeta ^ \rb_\eta - \bbeta ^ *$. By \eqref{eq:reboot_glr_lrsc} in Proposition \ref{cor:_reboot_glr_lrsc}, we have \[
\begin{aligned}
	\kappa \|\widetilde \bDelta_\eta\|_2 ^ 2 
	&\leq \widetilde \ell (\bbeta ^ * + \widetilde \bDelta_\eta) - \widetilde \ell (\bbeta ^ *) - \nabla \widetilde \ell (\bbeta ^ *) ^ \top \widetilde \bDelta_\eta 
	\leq - \nabla \widetilde\ell (\bbeta ^ *) ^ \top \widetilde \bDelta_\eta 
	\leq \|\nabla \widetilde \ell (\bbeta ^ *)\|_2 \|\widetilde \bDelta_\eta\|_2.
\end{aligned}
\]		
which implies that
\begin{equation}
\label{eq:proof1_delta1}
	\|\widetilde \bDelta _ \eta\| _ 2 
	\leq \kappa ^ {-1} \|\nabla \widetilde \ell(\bbeta ^ *)\|_2 
	\leq \kappa ^ {-1} \bigg\{\check C_{\kappa, \phi, M, K, \bSigma ^ {-1}}' \bigg(\frac{p \vee \log n}{mn}\bigg) ^ {1 / 2}
		+ \check C_{\kappa, \phi, M, K, \bSigma ^ {-1}}'' \bigg(\frac{p \vee \log n}{n}\bigg) ^ 2\bigg\},
\end{equation}
where $\check C_{\kappa, \phi, M, K, \bSigma ^ {-1}}'$ and $\check C_{\kappa, \phi, M, K, \bSigma ^ {-1}}''$ is defined in Theorem \ref{thm:glm_reboot_gradient}. Then with probability at least $1 - (12 m + 14) n ^ {-4}$, we have $\|\widetilde \bDelta _ \eta\|_2 < 1$, which further implies that $ \widetilde \bDelta = \widetilde \bDelta _ \eta$ according to the construction of $\widetilde \bDelta _ \eta$. The conclusion thus follows.
\end{proof}

\subsection{Proof of Corollary \ref{thm:misspecified}}
Applying the fourth-order Taylor expansion of $\nabla \ell (\widehat \bbeta ^ {(k)}; \widetilde \bz ^ {(k)})$ at $\bbeta ^ *$ yields that 
\[ 
\begin{aligned}
	\nabla \ell (\widehat \bbeta ^ {(k)}; \widetilde \bz ^ {(k)})
	&= \nabla \ell (\bbeta ^ *; \widetilde \bz ^ {(k)}) 
	+ \nabla ^ 2 \ell (\bbeta ^ *; \widetilde \bz ^ {(k)}) \bDelta ^ {(k)}
	+ \nabla ^ 3 \ell (\bbeta ^ *; \widetilde \bz ^ {(k)}) (\bDelta ^ {(k)} \otimes \bDelta ^ {(k)}), \\
	&\quad+ \bigg\{\int_0 ^ 1 \nabla ^ 4 \ell (\bbeta ^ * + v (\widehat \bbeta ^ {(k)} - \bbeta ^ *), \widetilde \bz ^ {(k)}) dv\bigg\} (\bDelta ^ {(k)} \otimes \bDelta ^ {(k)} \otimes \bDelta ^ {(k)}). 
\end{aligned}
\]
By taking the conditional expectation given $\widehat \bbeta ^ {(k)}$, we have
\begin{equation} 
\label{eq:proof2_subtit1}
	\begin{aligned}
		0 &= \EE \big\{\nabla \ell (\widehat \bbeta ^ {(k)} ;\widetilde \bz ^ {(k)}) \big|\, \widehat \bbeta ^ {(k)}\big\}
		  + \widetilde \bSigma \bDelta ^ {(k)}
		  + \widetilde \bTheta (\bDelta ^ {(k)} \otimes \bDelta ^ {(k)})+ \widetilde \bR_4 ^ {(k)} (\bDelta ^ {(k)} \otimes \bDelta ^ {(k)} \otimes \bDelta ^ {(k)}),
	\end{aligned}
\end{equation}
where $\widetilde \bR_4 ^ {(k)} = \EE \big\{\int_0 ^ 1 \nabla ^ 4 \ell (\bbeta ^ * + v (\widehat \bbeta ^ {(k)} - \bbeta ^ *), \widetilde \bz ^ {(k)}) dv \big|\, \widehat \bbeta ^ {(k)}\big\}$. Similarly, by the fourth-order Taylor expansion and the fact that $\nabla \ell ^ {(k)} (\widehat \bbeta ^ {(k)}) = 0$ , we have
\[
\begin{aligned}
	\nabla \ell ^ {(k)} (\widehat \bbeta ^ {(k)})
	&= \nabla \ell ^ {(k)} (\bbeta ^ *) + \bSigma \bDelta ^ {(k)} + \bTheta (\bDelta ^ {(k)} \otimes \bDelta ^ {(k)}) + (\nabla ^ 2 \ell ^ {(k)} (\bbeta ^ *) - \bSigma) \bDelta ^ {(k)} \\
	&\quad + (\nabla ^ 3 \ell^{(k)} (\bbeta ^ *) - \bTheta) (\bDelta ^ {(k)} \otimes \bDelta ^ {(k)})
  		+ \bR_4 ^ {(k)} (\bDelta ^ {(k)} \otimes \bDelta ^ {(k)} \otimes \bDelta ^ {(k)})
  	= 0, 
\end{aligned}
\]
where $\bR_4 ^ {(k)} := \int_0 ^ 1 \nabla ^ 4 \ell ^ {(k)} \bigl(\bbeta ^ * + v (\widehat \bbeta ^ {(k)} - \bbeta ^ *)\bigr) dv$. This can be rearranged as 
\beq
    \label{eq:substi_3}
    \begin{aligned}
        \bSigma \bDelta ^ {(k)} + \bTheta (\bDelta ^ {(k)} \otimes \bDelta ^ {(k)}) 
        &= - \nabla \ell ^ {(k)} (\bbeta ^ *) - (\nabla ^ 2 \ell ^ {(k)} (\bbeta ^ *) - \bSigma) \bDelta ^ {(k)} - (\nabla ^ 3 \ell ^ {(k)} (\bbeta ^ *) - \bTheta)  (\bDelta ^ {(k)} \otimes \bDelta ^ {(k)}) \\
        &\quad - \bR_4 ^ {(k)} (\bDelta ^ {(k)} \otimes \bDelta ^ {(k)} \otimes \bDelta ^ {(k)}).
    \end{aligned}
\eeq
Subtract \eqref{eq:substi_3} from \eqref{eq:proof2_subtit1}, we have
\[
	\begin{aligned}
		0
		&= \EE \big\{\nabla \ell (\widehat \bbeta ^ {(k)} ;\widetilde \bz ^ {(k)}) \big|\, \widehat \bbeta ^ {(k)}\big\} 
		- \nabla \ell ^ {(k)} (\bbeta ^ *)  
		- (\nabla ^ 2 \ell ^ {(k)} (\bbeta ^ *) - \bSigma) \bDelta ^ {(k)}
		- (\nabla ^ 3 \ell ^ {(k)} (\bbeta ^ *) - \bTheta)  (\bDelta ^ {(k)} \otimes \bDelta ^ {(k)}) \\
		&\quad
		- \bR_4 ^ {(k)} (\bDelta ^ {(k)} \otimes \bDelta ^ {(k)} \otimes \bDelta ^ {(k)})
		+ \widetilde \bR_4 ^ {(k)} (\bDelta ^ {(k)} \otimes \bDelta ^ {(k)} \otimes \bDelta ^ {(k)})
		+ (\widetilde \bSigma - \bSigma) \bDelta ^ {(k)}
		+ (\widetilde \bTheta - \bTheta)  (\bDelta ^ {(k)} \otimes \bDelta ^ {(k)}).
	\end{aligned}
\]
Note that $\nabla \widetilde \ell(\bbeta ^ *) = \frac{1}{m} \sum_{k=1} ^ m \EE \big\{ \nabla \ell (\bbeta ^ *; \widetilde \bz ^ {(k)})|\, \widehat \bbeta ^ {(k)}\big\}$. Then we have the following decomposition: 
\[
\begin{aligned}
	\nabla \widetilde \ell(\bbeta ^ *) 
	= &\underbrace{\nabla \ell (\bbeta ^ *)}_{T_1}
	+ \underbrace{\frac{1}{m} \sum_{k=1} ^ m (\nabla ^ 2 \ell ^ {(k)} (\bbeta ^ *) - \bSigma) \bDelta ^ {(k)}}_{T_2} 
	+ \underbrace{\frac{1}{m} \sum_{k=1} ^ m (\nabla ^ 3 \ell ^ {(k)} (\bbeta ^ *) - \bTheta) (\bDelta ^ {(k)} \otimes \bDelta ^ {(k)})}_{T_3} \\ 
	&+ \underbrace{\frac{1}{m} \sum_{k=1} ^ m \bR_4 ^ {(k)} (\bDelta ^ {(k)} \otimes \bDelta ^ {(k)} \otimes \bDelta ^ {(k)})}_{T_4} 
	- \underbrace{\frac{1}{m} \sum_{k=1} ^ m \widetilde \bR_4 ^ {(k)} (\bDelta ^ {(k)} \otimes \bDelta ^ {(k)} \otimes \bDelta ^ {(k)})}_{T_5} \\
	&+ \underbrace{\frac{1}{m} \sum_{k=1} ^ m(\widetilde \bSigma - \bSigma) \bDelta ^ {(k)}}_{T_6}
	+ \underbrace{\frac{1}{m} \sum_{k=1} ^ m(\widetilde \bTheta - \bTheta) (\bDelta ^ {(k)} \otimes \bDelta ^ {(k)})}_{T_7}.
\end{aligned}	
\]
Following the proof in Theorem \ref{thm:glm_reboot_gradient}, we introduce the following events:
\[
\begin{aligned}
	\cE ^ {(k)} := 
		&\big\{\|\nabla \ell ^ {(k)} (\bbeta ^ *)\|_2
			\leq C_1 (\phi M) ^ {1 / 2} K \Upsilon_1\big\}\cap \\
		&\big\{\|\nabla ^ 2 \ell ^ {(k)} (\bbeta ^ *) - \bSigma\|_2 
			\leq C_2 M K ^ 2 \Upsilon_1\big\}\cap \\
		&\big\{\|\nabla ^ 3 \ell ^ {(k)} (\bbeta ^ *) - \bTheta\|_2
			\leq C_3 M K ^ 3 \Upsilon_1\big\} \cap \\
		&\big\{\|\bR_4 ^ {(k)}\|_2 \leq C_4 M K ^ 4 \big\}, \\
	\cF^{(k)} := 
		&\big\{\ltwonorm{\bDelta ^ {(k)}} \leq C_5 \kappa ^ {-1} (\phi M) ^ {1 / 2} K \Upsilon_1\big\}
\end{aligned}
\]
where $\Upsilon_1 := \{(p \vee \log n) / n\} ^ {1 / 2}$ and $C_1, C_2, \ldots, C_5$ are constants. Define the intersection of all the above events by $\cA$. Then we work on the bounds for $\{T_i\}_{i=1} ^ 7$ conditional on event $\cA$. The statistical error bound for $T_1, \ldots, T_5$ is the same as in Theorem \ref{thm:glm_reboot_gradient}. Therefore, we focus on $T_6$ and $T_7$.
\paragraph*{Bound for $T_6$}
Let $\widetilde \bx = \bw + \bdelta$ with $\bw \sim \cN(\bmu, \bS)$ and $\bdelta \sim \cN(\widetilde \bmu - \bmu, \widetilde \bS -\bS)$, where $\bw$ and $\bdelta$ are independent. By the definition of $\widetilde \bSigma$ and $\bSigma$, we have
\[
\begin{aligned}
	\widetilde \bSigma - \bSigma
	&=\EE\{b''(\widetilde \bx ^ \top \bbeta ^ *)\widetilde \bx \widetilde \bx ^ \top\} 
	- \EE\{b''(\bx ^ \top \bbeta ^ *)\bx\bx ^ \top\} 
	\!=\! \EE\{b''(\bw ^ \top \bbeta ^ * \!+\! \bdelta ^ \top \bbeta ^ *)(\bw \!+\! \bdelta) (\bw \!+\! \bdelta) ^ \top\} - \EE\{b''(\bx ^ \top \bbeta ^ *)\bx\bx ^ \top\} \\
	&= \EE\{b''(\bw ^ \top \bbeta ^ * + \bdelta ^ \top \bbeta ^ *)\bw \bw ^ \top\}  
		+ 2 \EE\{b''(\bw ^ \top \bbeta ^ * + \bdelta ^ \top \bbeta ^ *)\bw \bdelta ^ \top\} 
		+ \EE\{b''(\bw ^ \top \bbeta ^ * + \bdelta ^ \top \bbeta ^ *)\bdelta \bdelta ^ \top\} \\
		&\quad - \EE\{b''(\bx ^ \top \bbeta ^ *)\bx \bx ^ \top\} \\
	&= \EE\{b''(\bw ^ \top \bbeta ^ *)\bw \bw ^ \top\} 
	   + \EE\{b'''(\bw ^ \top \bbeta ^ * + t\bdelta ^ \top \bbeta ^ *)(\bdelta ^ \top \bbeta ^ *)\bw \bw ^ \top\}
	   + 2 \EE\{b''(\bw ^ \top \bbeta ^ * + \bdelta ^ \top \bbeta ^ *)\bw \bdelta ^ \top\} \\
	   &\quad + \EE\{b''(\bw ^ \top \bbeta ^ * + \bdelta ^ \top \bbeta ^ *)\bdelta \bdelta ^ \top\} - \EE\{b''(\bx ^ \top \bbeta ^ *)\bx \bx ^ \top\},
\end{aligned}
\]
where $t \in [0, 1]$ and we apply Taylor's expansion in the last equation. Since $\EE\{b''(\bw ^ \top \bbeta ^ *)\bw \bw ^ \top\} = \EE\{b''(\bx ^ \top \bbeta ^ *) \bx \bx ^ \top\}$, it follows that
\[
	\widetilde \bSigma - \bSigma
	= \EE\{b'''(\bw ^ \top \bbeta ^ * + t\bdelta ^ \top \bbeta ^ *)(\bdelta ^ \top \bbeta ^ *)\bw \bw ^ \top\}
	+ 2 \EE\{b''(\bw ^ \top \bbeta ^ * + \bdelta ^ \top \bbeta ^ *)\bw \bdelta ^ \top\} 
	+ \EE\{b''(\bw ^ \top \bbeta ^ * + \bdelta ^ \top \bbeta ^ *)\bdelta \bdelta ^ \top\}.
\]
By the triangle inequality and the Cauchy-Schwarz inequality, we have
\begin{equation}
\begin{aligned}
\label{eq:mis_Sigma}
	\|\widetilde \bSigma - \bSigma\|_2
	&= \sup_{\bu \in \cS ^ {p-1}} |\bu ^ \top (\widetilde \bSigma - \bSigma)\bu| \\
	&\leq \sup_{\bu \in \cS ^ {p-1}} \big|\EE\{b'''(\bw ^ \top \bbeta ^ * + t\bdelta ^ \top \bbeta ^ *)(\bdelta ^ \top \bbeta ^ *) (\bw ^ \top \bu) ^ 2\}\big| 
		+ 2 \sup_{\bu \in \cS ^ {p-1}} \big|\EE\{b''(\bw ^ \top \bbeta ^ * + \bdelta ^ \top \bbeta ^ *) (\bw ^ \top \bu) (\bdelta ^ \top \bu)\}\big| \\
		&\quad+ \sup_{\bu \in \cS ^ {p-1}} \big|\EE\{b''(\bw ^ \top \bbeta ^ * + \bdelta ^ \top \bbeta ^ *) (\bdelta ^ \top \bu ) ^ 2\}\big| \\
	&\leq \sup_{\bu \in \cS ^ {p-1}} \big[\EE\{b'''(\bw ^ \top \bbeta ^ * + t\bdelta ^ \top \bbeta ^ *)(\bdelta ^ \top \bbeta ^ *)\} ^ 2 \EE \{(\bw ^ \top \bu) ^ 4\}\big] ^ {1/2} \\
		&\quad+ 2 \sup_{\bu \in \cS ^ {p-1}} \big[\EE\{b''(\bw ^ \top \bbeta ^ * + \bdelta ^ \top \bbeta ^ *) (\bw ^ \top \bu)\}^2 \EE \{(\bdelta ^ \top \bu)^2\}\big] ^ {1/2} 
		+ \sup_{\bu \in \cS ^ {p-1}} M \EE\{ (\bdelta ^ \top \bu ) ^ 2\} \\
	&\leq M K ^ 2 \sup_{\bu \in \cS ^ {p-1}} \big[\EE\{(\bdelta ^ \top \bbeta ^ *)^2\} \big] ^ {1/2} 
		+ 2 M K \sup_{\bu \in \cS ^ {p-1}} \big[\EE \{(\bdelta ^ \top \bu)^2\}\big] ^ {1/2} 
		+ M \sup_{\bu \in \cS ^ {p-1}} \EE\{(\bdelta ^ \top \bu ) ^ 2\}.
\end{aligned}
\end{equation}
For any $\bu \in \cS ^ {p-1}$, $\bdelta ^ \top \bu \sim \cN((\widetilde \bmu - \bmu) ^ \top \bu, \bu ^ \top (\widetilde \bS -\bS) \bu)$. By the property of normal distribution,
\[
\begin{aligned}
	\sup_{\bu \in \cS ^ {p-1}} &\EE\{(\bdelta ^ \top \bu ) ^ 2\}
	 = \sup_{\bu \in \cS ^ {p-1}} \big[\bu ^ \top (\widetilde \bS -\bS) \bu 
	 + \{(\widetilde \bmu - \bmu) ^ \top \bu\} ^ 2\big]
	 \leq \|\widetilde \bS -\bS\|_2 + \|\widetilde \bmu - \bmu\|_2 ^ 2.
\end{aligned}
\]
Similarly, $\EE\{(\bdelta ^ \top \bbeta ^ *) ^ 2\} \leq \|\bbeta ^ *\|_2 ^ 2 \big(\|\widetilde \bS -\bS\|_2 + \|\widetilde \bmu - \bmu\|_2 ^ 2 \big)$. Combining with the bound \eqref{eq:mis_Sigma}, we have
\[
	\|\widetilde \bSigma - \bSigma\|_2 
	\leq M K ^ 2 \|\bbeta ^ *\|_2 (\|\widetilde \bS -\bS\|_2 ^ {1/2} + \|\widetilde \bmu - \bmu\|_2).
\]
By Lemma \ref{lem:avergae_est}, we have with probability at least $1 - 2 n ^ {-4} - 4 m n ^ {-4}$ that
\[
	\|T_6\|_2 
	\leq \|\widetilde \bSigma - \bSigma\|_2 \bigg\|\frac{1}{m} \sum_{k=1} ^ m  \bDelta ^ {(k)}\bigg\|_2 
	\lesssim C_{\kappa, \phi, M, K, \bbeta ^ *, \bSigma ^ {-1}} (\|\widetilde \bS -\bS\|_2 ^ {1/2} + \|\widetilde \bmu - \bmu\|_2) \bigg(\frac{p \vee \log n}{n}\bigg),
\]
where $C_{\kappa, \phi, M, K, \bbeta ^ *, \bSigma ^ {-1}} := (\kappa ^ {-2} \phi M ^ 3 K ^ 7 + \kappa ^ {-1} \phi ^ {1 / 2} M ^ {5 / 2} K ^ 5) \|\bSigma ^ {-1}\|_2 \|\bbeta ^ *\|_2$.

\paragraph*{Bound for $T_7$}
By the definition of $\widetilde \bTheta$ and $\bTheta$, we have
\[
\begin{aligned}
	&\|\widetilde \bTheta - \bTheta\|_2
	= \sup_{\bu \in \cS ^ {p-1}} |\EE b'''(\widetilde \bx ^ \top \bbeta ^ *) (\widetilde \bx ^ \top \bu) ^ 3 - \EE b'''(\bx ^ \top \bbeta ^ *) (\bx ^ \top \bu) ^ 3| \\
	&= \sup_{\bu \in \cS ^ {p-1}} |\EE b'''(\bw ^ \top \bbeta ^ * + \bdelta ^ \top \bbeta ^ *) (\bw ^ \top \bu + \bdelta ^ \top \bu) ^ 3 - \EE b'''(\bx ^ \top \bbeta ^ *) (\bx ^ \top \bu) ^ 3| \\
	&= \sup_{\bu \in \cS ^ {p-1}} |\EE b'''(\bw ^ \top \bbeta ^ *)(\bw ^ \top \bu) ^ 3
	+ \EE b''''(\bw ^ \top \bbeta ^ * + t \bdelta ^ \top \bbeta ^ *)(\bw ^ \top \bu) ^ 3 (\bdelta ^ \top \bbeta ^ *)
	+ 3 \EE b'''(\bw ^ \top \bbeta ^ * + \bdelta ^ \top \bbeta ^ *) (\bw ^ \top \bu) (\bdelta ^ \top \bu) ^ 2 \\
	&\quad+ 3 \EE b'''(\bw ^ \top \bbeta ^ * + \bdelta ^ \top \bbeta ^ *)(\bw ^ \top \bu) ^ 2 (\bdelta ^ \top \bu) 
	+ \EE b'''(\bw ^ \top \bbeta ^ * + \bdelta ^ \top \bbeta ^ *) (\bdelta ^ \top \bu) ^ 3\} 
	- \EE b'''(\bx ^ \top \bbeta ^ *) (\bx ^ \top \bu\}|,
\end{aligned}
\]
where $t \in [0, 1]$ and we apply Taylor's expansion in the last equation. Similar to the bound \eqref{eq:mis_Sigma}, we have
\[
\begin{aligned}
	\|\widetilde \bTheta - \bTheta\|_2
	\lesssim M K ^ 3 \|\bbeta ^ *\|_2 (\|\widetilde \bS -\bS\|_2 ^ {1/2} + \|\widetilde \bmu - \bmu\|_2).
\end{aligned}
\]
Therefore,
\[
\begin{aligned}
	\|T_7\|_2 
	&\leq \|\widetilde \bTheta - \bTheta\|_2 \bigg\|\frac{1}{m} \sum_{k=1} ^ m (\bDelta ^ {(k)} \otimes \bDelta ^ {(k)})\bigg\|_2 
	\leq \frac{1}{m} \sum_{k=1} ^ m \|\widetilde \bTheta - \bTheta\|_2\|\bDelta ^ {(k)}\|_2 ^ 2 \\
	&\lesssim \kappa ^ {-2} \phi M ^ 2 K ^ 5 \|\bbeta ^ *\|_2 (\|\widetilde \bS -\bS\|_2 ^ {1/2} + \|\widetilde \bmu - \bmu\|_2) \bigg(\frac{p \vee \log n}{n}\bigg).
\end{aligned}
\]

Combining the bounds for $\{T_i\}_{i=1} ^ 7$ conditional on event $\cA$, we have with probability at least $1 - 16 n ^ {-4} - 4 m n ^ {-4}$ that 
\begin{equation}
\begin{aligned}
	\|\nabla \widetilde \ell(\bbeta ^ *)\|_2
	\lesssim &\check C_{\kappa, \phi, M, K, \bSigma ^ {-1}}' \bigg(\frac{p \vee 4 \log n}{mn}\bigg) ^ {1 / 2} 
	+ \check C_{\kappa, \phi, M, K, \bSigma ^ {-1}}'' \bigg(\frac{p \vee 4 \log n}{n}\bigg) ^ 2 \\
	&+ C_{\kappa, \phi, M, K, \bbeta ^ *, \bSigma ^ {-1}} (\|\widetilde \bS -\bS\|_2 ^ {1/2} + \|\widetilde \bmu - \bmu\|_2) \bigg(\frac{p \vee \log n}{n}\bigg), 
\end{aligned}
\end{equation}
where $\check C_{\kappa, \phi, M, K, \bSigma ^ {-1}}'$ and $\check C_{\kappa, \phi, M, K, \bSigma ^ {-1}}''$ is defined in Theorem \ref{thm:glm_reboot_gradient}. Consider the failure probability of $\cA$, we derive the desired result.

\subsection{Proof of Lemma \ref{prop:pr1}}

\begin{proof}
The proof follows that of Lemma 5 in \cite{ma2018implicit}. 
For notational simplicity, we omit ``${(k)}$'' in the superscript in the following proof. Let 
\[
	\bY := \frac{1}{n} \sum_{i=1}^{n} y_i \bx_i \bx_i^\top 
	    = \frac{1}{n} \sum_{i=1}^{n} (\bx_i ^ \top \bbeta ^ *) ^ 2 \bx_i \bx_i ^ \top + \varepsilon_i \bx_i \bx_i ^ \top.
\]
By Lemmas \ref{lem:pr_center_hess} and \ref{lem:pr_center_veps}, we have with probability at least $1 - 18 n ^ {-2}$ that
\[
	\|\bY - \EE \bY\|_2 \leq r \|\bbeta\| _ 2 ^ 2,
\]
provided that $n \geq C_r p ^ 2$ for some positive constant $C_r$ depending on $r$. Since $\EE \bY = \|\bbeta^ *\|_2^2 \bI_p + 2\bbeta ^ *\bbeta ^ {*\top}$, it follows that $\lambda_1(\EE \bY ) = 3\|\bbeta ^ *\|_2^2$, $\lambda_2(\EE \bY ) = \|\bbeta ^ *\|_2^2$, and the leading eigenvector of $\EE \bY $, denoted by $\bv_1(\EE \bY )$, is $\bbeta ^ * / \|\bbeta^ *\|_2$. On one hand, by the Davis-Kahan theorem (Theorem 1 in \cite{yu2015useful}), we obtain that
\[
	\|\bv_1(\bY) - \bv_1(\EE \bY )\|_2 
	\leq 2 \sqrt{2} \frac{\|\bY - \EE \bY \|_2}{\lambda_1(\EE \bY ) - \lambda_2(\EE \bY )}	
	\le \sqrt{2} r.
\]
On the other hand, by Weyl's inequality, we have
\[
\begin{aligned}
	\big|\sqrt{\lambda_1(\bY)/3} - \|\bbeta ^ *\|_2\big|
	&= \bigg|\frac{\lambda_1(\bY)/3 - \|\bbeta ^ *\|_2 ^ 2}{\sqrt{\lambda_1(\bY)/3} + \|\bbeta ^ *\|_2}\bigg| 
	\leq \frac{|\lambda_1(\bY) - 3\|\bbeta ^ *\|_2 ^ 2|}{3\|\bbeta ^ *\|_2}  
	\leq \frac{\|\bY - \EE \bY \|_2}{3\|\bbeta ^ *\|_2} 
		\le \frac{r}{3} \|\bbeta ^ *\|_2 .
\end{aligned}
\]
Therefore, we deduce that if $n \geq C_r p ^ 2$, with probability at least $1 - 18 n ^ {-2}$ that
\[
\begin{aligned}
	\|\widehat \bbeta_{\rm{init}} - \bbeta ^ *\|_2 
	&= \big\|\sqrt{\lambda_1(\bY)/3} \bv_1(\bY) - \bbeta ^ *\big\|_2 \\
	&= \big\|\sqrt{\lambda_1(\bY)/3} \bv_1(\bY) - \|\bbeta ^ *\|_2 \bv_1(\bY) + \|\bbeta ^ *\|_2 \bv_1(\bY) - \bbeta ^ *\big\|_2 \\
	&\leq \big|\sqrt{\lambda_1(\bY)/3} - \|\bbeta ^ *\|_2\big| 
	+ \|\bbeta ^ *\|_2 \|\bv_1(\bY) - \bv_1(\EE \bY )\|_2
		\le 2 r \|\bbeta ^ *\|_2.
\end{aligned}
\]
This leads to the final conclusion.
\end{proof}

\subsection{Proof of Proposition \ref{prop:pr2}}
\begin{proof}
For simplicity, we omit ``${(k)}$'' in the superscript in the following proof. Define $\bSigma := \EE \{\nabla ^ 2  \ell(\bbeta ^ *)\} = 2\|\bbeta ^ * \|_2 ^ 2 \bI_p + 4 \bbeta ^ * \bbeta ^ {*\top}$. We further define the following events:
\[
\begin{aligned}
	&\cH_1 := \bigg\{\bigg\|\frac{1}{n} \sum_{i=1}^{n} (\bx_i ^ \top \bbeta ^ *) ^ 2 \bx_i \bx_i ^ \top
	- \|\bbeta ^ *\| _ 2 ^ 2 \bI_p - 2 \bbeta ^ * \bbeta ^ {*\top}\bigg\|_2
	\leq \frac{1}{13} \|\bbeta ^ *\| _ 2 ^ 2\bigg\}, \\
	&\cH_2 := \bigg\{\frac{1}{n} \|\bX\|_{2 \rightarrow 4} ^ 4 \leq \frac{10}{3}\bigg\} 
		\cap \bigg\{\|\widehat \bbeta _ {\rm{init}} - \bbeta ^ *\|_2 \le \frac{1}{13} \|\bbeta ^ *\|_2\bigg\}, \\
	&\cH_3 := \bigg\{\bigg\|\frac{1}{n} \sum_{i=1}^{n} \varepsilon_i \bx_i \bx_i ^ \top \bigg\|_2 
	\le \frac{1}{13} \|\bbeta ^ *\| _ 2 ^ 2\bigg\}, 
\end{aligned}
\]
where we choose $n \geq C_1 p ^ 2$ for some sufficient large $C_1$ so that the events hold. Define the intersection of all the above events by $\cH$. By Lemmas \ref{PRlem6}, \ref{lem:pr_center_hess} and \ref{lem:pr_center_veps}, we have $\PP (\cH ^ c) \leq 38 n ^ {-2}$. From now on, we assume $\cH$ holds; we will take into account the failure probability of $\cH$ to that of the final conclusion in the end. Recall the definition of $\delta \ell (\bbeta ^ * + \bDelta; \bbeta ^ *)$ for any vector $\bDelta$ satisfied that $\|\bDelta\| _ 2 < \frac{1}{13} \|\bbeta ^ *\|_2$:
\[
	\delta \ell (\bbeta ^ * + \bDelta; \bbeta ^ *)
		= \ell (\bbeta ^ * + \bDelta) - \ell (\bbeta ^ *) - \nabla \ell (\bbeta ^ *) ^ \top \bDelta 
		= \int_0 ^ 1 \frac{1}{2} \bDelta ^ \top \nabla ^ 2 \ell (\bbeta ^ * + v \bDelta) \bDelta dv
\]
for some $v \in [0, 1]$ that depends on $\bDelta$. 
Then we have
\begin{equation}
\label{eq:pr_decmp1}
	\begin{aligned}
		\int_0 ^ 1 \frac{1}{2} \bDelta ^ \top \nabla ^ 2 &\ell (\bbeta ^ * + v \bDelta) \bDelta dv
		= \int_0 ^ 1 \frac{1}{2n} \sum_{i=1}^{n} \big[3\{\bx_i ^ \top (\bbeta ^ * + v \bDelta)\} ^ 2 - y_i\big] (\bx_i ^ \top \bDelta) ^ 2 dv \\
		&= \int_0 ^ 1  \frac{1}{2n} \sum_{i=1}^{n} \big\{2(\bx_i ^ \top \bbeta ^ *) ^ 2 (\bx_i ^ \top \bDelta) ^ 2
		+ 6v (\bx_i ^ \top \bbeta ^ *)(\bx_i ^ \top\bDelta) ^ 3
		+ 3v ^ 2(\bx_i ^ \top\bDelta) ^ 4 
		- \varepsilon_i (\bx_i ^ \top \bDelta) ^ 2\big\} dv \\
		&\ge \frac{1}{n} \sum_{i=1}^{n} (\bx_i ^ \top \bbeta ^ *) ^ 2 (\bx_i ^ \top \bDelta) ^ 2
		+ \int_0 ^ 1  \frac{3v}{n} \sum_{i=1}^{n} (\bx_i ^ \top \bbeta ^ *)(\bx_i ^ \top\bDelta) ^  3 dv
		- \frac{1}{2n} \sum_{i=1}^{n}\varepsilon_i (\bx_i ^ \top \bDelta) ^ 2.
	\end{aligned}
\end{equation}
By $\cH_1$, we have
\begin{equation}
\label{eq:pr_loc_1}
	\frac{1}{n} \sum_{i=1}^{n} (\bx_i ^ \top \bbeta ^ *) ^ 2 (\bx_i ^ \top \bDelta) ^ 2
	\ge 2 (\bbeta ^ {*\top} \bDelta) ^ 2 + \frac{12}{13} \|\bbeta ^ *\| _ 2 ^ 2\|\bDelta\| _ 2 ^ 2
	\ge \frac{12}{13} \|\bbeta ^ *\| _ 2 ^ 2\|\bDelta\| _ 2 ^ 2.
\end{equation}
By $\cH_2$ and the fact that $\|\bDelta\| _ 2 < \frac{1}{13} \|\bbeta ^ *\|_2$, we have
\[
\begin{aligned}
	\frac{3}{n} \sum_{i=1}^{n} (\bx_i ^ \top \bbeta ^ *) (\bx_i ^ \top\bDelta) ^ 3 
	&\le \frac{3}{n} \bigg\{\sum_{i=1}^{n} (\bx_i ^ \top \bbeta ^ *) ^ 4\bigg\} ^ {1/4} \bigg\{\sum_{i=1}^{n} (\bx_i ^ \top \bDelta) ^ 4\bigg\} ^ {3/4} \\
	&\le  \frac{3}{n} \|\bX\|_{2 \rightarrow 4} ^ 4 \|\bbeta ^ *\| _ 2 \|\bDelta\| _ 2 ^ 3 
	\le 10 \|\bbeta ^ *\| _ 2 \|\bDelta\| _ 2 ^ 3
	\le \frac{10}{13} \|\bbeta ^ *\| _ 2 ^ 2 \|\bDelta\| _ 2 ^ 2.
\end{aligned}
\]
Then we have
\begin{equation}
	\label{eq:pr_loc_2}
	\int_0 ^ 1  \frac{3v}{n} \sum_{i=1}^{n} (\bx_i ^ \top \bbeta ^ *)(\bx_i ^ \top\bDelta) ^  3 dv
	\leq \frac{10}{13} \|\bbeta ^ *\| _ 2 ^ 2 \|\bDelta\| _ 2 ^ 2.
\end{equation}
By $\cH_3$, we have
\begin{equation}
\label{eq:pr_loc_3}
	\bigg|\frac{1}{2n} \sum_{i=1}^{n} \varepsilon_i (\bx_i ^ \top \bDelta) ^ 2\bigg| 
	\le \frac{1}{26} \|\bbeta ^ *\| _ 2 ^ 2 \|\bDelta\| _ 2 ^ 2.
\end{equation}
Combining the bounds \eqref{eq:pr_loc_1}, \eqref{eq:pr_loc_2} and \eqref{eq:pr_loc_3} delivers that
\[
\begin{aligned}
	\int_0 ^ 1 \frac{1}{2} \bDelta ^ \top \nabla ^ 2 \ell (\bbeta ^ * + v \bDelta) \bDelta dv  
	&\ge \frac{12}{13} \|\bbeta ^ *\| _ 2 ^ 2\|\bDelta\| _ 2 ^ 2 - \frac{10}{13} \|\bbeta ^ *\| _ 2 ^ 2\|\bDelta\| _ 2 ^ 2 - \frac{1}{26} \|\bbeta ^ *\| _ 2 ^ 2\|\bDelta\| _ 2 ^ 2 \\
	&\ge \frac{1}{13} \|\bbeta ^ *\| _ 2 ^ 2\|\bDelta\| _ 2 ^ 2.
\end{aligned}
\]
Set $\bDelta = \widehat \bbeta - \bbeta ^ *$. Given that $\ell(\widehat \bbeta) \le \ell(\bbeta ^ *)$, we have 
\[
	\frac{1}{13} \|\bbeta ^ *\| _ 2 ^ 2\|\bDelta\| _ 2 ^ 2 
	\leq \delta \ell (\bbeta ^ * + \bDelta; \bbeta ^ *)
	\leq - \nabla \ell (\bbeta ^ *) ^ \top \bDelta 
	\leq \|\nabla \ell (\bbeta ^ *)\|_2 \|\bDelta\|_2,
\]
which further implies that
\[
	\|\bDelta\| _ 2 \leq 13 \|\bbeta ^ *\| _ 2 ^ {-2} \|\nabla \ell (\bbeta ^ *)\|_2.
\]
Now we derive the rate of $\|\nabla \ell (\bbeta ^ *)\|_2$. By Lemma \ref{lem:pr_center_veps}, we have with probability at least $1 - 4e ^ {-\xi}$ that
\begin{equation}
	\|\nabla \ell(\bbeta ^ *)\|_2 
	= \bigg\|\frac{1}{n} \sum_{i=1}^{n} \varepsilon_i (\bx_i ^ \top \bbeta ^ *) \bx_i\bigg\|_2 
	\lesssim \|\bbeta ^ *\|_2 \bigg(\frac{p \vee \xi}{n}\bigg) ^ {1/2}.
\end{equation}
Finally, considering the failure probability of $\cH$, we can obtain that with probability at least $1 - 38 n ^ {-2} - 4e ^ {-\xi}$ that
\[
	\|\widehat \bbeta - \bbeta ^ *\| _ 2 \lesssim  \|\bbeta ^ *\|_2  ^ {-1}\bigg(\frac{p \vee \xi}{n}\bigg) ^ {1/2}.
\]
In the following, we derive the moment bound $\EE \|\widehat \bbeta - \bbeta ^ *\|_2 ^ 2$. 
\[
	\begin{aligned}
		\EE \|\widehat \bbeta - \bbeta ^ *\|_2 ^ 2 
		&= \EE (\mathbbm{1}_{\cH}\|\widehat \bbeta - \bbeta ^ *\|_2 ^ 2) + \EE (\mathbbm{1}_{\cH ^ c}\|\widehat \bbeta - \bbeta ^ *\|_2 ^ 2) \\
		&\lesssim \EE \big\{\|\bbeta ^ *\|_2 ^ {-2} \|\nabla \ell (\bbeta ^ *)\|_2\big\} ^ 2 + \|\bbeta ^ *\|_2 ^ 2 \PP(\cH ^ c) 
		\lesssim \|\bbeta ^ *\| _ 2 ^ {-2} \frac{p}{n}.
	\end{aligned}
\]

\end{proof}

\subsection{Proof of Theorem \ref{thm:reboot_phase_retrieval}}
Recall that $\EE \big(\cdot \big|\, \widehat \bbeta ^ {(k)}\big)$ is the conditional expectation given $\widehat \bbeta ^ {(k)}$. By the definition, we have
\[
\begin{aligned}
	\nabla ^ 2 \widetilde \ell (\bbeta)
	&= \frac{1}{m} \sum_{k = 1} ^ m \EE \Big\{ \nabla ^ 2 \ell \big(\bbeta \, ;\,  \widetilde \bz ^ {(k)}\big) \big|\, \widehat \bbeta ^ {(k)}\Big\} 
	= \frac{1}{m} \sum_{k = 1} ^ m \EE \Big\{ 3 \big(\widetilde \bx_i ^ {(k)\top} \bbeta\big) ^ 2 \widetilde \bx_i ^ {(k)} \widetilde \bx_i ^ {(k)\top} - \big(\widetilde \bx_i ^ {(k)\top} \widehat \bbeta ^ {(k)} \big) \widetilde \bx_i ^ {(k)} \widetilde \bx_i ^ {(k)\top} \big|\, \widehat \bbeta ^ {(k)}\Big\} \\
	&= 3\|\bbeta\|_2 ^ 2 \bI_p + 6 \bbeta \bbeta ^ \top 
	- \|\bbeta ^ {(k)}\|_2 ^ 2 \bI_p - 2\bbeta ^ {(k)} \bbeta ^ {(k) \top}.
\end{aligned}
\]
By the definition of $\delta \widetilde \ell (\bbeta ^ * + \bDelta; \bbeta ^ *)$ for any vector $\bDelta$ satisfied that $\|\bDelta\| _ 2 < \frac{1}{13} \|\bbeta ^ *\|_2$, we have
\[
	\delta \widetilde \ell (\bbeta ^ * + \bDelta; \bbeta ^ *)
		= \widetilde \ell (\bbeta ^ * + \bDelta) - \widetilde \ell (\bbeta ^ *) - \nabla \widetilde \ell (\bbeta ^ *) ^ \top \bDelta 
		= \frac{1}{2} \int_0 ^ 1 \bDelta ^ \top \nabla ^ 2 \widetilde \ell (\bbeta ^ * + v \bDelta) \bDelta dv.
\]
Note that
\[
\begin{aligned}
	\int_0 ^ 1 \bDelta ^ \top &\nabla ^ 2 \widetilde \ell (\bbeta ^ * + v \bDelta) \bDelta dv \\
	&= \int_0 ^ 1 \Big\{3\|\bbeta ^ * + v \bDelta\|_2 ^ 2 \|\bDelta\|_2 ^ 2 + 6 (\bbeta ^ {*\top} \bDelta + v \|\bDelta\|_2 ^ 2) ^ 2
	- \|\bbeta ^ *\|_2 ^ 2 \|\bDelta\|_2 ^ 2 - 2 (\bbeta ^ {*\top} \bDelta) ^ 2 \Big\} dv \\
	&\ge \int_0 ^ 1 \Big\{3\|\bbeta ^ *\|_2 ^ 2 \|\bDelta\|_2 ^ 2 + 6 v (\bbeta ^ {*\top} \bDelta) \|\bDelta\|_2 ^ 2 + 12 v (\bbeta ^ {*\top} \bDelta) \|\bDelta\|_2 ^ 2 \Big\} dv.
\end{aligned}
\]
Since $\bbeta ^ {*\top} \bDelta \leq \|\bbeta ^ *\|_2 \|\bDelta\|_2 \leq \frac{1}{13} \|\bbeta ^ *\|_2 ^ 2$, it follows that
\[
\begin{aligned}
	\int_0 ^ 1 \bDelta ^ \top &\nabla ^ 2 \widetilde \ell (\bbeta ^ * + v \bDelta) \bDelta dv 
	\ge 3\|\bbeta ^ *\|_2 ^ 2 \|\bDelta\|_2 ^ 2 -  \frac{18}{13} \|\bbeta ^ *\|_2 ^ 2 \|\bDelta\|_2 ^ 2 
	\ge \|\bbeta ^ *\|_2 ^ 2 \|\bDelta\|_2 ^ 2.
\end{aligned}
\]
Set $\bDelta = \widehat \bbeta ^ \rb - \bbeta ^ *$. Given that $\widetilde \ell(\widehat \bbeta ^ \rb) \le \widetilde \ell(\bbeta ^ *)$, we have
\[
	\frac{1}{2} \|\bbeta ^ *\|_2 ^ 2\|\bDelta\| _ 2 ^ 2 
	\leq \delta \widetilde \ell (\bbeta ^ * + \bDelta; \bbeta ^ *)
	\leq - \nabla \widetilde \ell (\bbeta ^ *) ^ \top \bDelta 
	\leq \|\nabla \widetilde \ell (\bbeta ^ *)\|_2 \| \bDelta\|_2,
\]
which further implies that
\begin{equation}
\label{eq:pr_grad}
	\|\bDelta\|_2 \leq 2 \|\bbeta ^ *\|_2 ^ {-2} \|\nabla \widetilde \ell (\bbeta ^ *)\|_2 .
\end{equation}
Now we analyze the rate of $\|\nabla \widetilde \ell(\bbeta ^ *)\|_2$. For convenience, we write $\widehat \bbeta ^ {(k)} - \bbeta ^ *$ as $\bDelta ^ {(k)}$ for $k \in [m]$. Applying the third-order Taylor expansion of $\nabla \widetilde\ell ^ {(k)}(\bbeta)$ at $\bbeta ^ *$ yields that 
\begin{equation} 
\label{eq:pr_subtit1}
	\begin{aligned}
		\nabla \widetilde \ell ^ {(k)} (\widehat \bbeta ^ {(k)}) 
		&= \nabla \widetilde \ell ^ {(k)} (\bbeta ^ *) + \bSigma \bDelta ^ {(k)}  
  		+ \widetilde \bR_3 ^ {(k)} (\bDelta ^ {(k)} \otimes \bDelta ^ {(k)}) = 0,
	\end{aligned}
\end{equation}
where $\widetilde \bR_3 ^ {(k)} = \EE \big\{\int_0 ^ 1 \nabla ^ 3 \widetilde \ell ^ {(k)} \big(\bbeta ^ * + v(\widehat \bbeta ^{(k)} - \bbeta ^ *)\big) dv \big|\, \widehat \bbeta ^{(k)}\big\}$. Similarly, by the third-order Taylor expansion and the fact that $\nabla \ell ^ {(k)} (\widehat \bbeta ^ {(k)}) = 0$ , we have
\[
		\nabla \ell ^ {(k)} (\widehat \bbeta ^ {(k)})
		= \nabla \ell ^ {(k)} (\bbeta ^ *) + \bSigma \bDelta ^ {(k)} + (\nabla ^ 2 \ell ^ {(k)} (\bbeta ^ *) - \bSigma) \bDelta ^ {(k)} 
  		+ \bR_3 ^ {(k)} (\bDelta ^ {(k)} \otimes \bDelta ^ {(k)} ) 
  		= 0, 
\]
where $\bR_3 ^ {(k)} = \int_0 ^ 1 \nabla ^ 3 \ell ^ {(k)} (\bbeta ^ * + v(\widehat \bbeta ^{(k)} - \bbeta ^ *)) dv$. This implies that
\begin{equation}
\label{eq:pr_subtit2}
        \bSigma  \bDelta ^ {(k)}
        = - \nabla \ell ^ {(k)} (\bbeta ^ *) - (\nabla ^ 2 \ell ^ {(k)} (\bbeta ^ *) - \bSigma) \bDelta ^ {(k)} 
        	-  \bR_3 ^ {(k)} (\bDelta ^ {(k)} \otimes \bDelta ^ {(k)}).
\end{equation}
Substituting \eqref{eq:pr_subtit2} into \eqref{eq:pr_subtit1}, we have
\begin{equation}
	\begin{aligned}
		0
		&= \nabla \widetilde \ell ^ {(k)} (\bbeta ^ *) - \nabla \ell ^ {(k)} (\bbeta ^ *) 
		- (\nabla ^ 2 \ell ^ {(k)} (\bbeta ^ *) - \bSigma) \bDelta ^ {(k)}
		- \bR_3 ^ {(k)} (\bDelta ^ {(k)} \otimes \bDelta ^ {(k)})
		+ \widetilde \bR_3 ^ {(k)} (\bDelta ^ {(k)} \otimes \bDelta ^ {(k)}) .
	\end{aligned}
\end{equation}
Note that $\nabla \widetilde \ell(\bbeta ^ *) = \frac{1}{m} \sum_{k=1}^m \nabla \widetilde \ell ^ {(k)}(\bbeta ^ *)$. Then we have the following decomposition: 
\[
	\begin{aligned}
		\nabla \widetilde \ell(\bbeta ^ *) 
		= &\underbrace{\nabla \ell (\bbeta ^ *)}_{T_1}
		+ \underbrace{\frac{1}{m} \sum_{k=1} ^ m (\nabla ^ 2 \ell ^ {(k)} (\bbeta ^ *) - \bSigma) \bDelta ^ {(k)}}_{T_2} 
		+ \underbrace{\frac{1}{m} \sum_{k=1} ^ m \bR_3 ^ {(k)} (\bDelta ^ {(k)} \otimes \bDelta ^ {(k)})}_{T_3}
		- \underbrace{\frac{1}{m} \sum_{k=1} ^ m \widetilde \bR_3 ^ {(k)} (\bDelta ^ {(k)} \otimes \bDelta ^ {(k)})}_{T_4}.
	\end{aligned}
\]
We introduce the following events: 
\[
\begin{aligned}
	\cE ^ {(k)} :=
		& \bigg\{\big\|\widehat \bbeta ^ {(k)} - \bbeta ^ *\big\|_2 
		\le \frac{C_1}{\ltwonorm{\bbeta ^ *}} \bigg(\frac{p \vee \xi}{n}\bigg) ^ {1/2}\bigg\}\cap \\
		&\bigg\{\|\nabla ^ 2 \ell ^ {(k)} (\bbeta ^ *) - \bSigma\|_2 
		\le C_2\|\bbeta\|_2 ^ 2 \bigg(\frac{p \vee \xi}{n}\bigg) ^ {1/2} 
			+ C_2\|\bbeta\|_2 ^ 2 \frac{(p \vee \xi) \xi}{n}\bigg\} \cap \\
		&\Big\{\big\|\bR_3 ^ {(k)}\big\|_2 
		\le C_3 \|\bbeta ^ *\|_2 \Big\} \cap
		\Big\{\big\|\widetilde \bR_3 ^ {(k)}\big\|_2 
		\le C_3 \|\bbeta ^ *\|_2 \Big\},
\end{aligned}
\]
where $C_1, C_2, C_3$ are constants. Let $\cE := \cap_{i=1} ^ m \cE ^ {(k)}$. By Proposition \ref{prop:pr2}, Lemmas \ref{PRlem6}, \ref{lem:pr_center_hess} and \ref{lem:pr_center_veps}, we have 
\begin{equation}
\label{eq:pr_prob_A}
	\PP (\cA^c) 
	\leq 42 m n ^ {-2} + 4 m e ^ {-\xi}.
\end{equation}
From now on, we assume $\cE$ holds; we will take into account the failure probability of $\cE$ to that of the final conclusion in the end. Then we consider the bounds for $T_1, T_2, T_3$ and $T_4$ conditional on events $\cE$. We categorize these terms into variance and bias terms: $T_1$ is the variance term, and $T_2, T_3, T_4$ are the bias terms.

\paragraph*{\bf \underline{Variance terms}}
\paragraph*{Bound for $T_1$}
By Lemma \ref{lem:pr_center_veps}, we have with probability at least $1 -4 e ^ {-\xi}$ that
\[
\begin{aligned}
	\|T_1\| = \bigg\|\frac{1}{mn} \sum_{k=1}^{m} \sum_{i=1}^{n} \varepsilon_i ^ {(k)} (\bx_i ^ {(k)\top} \bbeta ^ *) \bx_i ^ {(k)} \bigg\|_2 
	\lesssim \|\bbeta ^ *\|_2 \bigg(\frac{p \vee \xi}{mn}\bigg) ^ {1/2}.
\end{aligned}
\]

\paragraph*{\bf \underline{Bias terms}}
\paragraph*{Bound for $T_2$}
By the triangle inequality, we have
\[
\begin{aligned}
	\bigg\|\frac{1}{m} \sum_{k=1} ^ m (\nabla ^ 2 \ell ^ {(k)} (\bbeta ^ *) - \bSigma) \bDelta ^ {(k)}\bigg\|_2
	\leq & \bigg\|\frac{1}{m} \sum_{k=1} ^ m (\nabla ^ 2 \ell ^ {(k)} (\bbeta ^ *) - \bSigma) \bDelta ^ {(k)} - \EE (\nabla ^ 2 \ell ^ {(k)} (\bbeta ^ *) - \bSigma) \bDelta ^ {(k)} \bigg\|_2\\
	&+ \bigg\| \frac{1}{m} \sum_{k=1} ^ m \EE (\nabla ^ 2 \ell ^ {(k)} (\bbeta ^ *) - \bSigma) \bDelta ^ {(k)}\bigg\|_2.
\end{aligned}
\]
Similarly, by Hoeffding's inequality, we have with probability at least $1-2 e^{-\xi}$ that
\[
\begin{aligned}
	\bigg\|\frac{1}{m} \sum_{k=1} ^ m (\nabla ^ 2 \ell ^ {(k)} (\bbeta ^ *) - \bSigma) \bDelta ^ {(k)}
		- \EE (\nabla ^ 2 \ell ^ {(k)} (\bbeta ^ *) - \bSigma) \bDelta ^ {(k)} \bigg\|_2  
		\lesssim \|\bbeta ^ *\|_2 \frac{p \vee \xi}{(mn) ^ {1/2}} \bigg\{\bigg(\frac{p \vee \xi}{ n}\bigg) ^ {1/2} + \frac{(p \vee \xi) \xi}{ n}\bigg\}.
\end{aligned}
\]
Let $\xi = 2\log n$. Suppose $n \ge C_7 \log ^ 5 n$ for some constant $C_7$. Consider the bound in two cases: $p \le \log^2 n$ and $p > \log^2 n$. In the first case $p \le \log ^ 2 n$:
\[
	\frac{p \vee \xi}{(mn) ^ {1/2}} \bigg\{\bigg(\frac{p \vee \xi}{ n}\bigg) ^ {1/2} + \frac{(p \vee \xi) \xi}{ n}\bigg\}
	\lesssim \bigg(\frac{1}{mn} \times\frac{\log ^ 5n}{n}\bigg) ^ {1/2}
	\lesssim \bigg(\frac{1}{mn}\bigg) ^ {1/2}.
\]
In the second case $p > \log^2 n$:
\[
	\frac{p \vee \xi}{(mn) ^ {1/2}} \bigg\{\bigg(\frac{p \vee \xi}{ n}\bigg) ^ {1/2} + \frac{(p \vee \xi) \xi}{ n}\bigg\}
	\lesssim\bigg(\frac{1}{mn}\bigg) ^ {1/2} \frac{p^2 \log n}{n}
	\lesssim \bigg(\frac{p}{mn}\bigg) ^ {1/2}.
\]
Therefore, we have with probability at least $1-2 n ^ {-2}$ that
\begin{equation}
		\bigg\|\frac{1}{m} \sum_{k=1} ^ m (\nabla ^ 2 \ell ^ {(k)} (\bbeta ^ *) - \bSigma) \bDelta ^ {(k)}
		- \EE (\nabla ^ 2 \ell ^ {(k)} (\bbeta ^ *) - \bSigma) \bDelta ^ {(k)} \bigg\|_2 
		\lesssim \bigg(\frac{p}{mn}\bigg) ^ {1/2}.
\end{equation}
By the Cauchy-Schwarz inequality, we have
\[
	\begin{aligned}
		\big\|\EE (\nabla ^ 2 \ell ^ {(1)} (\bbeta ^ *) - \bSigma) \bDelta ^ {(1)} \big\|_2
		\leq \{\EE \|\nabla ^ 2 \ell ^ {(1)} (\bbeta ^ *) - \bSigma\|_2 ^ 2 \EE \| \bDelta ^ {(1)}\|_2 ^ 2\} ^ {1/2} 
		\lesssim \|\bbeta ^ *\|_2  \frac{p}{n}.
	\end{aligned}
\]
Then we have with probability at least $1 - 2n ^ {-2}$ that
\[
	\bigg\|\frac{1}{m} \sum_{k=1} ^ m (\nabla ^ 2 \ell ^ {(k)} (\bbeta ^ *) - \bSigma) \bDelta ^ {(k)}\bigg\|_2
	\lesssim \|\bbeta ^ *\|_2 \bigg\{\bigg(\frac{p}{mn}\bigg) ^ {1/2} + \frac{p}{n}\bigg\}.
\]

\paragraph*{Bound for $T_3$}
By $\cE$ and $\cF$, we have
\[
	\bigg\|\frac{1}{m} \sum_{k=1} ^ m \bR_3 ^ {(k)} (\bDelta ^ {(k)} \otimes \bDelta ^ {(k)})\bigg\|_2
	\le \frac{1}{m} \sum_{k=1} ^ m \big\|\bR_3 ^ {(k)}\big\|_2 \big\|\bDelta ^ {(k)}\big\|_2 ^ 2 
	\lesssim \frac{1}{\|\bbeta ^ *\|_2} \frac{p \vee \xi}{n}.
\]

\paragraph*{Bound for $T_4$}
Similar to $T_3$, we have
\[
	\bigg\|\frac{1}{m} \sum_{k=1} ^ m \widetilde \bR_3 ^ {(k)} (\bDelta ^ {(k)} \otimes \bDelta ^ {(k)})\bigg\|_2
	\le \frac{1}{m} \sum_{k=1} ^ m \big\|\widetilde \bR_3 ^ {(k)}\big\|_2 \big\|\bDelta ^ {(k)}\big\|_2 ^ 2 
	\lesssim \frac{1}{\|\bbeta ^ *\|_2} \frac{p \vee \xi}{n}.
\]
Let $\xi = 2 \log n$. Combining the above bounds and considering the failure probability of $\cA$, we can obtain that with probability at least $1 - (46m + 6) n ^ {-2}$ that
\[
	\big\|\widehat \bbeta ^ \rb - \bbeta ^ *\big\|_2
		\lesssim \frac{1}{\|\bbeta ^ *\|_2} \bigg(\frac{p}{mn}\bigg) ^ {1/2} + 
		\bigg(\frac{1}{\|\bbeta ^ *\|_2} + \frac{1}{\|\bbeta ^ *\|_2 ^ 3}\bigg) \bigg(\frac{p \vee \log n}{n}\bigg).
\]

\section{Proof of technical lemmas}
\subsection{Lemmas for the generalized linear models}
\label{appendix:B.1}

In this section, we provide the proofs of some technical lemmas in the generalized linear models. 

\begin{lem}
\label{lem:two_delta}
Under Conditions \ref{con:distribution} and \ref{con:b_four_prime}, if $n \ge \max(p, 4\log n)$, then we have that with probability at least $1 - 6 n ^ {-4}$ such that
\[
    \bDelta ^ {(k)} \otimes \bDelta ^ {(k)} - (\bSigma ^ {-1} \nabla \ell ^ {(k)}(\bbeta ^ *)) \otimes (\bSigma^{-1} \nabla \ell ^ {(k)}(\bbeta ^ *)) 
	\lesssim C_{\kappa, \phi, M, K}' \|\bSigma ^ {-1}\|_2 ^ 2 \bigg(\frac{p \vee 4 \log n}{n}\bigg) ^ {3 / 2},
\]
and
\[
\begin{aligned}
    \bDelta ^ {(k)} \otimes \bDelta ^ {(k)} \otimes \bDelta ^ {(k)} 
    - (\bSigma ^ {-1} \nabla \ell ^ {(k)}(\bbeta ^ *)) \otimes (\bSigma ^ {-1} \nabla \ell ^ {(k)}&(\bbeta ^ *)) \otimes (\bSigma ^ {-1} \nabla \ell ^ {(k)} (\bbeta ^ *)) \\
	&\lesssim C_{\kappa, \phi, M, K}'' \|\bSigma ^ {-1}\|_2 ^ 3 \bigg(\frac{p \vee 4 \log n}{n}\bigg) ^ 2.
\end{aligned}
\]
for any $k \in [m]$, where $C_{\kappa, \phi, M, K}' = \kappa ^ {-2} \phi ^ {3 / 2} M ^ {5 / 2} K ^ 6 + \kappa ^ {-1} \phi M ^ 2 K ^ 4$ and $C_{\kappa, \phi, M, K}'' = \kappa ^ {-2} \phi ^ 2 M ^ 3 K ^ 7  + \kappa ^ {-1} \phi ^ {3 / 2} M ^ {5 / 2} K ^ 5$.
\end{lem}

\begin{proof} 
For simplicity, we omit ``${(k)}$'' in the superscript in the following proof. \eqref{eq:lem1_1} yields that
\[
\begin{aligned}
	\bDelta  \otimes \bDelta
    = &(\bSigma ^ {-1} \nabla \ell(\bbeta ^ *)) \otimes (\bSigma ^ {-1} \nabla \ell(\bbeta ^ *)) 
    + (\bSigma ^ {-1} \nabla \ell (\bbeta ^ *)) \otimes \{\bSigma ^ {-1} (\bR_2 - \bSigma) \bDelta\} \\
    & + \{\bSigma ^ {-1} (\bR_2 - \bSigma) \bDelta\} \otimes (\bSigma ^ {-1} \nabla \ell(\bbeta ^ *))
    + \{\bSigma ^ {-1} (\bR_2 - \bSigma) \bDelta\} \otimes \{\bSigma ^ {-1} (\bR_2 - \bSigma) \bDelta\}.
\end{aligned} 
\]
We can bound each of the last three outer products similarly. Therefore, we focus on the second one for simplicity. By display (\ref{ineq:grad_upper_bound}), Lemma \ref{prop:glm_local_mse} and Lemma \ref{lem:center_hessian}, we have with probability at least $1 - 6 n ^ {-4}$ that
\[
\begin{aligned}
	\|(\bSigma ^ {-1} \nabla \ell (\bbeta ^ *)) \otimes \{\bSigma ^ {-1} (\bR_2 - \bSigma) \bDelta\}\|_2
	&\leq \|\bSigma ^ {-1}\|_2 ^ 2 \|\nabla \ell (\bbeta ^ *)\|_2 \|\bR_2 - \bSigma\|_2 \|\bDelta\| _ 2 \\
	&\lesssim \kappa ^ {-1} \phi M K ^ 2 \{\kappa ^ {-1} \phi ^ {1 / 2} M ^ {3 / 2} K ^ 4  + M K ^ 2\} \|\bSigma ^ {-1}\|_2 ^ 2 \bigg(\frac{p \vee 4 \log n}{n}\bigg) ^ {3 / 2}.
\end{aligned}
\]
Therefore, we have with probability at least $1 - 6 n ^ {-4}$ such that
\[
\begin{aligned}
    \bDelta \otimes \bDelta - (\bSigma ^ {-1} \nabla \ell(\bbeta ^ *)) \otimes (\bSigma^{-1} \nabla \ell (\bbeta ^ *))
	\lesssim \kappa ^ {-1} \phi M K ^ 2 \{\kappa ^ {-1} \phi^ {1 / 2} M ^ {3 / 2} K ^ 4 + M K ^ 2\} \|\bSigma ^ {-1}\|_2 ^ 2 \bigg(\frac{p \vee 4 \log n}{n}\bigg) ^ {3 / 2}.
\end{aligned}
\]
Similarly, it holds with probability at least $1 - 6 n ^ {-4}$ that 
\[
\begin{aligned}
    \bDelta \otimes \bDelta \otimes \bDelta
    &- (\bSigma ^ {-1} \nabla \ell (\bbeta ^ *)) \otimes (\bSigma^{-1} \nabla \ell (\bbeta ^ *)) \otimes (\bSigma^{-1} \nabla \ell (\bbeta ^ *)) \\
	&\quad\quad\quad\quad\quad\quad\quad\quad \lesssim \kappa ^ {-1} (\phi M) ^ {3 / 2} K ^ 3 \{\kappa ^ {-1} \phi ^ {1 / 2} M ^ {3 / 2} K ^ 4  + M K ^ 2\} \|\bSigma ^ {-1}\|_2 ^ 3 \bigg(\frac{p \vee 4 \log n}{n}\bigg) ^ 2.
\end{aligned}
\]
\end{proof}

\begin{lem} 
\label{lem:center_Sigma} 
Under Conditions \ref{con:distribution} and \ref{con:b_double_prime}, for any $\xi > 0$, we have with probability at least $1 - 2 e ^ {-\xi}$ that
\[
	\|\nabla ^ 2 \ell ^ {(k)} (\bbeta) - \EE \{\nabla ^ 2 \ell ^ {(k)} (\bbeta)\}\|_2 
	\lesssim M K ^ 2 \bigg\{ \bigg(\frac{p \vee \xi}{n}\bigg)^{1/2} + \frac{p \vee \xi}{n}\bigg\},
\]
for any $\bbeta \in \RR ^ p$ and any $k \in [m]$.
\end{lem}

\begin{proof}
For simplicity, we omit ``${(k)}$'' in the superscript in the following proof. Note that
\[
\begin{aligned}
	\|\nabla ^ 2 \ell (\bbeta) - \EE \{\nabla ^ 2 \ell (\bbeta)\}\|_2 
	&= \max_{\|\bu\|_2 = 1} \big<(\nabla ^ 2 \ell (\bbeta) - \EE \{\nabla ^ 2 \ell (\bbeta)\}) \bu, \bu\big> \\
	&\leq 2 \max_{\bu \in \cN(1/4)} \big<(\nabla ^ 2 \ell (\bbeta) - \EE \{\nabla ^ 2 \ell (\bbeta)\}) \bu, \bu\big> \\
	&\leq 2 \max_{\bu \in \cN(1/4)} \bigg|\frac{1}{n} \sum_{i=1}^n b''(\bx_i \bbeta) (\bx_i ^ \top \bu) ^ 2 - \EE \{b''(\bx_i \bbeta)(\bx_i ^ \top \bu) ^ 2\}\bigg|.
\end{aligned}
\]
Lemma 2.7.6 in \cite{vershynin2010introduction} implies that
\[
\begin{aligned}
	\|b''(\bx_i \bbeta)(\bx_i ^ \top \bu) ^ 2 - \EE \{b''(\bx_i \bbeta)(\bx_i ^ \top \bu) ^ 2\} \|_{\psi_1} 
	&\leq C_1 M \|(\bx_i ^ \top \bu) ^ 2\|_{\psi_1} \\
	&\leq C_1 M \|\bx_i ^ \top \bu\|_{\psi_2} ^ 2 
	\leq C_1 M K ^ 2,
\end{aligned}
\]
where $C_1$ is a constant. By Bernstein's inequality, we have
\[
\begin{aligned}
	&\PP \bigg[ \max_{\bu \in \cN(1/4)} \bigg|\frac{1}{n} \sum_{i=1}^n b''(\bx_i \bbeta) (\bx_i ^ \top \bu) ^ 2 - \EE \{b''(\bx_i \bbeta)(\bx_i ^ \top \bu) ^ 2\}\bigg|
	\geq C_2 M K ^ 2 \max \bigg\{\bigg(\frac{t}{n}\bigg) ^ {1/2}, \frac{t}{n}\bigg\}\bigg] \\
	&\leq \sum_{\bu \in \cN(1/4)} \PP \bigg[\bigg|\frac{1}{n} \sum_{i=1}^n b''(\bx_i \bbeta) (\bx_i ^ \top \bu) ^ 2 - \EE \{b''(\bx_i \bbeta)(\bx_i ^ \top \bu) ^ 2\}\bigg| \geq C_2 M K ^ 2 \max \bigg\{\bigg(\frac{t}{n}\bigg) ^ {1/2}, \frac{t}{n}\bigg\}\bigg] \\
	&\leq 2 e ^ {-(t - 9 \log p)},
\end{aligned}
\]
where $C_2$ is a constant. Substituting $\xi = t - 9 \log p$ into the bound with positive $\xi$, we can find a universal constant $C_3$ with probability at least $1 - 2e ^ {- \xi}$ that 
\[
	\|\nabla ^ 2 \ell (\bbeta) - \EE \{\nabla ^ 2 \ell (\bbeta)\}\|_2 \leq C_3 M K ^ 2 \bigg\{ \bigg(\frac{p \vee \xi}{n}\bigg)^{1/2} + \frac{p \vee \xi}{n}\bigg\}.
\]
\end{proof}

\begin{lem}
\label{lem:b_three_prime}
Under Conditions \ref{con:distribution} and \ref{con:b_four_prime}, for any $\xi > 0$, we have with probability at least $1 - 2 e ^ {- \xi}$  that
\[
	\big\|\nabla ^ 3 \ell ^ {(k)} (\bbeta) - \EE \{\nabla ^ 3 \ell ^ {(k)} (\bbeta)\}\big\|_2
	\lesssim M K ^ 3 \bigg\{\bigg(\frac{p \vee \xi}{n}\bigg) ^ {1/2} + \frac{(p \vee \xi) ^ {3/2}}{n}\bigg\}.
\]
for any $\bbeta \in \RR ^ p$ and any $k \in [m]$,
\end{lem}

\begin{proof}
For simplicity, we omit ``${(k)}$'' in the superscript in the following proof. We introduce the shorthand $\bQ := \nabla ^ 3 \ell (\bbeta) - \EE\{\nabla ^ 3 \ell(\bbeta)\}$. Note that for any $\bbeta \in \RR ^ p$, 
\[
	\|\bQ\|_2 
		= \sup_{\|\bu\|_2 = 1} \|\bQ (\bu \otimes \bu )\|_2
		= \sup_{\|\bu\|_2 = 1} |\bQ (\bu \otimes \bu \otimes \bu)|.
\]
By Corollary 4.2.13 in \cite{vershynin2010introduction}, we can construct a $\frac{1}{8}$-net of the sphere $\cS ^ {p - 1}$ with cardinality $|\cN(1/8)| \leq 17 ^ p$. Therefore, given any $\bu \in \cS ^ {p -1}$, we can write $\bu = \bv + \bfsym \delta$ for some $\bv$ in $\cN(1 / 8)$ and an error vector $\bfsym \delta$ such that $\|\bfsym \delta\|_2 \leq 1/8$. Then we have
\[
	|\bQ (\bu \otimes \bu \otimes \bu)|
	= |\bQ (\bv \otimes \bv \otimes \bv) + 3\bQ (\bv \otimes \bv \otimes \bfsym \delta) + 3\bQ (\bv \otimes \bfsym \delta \otimes \bfsym \delta) + \bQ (\bfsym \delta \otimes \bfsym \delta \otimes \bfsym \delta)|.
\]
By the triangle inequality and H\"older's inequality, we have
\[
\begin{aligned}
	|\bQ (\bu \otimes \bu \otimes \bu)|
	&\leq |\bQ (\bv \otimes \bv \otimes \bv)| + 3 \|\bQ \|_2 \|\bv\|_2 ^ 2 \|\bfsym \delta\|_2 + 3 \|\bQ\|_2 \|\bv\|_2 \|\bfsym \delta\|_2 ^ 2 + \|\bQ\|_2 \|\bfsym \delta\|_2 ^ 3 \\
	&\leq |\bQ (\bv \otimes \bv \otimes \bv)| + \bigg(\frac{3}{8} + \frac{3}{8 ^ 2} + \frac{1}{8 ^ 3}\bigg) \|\bQ\|_2 \\
	&\leq |\bQ (\bv \otimes \bv \otimes \bv)| + \frac{1}{2} \|\bQ\|_2.
\end{aligned}
\]
Taking the supremum over $\bu \in \cS ^ {p -1}$ and $\bv \in \cN(1 / 8)$ yields 
\begin{equation}
\label{eq:b_three_prime2}
	\|\bQ\|_2 
	\leq 2 \!\!\!\!\!\! \sup_{\bv \in \cN(1/8)} \!\!\!\!\!\! |\bQ (\bv \otimes \bv \otimes \bv)| 
	\leq 2 \!\!\!\!\!\! \sup_{\bv \in \cN(1/8)} \! \bigg|\frac{1}{n} \sum_{i=1}^n b'''(\bx_i ^ \top \bbeta) (\bx_i ^ \top \bv)^ 3 - \EE \{b'''(\bx_i ^ \top \bbeta) (\bx_i ^ \top \bv) ^ 3\} \bigg|.
\end{equation}
Given any $\bv$, by sub-Gaussianity, we have
\[
    \EE \exp \bigg\{\bigg(\frac{|b'''(\bx_i ^ \top \bbeta) (\bx_i ^ \top \bv) ^ 3|}{M K ^ 3}\bigg) ^ {2 / 3}\bigg\} \leq 2, 
\]
from which we deduce that $\|b'''(\bx_i ^ \top \bbeta) (\bx_i ^ \top \bv) ^ 3\| _ {\psi_{2 / 3}} \leq M K ^ 3$. Applying the bound above display (3.6) in \cite{Adamczak2009RestrictedIP}, we have for any $t > 0$ that 
\[
\begin{aligned}
	\PP \bigg[\bigg|\frac{1}{n} \sum_{i=1}^n b'''(\bx_i ^ \top \bbeta) (\bx_i ^ \top \bv)^ 3 - \EE \{b'''(\bx_i ^ \top \bbeta) (\bx_i ^ \top \bv) ^ 3\} \bigg| & \geq C_1 M K ^ 3 \bigg\{\bigg(\frac{t}{n}\bigg) ^ {\!\!1 / 2} \!\!\!\!+ \frac{t ^ {3 / 2}}{n}\bigg\}\bigg] 
    \leq 2e ^ {-(t - 3)}, 
\end{aligned}
\]
where $C_1$ is a universal constant. Applying a union bound over $\bv \in \cN(1 / 8)$ and then \eqref{eq:b_three_prime2}, we deduce that 
\[	
\begin{aligned}
	&\PP \bigg[\big\| \nabla ^ 3 \ell (\bbeta) - \EE \{\nabla ^ 3 \ell (\bbeta)\}\big\| _ 2 \geq C_2 M K ^ 3 \bigg\{\bigg(\frac{t}{n}\bigg) ^ {1 / 2} + \frac{t ^ {3 / 2}}{n}\bigg\}\bigg] \\
	&\leq \PP \bigg[ \sup_{\ltwonorm{\bv} = 1} \bigg|\frac{1}{n} \sum_{i=1}^n b'''(\bx_i ^ \top \bbeta) (\bx_i ^ \top \bv)^ 3 - \EE \{b'''(\bx_i ^ \top \bbeta) (\bx_i ^ \top \bv) ^ 3\} \bigg| \geq C_2 M K ^ 3 \bigg\{\bigg(\frac{t}{n}\bigg) ^ {1 / 2} + \frac{t ^ {3 / 2}}{n}\bigg\}\bigg] \\
	&\le 2e ^ {-(t - 3 - p \log 17)}.
\end{aligned}
\]
Substituting $\xi = t - 3 - p \log 17$ into the bound with positive $\xi$, we derive the conclusion that with probability at least $1 - 2e ^ {- \xi}$ such that
\[
	\big\| \nabla ^ 3 \ell (\bbeta) - \EE \{\nabla ^ 3 \ell (\bbeta)\}\big\|_2
	\lesssim M K ^ 3 \bigg\{\bigg(\frac{p \vee \xi}{n}\bigg) ^ {1/2} + \frac{(p \vee \xi) ^ {3/2}}{n}\bigg\}.
\]
\end{proof}

\begin{lem}
\label{lem:b_four_prime}
	Under Condition \ref{con:distribution} and \ref{con:b_four_prime}, for any $\xi > 0$, we have with probability at least $1 - 2e ^ {-\xi}$ that
\[
	\big\|\nabla ^ 4 \ell ^ {(k)} (\bbeta) - \EE \{\nabla ^ 4 \ell ^ {(k)} (\bbeta)\}\big\|_2  
	\lesssim M K ^ 4 \bigg\{\bigg(\frac{p \vee \xi}{n}\bigg) ^ {1/2} + \frac{(p \vee \xi) ^ 2}{n}\bigg\},
\]
for any $\bbeta \in \RR ^ p$ and any $k \in [m]$. 
\end{lem}

\begin{proof}
For simplicity, we omit ``${(k)}$'' in the superscript in the following proof. Similar to \eqref{eq:b_three_prime2}, we have
\[
	\big\|\nabla ^ 4 \ell (\bbeta) - \EE \{\nabla ^ 4 \ell (\bbeta)\}\big\|_2 
	\leq 2 \!\!\!\!\!\! \sup_{\bv \in \cN(1/16)} \!\! \bigg|\frac{1}{n} \sum_{i=1}^n b''''(\bx_i ^ \top \bbeta) (\bx_i ^ \top \bv) ^ 4 - \EE \{b''''(\bx_i ^ \top \bbeta) (\bx_i ^ \top \bv) ^ 4\} \bigg|.
\]
Given any $\bv$, by sub-Gaussianity, we have
\[
    \EE \exp \bigg\{\bigg(\frac{|b''''(\bx_i ^ \top \bbeta) (\bx_i ^ \top \bv) ^ 4|}{M K ^ 4}\bigg) ^ {1 / 2}\bigg\}  \leq 2, 
\]
from which we deduce that $\|b''''(\bx_i ^ \top \bbeta) (\bx_i ^ \top \bv) ^ 4\| _ {\psi_{1 / 2}} \leq M K ^ 4$. Applying the bound above display (3.6) in \cite{Adamczak2009RestrictedIP}, we have for any $t > 0$ that 
\begin{equation}
\label{eq:b_four_prime2}
	\PP \bigg[\bigg|\frac{1}{n} \sum_{i=1}^n b''''(\bx_i ^ \top \bbeta) (\bx_i ^ \top \bv) ^ 4 - \EE \{b''''(\bx_i ^ \top \bbeta) (\bx_i ^ \top \bv) ^ 4\} \bigg| \geq C_1 M K ^ 4 \bigg\{\bigg(\frac{t}{n} \bigg) ^ {\!\!1 / 2} \!\!\!\! + \frac{t ^ 2}{n}\bigg\}\bigg] \leq 2e ^ {-(t - 3)}, 
\end{equation}
where $C_1$ is a universal constant. Applying a union bound over $\bv \in \cN(1 / 16)$ and then \eqref{eq:b_four_prime2}, we deduce that 
\[
	\PP \bigg[\big\|\nabla ^ 4 \ell (\bbeta) - \EE \{\nabla ^ 4 \ell (\bbeta)\}\big\|_2  \ge C_2 M K ^ 4 \bigg\{\bigg(\frac{t}{n} \bigg) ^ {1 / 2} + \frac{t ^ 2}{n} \bigg\} \bigg] 
	\leq 2e ^ {-(t - 3 - p \log 33)}.
\]
Substituting $\xi = t - 3 - p \log 33$ into the bound with positive $\xi$, we can find universal constants $C_3$ with probability at least $1 - 2e ^ {- \xi}$ such that 
\[
	\big\|\nabla ^ 4 \ell (\bbeta) - \EE \{\nabla ^ 4 \ell (\bbeta)\}\big\|_2 
	\leq C_3 M K ^ 4 \bigg\{\bigg(\frac{p \vee \xi}{n}\bigg) ^ {1/2} + \frac{(p \vee \xi) ^ 2}{n}\bigg\}.
\]
\end{proof}

\begin{lem}
\label{lem:b_five_prime}
	Under Condition \ref{con:distribution} and \ref{con:b_four_prime}, for any $\xi > 0$, we have with probability at least $1 - 2e ^ {-\xi}$ that
\[
	\big\|\nabla ^ 5 \ell ^ {(k)} (\bbeta) \big\|_2  
	\lesssim M K ^ 5 \bigg\{1 + \bigg(\frac{p \vee \xi}{n}\bigg) ^ {1/2} + \frac{(p \vee \xi) ^ {5 / 2}}{n}\bigg\},
\]for any $\bbeta \in \RR ^ p$ and any $k \in [m]$. 
\end{lem}

\begin{proof}
For simplicity, we omit ``${(k)}$'' in the superscript in the following proof. Similar to Lemmas \ref{lem:b_three_prime} and \ref{lem:b_four_prime}, for any $t > 0$ and $\bv \in \cN(1 / 32)$, we have with probability at least $2e ^ {-(t - 3)}$ that
\[
	\bigg|\frac{1}{n} \sum_{i=1}^n b'''''(\bx_i ^ \top \bbeta) (\bx_i ^ \top \bv) ^ 5 - \EE \{b''''(\bx_i ^ \top \bbeta) (\bx_i ^ \top \bv) ^ 5\} \bigg| \lesssim M K ^ 5 \bigg\{\bigg(\frac{t}{n} \bigg) ^ {\!\!1 / 2}  + \frac{t ^ {5 / 2}}{n}\bigg\}, 
\]
Note that $\EE \{b''''(\bx_i ^ \top \bbeta) (\bx_i ^ \top \bv) ^ 5\} \lesssim M K ^ 5$ and 
\[
	\big\|\nabla ^ 5 \ell (\bbeta)\big\|_2 
	\leq 2 \!\!\!\!\!\! \sup_{\bv \in \cN(1/32)} \bigg|\frac{1}{n} \sum_{i=1}^n b'''''(\bx_i ^ \top \bbeta) (\bx_i ^ \top \bv) ^ 5  \bigg|.
\]
Applying a union bound over $\bv \in \cN(1 / 32)$, we derive the result.
\end{proof}

\begin{lem}
\label{lem:center_hessian}
Under Conditions \ref{con:distribution}, \ref{con:b_double_prime} and \ref{con:b_four_prime}, if $n \ge \max(p, 4\log n)$, then we have with probability at least $1 - 4 n ^ {-4}$ that
\[
	\|\nabla ^ 2 \ell ^ {(k)}(\bbeta) - \bSigma\|_2
	\lesssim M K ^ 3 \|\bbeta - \bbeta ^ *\| _ 2 + M K ^ 2 \bigg(\frac{p \vee \log n}{n}\biggr) ^ {1 / 2},
\]
for any $\bbeta \in \RR ^ p$ and $k \in [m]$. If $\bbeta$ satisfies that $\|\bbeta - \bbeta ^ *\| _ 2 \leq  2 \kappa ^ {-1} (\phi M) ^ {1 / 2} K \{(p \vee 4 \log n)/n\} ^ {1 / 2}$, then we have with probability at least $1 - 4 n ^ {-4}$ that
\[
	\|\nabla ^ 2 \ell ^ {(k)}(\bbeta) - \bSigma\|_2
	\lesssim C_{\kappa, \phi, M, K} \bigg(\frac{p \vee \log n}{n}\biggr) ^ {1 / 2},
\]
where $C_{\kappa, \phi, M, K} = \kappa ^ {-1} \phi ^ {1 / 2} M ^ {3 / 2} K ^ 4 + M K ^ 2$.
\end{lem}

\begin{proof}
For simplicity, we omit ``${(k)}$'' in the superscript in the following proof. Note that
\begin{equation}
	\|\nabla ^ 2 \ell(\bbeta) - \bSigma\| _ 2
	\leq \|\nabla ^ 2 \ell (\bbeta) - \nabla ^ 2 \ell (\bbeta ^ *)\| _ 2  + \|\nabla ^ 2 \ell(\bbeta ^ *) - \bSigma\| _ 2.
\end{equation}
By Taylor's expansion, we have
\[
	\nabla ^ 2 \ell(\bbeta) = \nabla ^ 2 \ell (\bbeta ^ *) + \nabla ^ 3 \ell (\bbeta')(\bbeta - \bbeta ^ *),
\]
where $\bbeta' = t \bbeta + (1 - t) \bbeta ^ *$ for some $t \in [0,1]$. Applying Lemma \ref{lem:b_three_prime} with $\xi = 4 \log n$, we have with probability at least $1 - 2 n ^ {-4}$ that 
\begin{equation}
\label{eq:hess1}
		\|\nabla ^ 2 \ell (\bbeta) - \nabla ^ 2 \ell (\bbeta ^ *)\| _ 2 
		\leq \|\nabla ^ 3 \ell (\bbeta')\| _ 2 \|\bbeta - \bbeta ^ *\| _ 2 
		\lesssim M K ^ 3 \|\bbeta - \bbeta ^ *\| _ 2.
\end{equation}
Lemma \ref{lem:center_Sigma} with $\xi = 4 \log n$ yields that with probability at least $1 - 2 n ^ {-4}$ such that 
\begin{equation}
\label{eq:hess2}
	\|\nabla ^ 2 \ell(\bbeta ^ *) - \bSigma\|_2 
	\lesssim M K ^ 2 \bigg(\frac{p \vee 4 \log n}{n}\biggr) ^ {1 / 2}.
\end{equation}
Combining \eqref{eq:hess1} and \eqref{eq:hess2}, the conclusion thus follows. In addition, if $\bbeta$ satisfied that $\|\bbeta - \bbeta ^ *\| _ 2 \leq  2 \kappa ^ {-1} (\phi M) ^ {1 / 2} K \{(p \vee 4 \log n)/n\} ^ {1 / 2}$, then we have with probability at least $1 - 4 n ^ {-4}$ that
\[
	\|\nabla ^ 2 \ell (\bbeta) - \bSigma\|_2
	\lesssim \kappa ^ {-1} \phi ^ {1 / 2} M  ^ {3 / 2} K ^ 4 \bigg(\frac{p \vee 4 \log n}{n}\biggr) ^ {1 / 2} + M K ^ 2 \bigg(\frac{p \vee 4 \log n}{n}\biggr) ^ {1 / 2}.
\]
\end{proof}

\begin{lem}
\label{lem:avergae_est}
Under the same conditions as in Lemma \ref{lem:center_hessian}, if $\widehat \bbeta ^ {(k)}$ satisfies that $\|\widehat \bbeta ^ {(k)} - \bbeta ^ *\| _ 2 \leq  2 \kappa ^ {-1} (\phi M) ^ {1 / 2} K \{(p \vee 4 \log n)/n\} ^ {1 / 2}$ for all $k \in [m]$, we have 
\[
\begin{aligned}
	\bigg\|\frac{1}{m} \sum_{k = 1} ^ m (\widehat \bbeta ^ {(k)} - \bbeta ^ *)\bigg\|_2
	\lesssim &(\phi M) ^ {1 / 2} K \|\bSigma ^ {-1}\|_2 \bigg(\frac{p \vee \log n}{mn}\bigg) ^ {1 / 2} \\
	&+ (\kappa ^ {-2} \phi M ^ 2 K ^ 5 + \kappa ^ {-1} \phi ^ {1 / 2} M ^ {3 / 2} K ^ 3) \|\bSigma ^ {-1}\|_2 \bigg(\frac{p \vee 4 \log n}{n}\bigg),
\end{aligned}	
\]
with probability at least $1 - 2 n ^ {-4} - 8 m n ^ {-4}$.
\end{lem}

\begin{proof}
For any $k \in [m]$, by Taylor's expansion, we have
\[
\begin{aligned}
	\nabla \ell ^ {(k)} (\widehat \bbeta ^ {(k)}) 
	&= \nabla \ell ^ {(k)} (\bbeta ^ *) 
	+ \int_0 ^ 1 \nabla ^ 2 \ell ^ {(k)} (\bbeta ^ * + v (\widehat \bbeta ^ {(k)} - \bbeta ^ *)) dv (\widehat \bbeta ^ {(k)} - \bbeta ^ *) \\ 
	&= \nabla \ell ^ {(k)} (\bbeta ^ *) 
	+ \bSigma (\widehat \bbeta ^ {(k)} - \bbeta ^ *)
	+ \bigg\{\int_0 ^ 1 \nabla ^ 2 \ell ^ {(k)} (\bbeta ^ * + v (\widehat \bbeta ^ {(k)} - \bbeta ^ *)) dv - \bSigma\bigg\} (\widehat \bbeta ^ {(k)} - \bbeta ^ *) = 0.
\end{aligned}
\]
Some algebra yields that
\[
	\widehat \bbeta ^ {(k)} - \bbeta ^ *
	= - \bSigma ^ {-1} \nabla \ell ^ {(k)} (\bbeta ^ *) 
	- \bSigma ^ {-1} \bigg\{\int_0 ^ 1 \nabla ^ 2 \ell ^ {(k)} (\bbeta ^ * + v (\widehat \bbeta ^ {(k)} - \bbeta ^ *)) dv - \bSigma\bigg\} (\widehat \bbeta ^ {(k)} - \bbeta ^ *).
\]
Therefore, 
\[
	\frac{1}{m} \sum_{k = 1} ^ m (\widehat \bbeta ^ {(k)} - \bbeta ^ *)
	= - \bSigma ^ {-1} \nabla \ell (\bbeta ^ *) 
	- \frac{1}{m} \sum_{k = 1} ^ m \bSigma ^ {-1} \bigg\{\int_0 ^ 1 \nabla ^ 2 \ell ^ {(k)} (\bbeta ^ * + v (\widehat \bbeta ^ {(k)} - \bbeta ^ *)) dv - \bSigma\bigg\} (\widehat \bbeta ^ {(k)} - \bbeta ^ *).
\]
By \eqref{ineq:grad_upper_bound_full}, we have with probability at least $1 - 2 n ^ {-4}$ that
\[
	\|\bSigma ^ {-1} \nabla \ell (\bbeta ^ *)\|_2
	\lesssim (\phi M) ^ {1 / 2} K \|\bSigma ^ {-1}\|_2 \bigg(\frac{p \vee \log n}{mn}\bigg) ^ {1 / 2}.
\]
By Lemma \ref{lem:center_hessian}, we have with probability at least $1 - 4 m n ^ {-4}$ that
\[
\begin{aligned}
	\bigg\|\frac{1}{m} \sum_{k = 1} ^ m \bSigma ^ {-1} \bigg\{\int_0 ^ 1 \nabla ^ 2 \ell ^ {(k)} &(\bbeta ^ * + v (\widehat \bbeta ^ {(k)} - \bbeta ^ *)) dv - \bSigma\bigg\} (\widehat \bbeta ^ {(k)} - \bbeta ^ *)\bigg\|_2 \\
	&\lesssim \frac{1}{m} \sum_{k = 1} ^ m \|\bSigma ^ {-1}\|_2 \|\widehat \bbeta ^ {(k)} - \bbeta ^ *\|_2 \bigg\|\int_0 ^ 1 \nabla ^ 2 \ell ^ {(k)} (\bbeta ^ * + v (\widehat \bbeta ^ {(k)} - \bbeta ^ *)) dv - \bSigma\bigg\|_2 \\
	&\lesssim (\kappa ^ {-2} \phi M ^ 2 K ^ 5 + \kappa ^ {-1} \phi ^ {1 / 2} M ^ {3 / 2} K ^ 3) \|\bSigma ^ {-1}\|_2 \bigg(\frac{p \vee 4 \log n}{n}\bigg),
\end{aligned}
\]
Combining these two bounds, we derive
\[
\begin{aligned}
	\bigg\|\frac{1}{m} \sum_{k = 1} ^ m (\widehat \bbeta ^ {(k)} &- \bbeta ^ *)\bigg\|_2
	\!\!\leq \|\bSigma ^ {-1} \nabla \ell (\bbeta ^ *) \|_2
	\!+\! \bigg\|\frac{1}{m} \sum_{k = 1} ^ m \bSigma ^ {-1} \bigg\{\int_0 ^ 1 \nabla ^ 2 \ell ^ {(k)} (\bbeta ^ * + v (\widehat \bbeta ^ {(k)} - \bbeta ^ *)) dv - \bSigma\bigg\} (\widehat \bbeta ^ {(k)} - \bbeta ^ *)\bigg\|_2 \\
	&\lesssim (\phi M) ^ {1 / 2} K \|\bSigma ^ {-1}\|_2 \bigg(\frac{p \vee \log n}{mn}\bigg) ^ {1 / 2}
	+ (\kappa ^ {-2} \phi M ^ 2 K ^ 5 + \kappa ^ {-1} \phi ^ {1 / 2} M ^ {3 / 2} K ^ 3) \|\bSigma ^ {-1}\|_2 \bigg(\frac{p \vee 4 \log n}{n}\bigg),
\end{aligned}	
\]
with probability at least $1 - 2 n ^ {-4} - 4 m n ^ {-4}$.
\end{proof}

\subsection{Lemmas for the noisy phase retrieval problem}
In this section, we provide the proof of the technical lemmas for estimation in the noisy phase retrieval problem. 

\begin{lem}
\label{PRlem4}
Let $f: \RR^n \rightarrow \RR$ be $K$-Lipschitz function and $\bx$ be the standard normal random vector in $\RR^n$. Then for every $t \geq 0$, we have
\[
	\PP \{f(\bx) - \EE f(\bx) > t\} \leq \exp \left(-t^2 / 2 K^2\right).
\]
\end{lem}

\begin{proof}
    See Proposition 34 in \cite{vershynin2010introduction}.
\end{proof}

\begin{lem}\label{PRlem6}
Let $\bA$ be an $n \times p$ matrix whose entries are independent standard normal random variables. Then, with probability at least $1 - 2 \exp \left(-t ^ 2 / 2\right),$ the following inequality hold,
\[
	\|\bA\|_{2 \rightarrow 4} \leq (3n) ^ {1/4} + \sqrt{p} + t.
\]
\end{lem}

\begin{proof}
The proof follows that of Lemma A.5 in \cite{2015Optimal} step by step. Note that $\|\bx\|_4 = \max_{\|\bv\|_{4/3}=1}\langle\bx,\bv\rangle$. Define $X_{\bu, \bv} = \langle\bA\bu, \bv\rangle$ on
\[
	T = \big\{(\bu, \bv): \bu \in \RR^p, \|\bu\|_2 = 1, \bv \in \RR^n, \|\bv\|_{4/3} = 1\big\}.
\]
Then$\|\bA\|_{2 \rightarrow 4} = \max_{(\bu, \bv) \in T} X_{\bu, \bv}$. Define $Y_{\bu,\bv} = \langle\bg, \bu\rangle + \langle\bh, \bv\rangle$ where $\bg$ and $\bh$ are independent standard Gaussian random vectors of dimensions $p$ and $n$ respectively.
For any $(\bu, \bv),(\bu', \bv') \in T,$ we have
\[
	\EE |X_{\bu, \bv} - X_{\bu', \bv'}|^2 = \|\bv\|_2^2 + \|\bv'\|_2^2 - 2\langle\bu, \bu'\rangle\langle\bv, \bv'\rangle
\]
and
\[
	\EE|Y_{\bu, \bv} - Y_{\bu', \bv'}|^2 = 2 + \|\bv\|_2^2 + \|\bv'\|_2^2 - 2\langle\bu, \bu'\rangle - 2\langle\bv, \bv'\rangle.
\]
Therefore,
$\EE |Y_{\bu, \bv} - Y_{\bu', \bv'}|^2 - \EE |X_{\bu, \bv} - X_{\bu', \bv'}|^2 = 2\big(1 - \langle\bu, \bu'\rangle\big) \big(1 - \langle\bv, \bv'\big) \geq 0$. Then applying the Sudakov-Fernique inequality \citep[][Theorem~7.2.11]{Ver18}, we deduce that
\[
	\EE \|\bA\|_{2 \rightarrow 4}  \leq \EE \max_{(\bu, \bv) \in T} Y_{\bu, \bv} = \EE \|\bg \|_2 + \EE\|\bh\|_4 \leq \big(\EE\|\bg\|_2^2\big)^{1/2} + \big(\EE \|\bh\|_4^4\big)^{1/4} = \sqrt{p} + (3n)^{1 / 4}.
\]
Note that $\|\cdot\|_{2 \rightarrow 4}$ is a 1-Lipschitz function, i.e., $\|\bA - \bB\|_{2 \rightarrow 4} \le \fnorm{\bA - \bB}$ for any $\bA, \bB \in \RR ^ {n \times p}$. By Lemma \ref{PRlem4} there, it holds with probability at least $1 - 2\exp(-t ^ 2 / 2)$ that
\[
	\|\bA\|_{2 \rightarrow 4} \leq \sqrt{p} + (3n) ^ {1/4} + t.
\]
\end{proof}

\begin{lem}
\label{lem:pr_center_hess}
Suppose that $\bx_1, \ldots, \bx_n$ are independent observations of $\bx$ and $\bx \sim \cN (\mathbf{0}_p, \bI_p)$. For any $\bbeta \in \RR^p$, it holds with probability at least $1 - 14 e ^ {-\xi}$ that
\[
	\bigg\|\frac{1}{n} \sum_{i=1}^{n} (\bx_i ^ \top \bbeta) ^ 2 \bx_i \bx_i ^ \top - \EE (\bx_i ^ \top \bbeta) ^ 2 \bx_i \bx_i ^ \top\bigg\|_2 
	\lesssim \|\bbeta\|_2 ^ 2 \bigg\{\bigg(\frac{p \vee \xi}{n}\bigg) ^ {1/2} + \frac{(p \vee \xi) \xi}{n}\bigg\},
\]
provided that $n \ge C p ^ 2$ for some positive constant $C$. Moreover, the moment bound is given by
\[
	\bigg\{\EE \bigg(\bigg\|\frac{1}{n} \sum_{i=1}^{n} (\bx_i ^ \top \bbeta) ^ 2 \bx_i \bx_i ^ \top - \EE (\bx_i ^ \top \bbeta) ^ 2 \bx_i \bx_i ^ \top\bigg\|_2 ^ 2 \bigg) \bigg\} ^ {1/2} 
	\lesssim \|\bbeta\|_2 ^ 2 \bigg(\frac{p}{n}\bigg) ^ {1 / 2}.
\]
\end{lem}

\begin{proof}
Following the proof of Lemma 7.4 in \cite{candes2015phase}, we give an explicit bound here. By unitary invariance, it is enough to consider $\bbeta = \be_1$. For any $i \in [n]$, let $x_{i,1}$ denote the first element of $\bx_i$ and $\bx_{i,-1}$ denote the remaining elements. That is $\bx_i = (x_{i,1}, \bx_{i,-1})$. Then we have
\[
	\begin{aligned}
		&\bigg\|\frac{1}{n} \sum_{i=1}^{n} x_{i,1}^2 
    		\bigg(\begin{array}{cc}
        		x_{i,1}^2        & x_{i,1}\bx_{i,-1}^\top \\
        		x_{i,1}\bx_{i,-1} & \bx_{i,-1}\bx_{i-1}^\top
        	\end{array}\bigg)
  		- (\bI_p + 2\be_1 \be_1^\top)\bigg\|_2 \\
		&\leq \bigg\|\frac{1}{n} \sum_{i=1}^{n} x_{i,1}^2 
    		\bigg(\begin{array}{cc}
        		x_{i,1}^2 & 0 \\
        		0        & 0
        	\end{array}\bigg)
  		- \bigg(\begin{array}{cc}
        		3 & 0 \\
        		0 & 0
        	\end{array}\bigg)\bigg\|_2 
		+ \bigg\|\frac{1}{n} \sum_{i=1}^{n} x_{i,1}^2 
    		\bigg(\begin{array}{cc}
        		0               & x_{i,1}\bx_{i,-1}^\top \\
        		x_{i,1}\bx_{i,-1} & 0
        	\end{array}\bigg)
  		- \bigg(\begin{array}{cc}
        		0 & 0 \\
        		0 & 0
        	\end{array}\bigg)\bigg\|_2 \\
		&\quad+ \bigg\|\frac{1}{n} \sum_{i=1}^{n} x_{i,1}^2 
    		\bigg(\begin{array}{cc}
        		0 & 0 \\
        		0 & \bx_{i,-1}\bx_{i-1}^\top
       		\end{array}\bigg)
  		- x_{i,1}^2 \bigg(\begin{array}{cc}
        		0 & 0 \\
        		0 & \bI_{p-1}
        	\end{array}\bigg)\bigg\|_2 
		+ \bigg\|\frac{1}{n} \sum_{i=1}^{n} x_{i,1}^2 
    	\bigg(\begin{array}{cc}
        		0 & 0 \\
        		0 & \bI_{p-1}
        	\end{array}\bigg)
  		- \bigg(\begin{array}{cc}
        		0 & 0 \\
        		0 & \bI_{p-1}
        	\end{array}\bigg)\bigg\|_2 \\
		&\leq \bigg|\frac{1}{n} \sum_{i=1}^{n} x_{i,1}^4 - 3\bigg| + 2\bigg\|\frac{1}{n} \sum_{i=1}^{n} x_{i,1}^3 \bx_{i,-1}\bigg\|_2 + \bigg\|\frac{1}{n} \sum_{i=1}^{n} x_{i,1}^2 (\bx_{i,-1}\bx_{i-1} ^ \top - \bI_{p-1})\bigg\|_2 + \bigg|\frac{1}{n} \sum_{i=1}^{n} x_{i,1}^2 - 1\bigg|.
	\end{aligned}
\]
By display (3.6) in \cite{Adamczak2009RestrictedIP}, we have with probability at least $1 - 4 e ^ {- \xi}$ that
\begin{equation}
\label{eq:pr_hessian_term1}
	\begin{aligned}
		\bigg|\frac{1}{n} \sum_{i=1}^{n} x_{i,1}^4 - 3\bigg| 
			\leq C_1 \bigg\{\bigg(\frac{\xi}{n}\bigg) ^ {1/2} + \frac{\xi ^ 2}{n}\bigg\}
		\quad\text{and}\quad
		\bigg|\frac{1}{n} \sum_{i=1}^{n} x_{i,1}^2 - 1\bigg| 
			\leq C_2 \bigg\{\bigg(\frac{\xi}{n}\bigg) ^ {1/2} + \frac{\xi}{n}\bigg\},
	\end{aligned}
\end{equation}
where $C_1$ and $C_2$ are constants. Before we analyze the bounds for the second and third term, we define the following events:
\[
\begin{aligned}
	&\cA_1 = \bigg\{\Big|\frac{1}{n} \sum_{i=1}^{n} x_{i,1} ^ 6 - 15 \Big| \le  C_3\bigg(\frac{\xi}{n}\bigg) ^ {1/2} + C_3\frac{\xi ^ 3}{n}\bigg\}, \\
	&\cA_2 = \bigg\{\Big|\frac{1}{n} \sum_{i=1}^{n} x_{i,1} ^ 4 - 3 \Big| \le C_4\bigg(\frac{\xi}{n}\bigg) ^ {1/2} + C_4\frac{\xi ^ 2}{n}\bigg\}, \\
	&\cA_3 = \Big\{\big|\max_i(x_{i,1} ^ 2 - 1)\big| \le C_5(\log n + \xi)\Big\},
\end{aligned}
\]
where $C_3$, $C_4$ and $C_5$ are constants. To bound the second term, we also define the event $\cE_1$:
\[
	\cE_1 := \bigg\{\bigg\|\frac{1}{n} \sum_{i=1}^{n} x_{i,1}^3 \bx_{i,-1}\bigg\|_2 \ge  t_1\bigg\}
	\ \text{ with }\ 
	t_1 = C_6 \bigg\{15 + \bigg(\frac{\xi}{n}\bigg) ^ {1/2} + \frac{\xi ^ 3}{n}\bigg\} ^ {1/2} \bigg(\frac{p \vee \xi}{n}\bigg) ^ {1/2}.
\]
where $C_6$ is constant. We observe that
\[
	\bigg\|\frac{1}{n} \sum_{i=1}^{n} x_{i,1}^3 \bx_{i,-1}\bigg\|_2 
	= \sup_{\|\bv\|_2=1} \bigg|\frac{1}{n} \sum_{i=1}^{n} x_{i,1}^3 (\bx_{i,-1} ^ \top \bv) \bigg|
	\leq 2\sup_{\bv \in \cN(1/4)} \bigg|\frac{1}{n} \sum_{i=1}^{n} x_{i,1} ^ 3 (\bx_{i,-1} ^ \top \bv) \bigg|,
\]
where $\cN(1/4)$ is the $1/4$-net of the unit sphere $\cS^{p-2}$. Conditional on $\{x_{i,1}\}_{i=1} ^ n$, by Hoeffding's inequality, we have
\begin{equation}
	\begin{aligned}
		\PP \Big(\cE_1 \Big| \{x_{i,1}\}_{i=1} ^ n\Big) 
		&\le \PP \bigg(\sup_{\bv \in \cN(1/4)} \bigg|\frac{1}{n} \sum_{i=1}^{n} x_{i,1} ^ 3 (\bx_{i,-1} ^ \top \bv) \bigg| \ge \frac{t_1}{2} \Big| \{x_{i,1}\}_{i=1} ^ n\bigg) \\
		&\le 9^{p-1} \times 2\exp \bigg(-\frac{c_1 n t_1^2}{\frac{1}{n}\sum_{i=1}^{n}x_{i,1}^6} \bigg),
	\end{aligned}
\end{equation}
where $c_1$ is a constant. Note that
\[
	\PP (\cE_1 \cap \cA_1)
	= \EE \mathbbm{1}_{\cE_1 \cap \cA_1}
	= \EE [\EE \{\mathbbm{1}_{\cE_1} \mathbbm{1}_{\cA_1} | (x_{i,1})_{i \in [n]}\}] \leq 2 e ^ {- \xi}.
\]
Define  By display (3.6) in \cite{Adamczak2009RestrictedIP}, we have $\PP (\cA_1 ^ c) \le 2 e ^ {- \xi}$. Therefore, we have
\begin{equation}
\label{eq:pr_hessian_term2}
	\begin{aligned}
		\PP (\cE_1) 
		\leq \PP (\cE_1 \cap \cA_1) + \PP (\cA_1 ^ c) 
		\leq \PP (\cE_1 \cap \cA_1) + \PP (\cA_1 ^ c) 
		\le 4 e ^ {- \xi}.
	\end{aligned}
\end{equation}
As for the bound of the third term, we define the event $\cE_2$:
\[
	\cE_2 := \bigg\{\bigg\|\frac{1}{n} \sum_{i=1}^{n} x_{i,1}^2 (\bx_{i,-1}\bx_{i-1} ^ \top - \bI_{p-1})\bigg\|_2 \ge t_2\bigg\},
\]
with
\[
	t_2 = C_7 \bigg\{3 + \bigg(\frac{\xi}{n}\bigg) ^ {1/2} + \frac{\xi ^ 2}{n}\bigg\} ^ {1/2} \bigg(\frac{p \vee \xi}{n}\bigg) ^ {1/2}
		 + C_7(1 + \log n + \xi) \frac{p \vee \xi}{n},
\]
where $C_7$ is constant. We note that
\begin{equation}
	\begin{aligned}
		\bigg\|\frac{1}{n} \sum_{i=1}^{n} x_{i,1}^2 (\bx_{i,-1}\bx_{i-1} ^ \top - \bI_{p-1})\bigg\|_2 
		&= \sup_{\|\bv\|_2=1} \bigg|\frac{1}{n} \sum_{i=1}^{n} x_{i,1}^2 \{(\bx_{i,-1} ^ \top \bv)^2 - 1\} \bigg| \\
		&\leq 2\sup_{\bv \in \cN(1/4)} \bigg|\frac{1}{n} \sum_{i=1}^{n} x_{i,1}^2 \{(\bx_{i,-1} ^ \top\bv)^2 - 1\} \bigg|.
	\end{aligned}
\end{equation}
Conditional on $\{x_{i,1}\}_{i=1} ^ n$, by Bernstein's inequality, we have
\[
	\begin{aligned}
		\PP \Big(\cE_2 \ge t_2 \Big| \{x_{i,1}\}_{i=1} ^ n\Big) 
		&\le \PP \bigg(\sup_{\bv \in \cN(1/4)} \bigg|\frac{1}{n} \sum_{i=1}^{n} x_{i,1}^2 \{(\bx_{i,-1} ^ \top\bv)^2 - 1\} \bigg| \ge \frac{t_2}{2} \Big| \{x_{i,1}\}_{i=1} ^ n\bigg) \\
		&\leq 9 ^ {p-1} \times 2\exp\bigg\{-c_2 \min\bigg(\frac{n t_2 ^ 2}{\frac{1}{n}\sum_{i=1}^{n} x_{i,1} ^ 4}, \frac{nt_2}{\max_i \{x_{i,1}^2\}}\bigg)\bigg\}.
	\end{aligned}
\]
where $c_2$ is a constant. Similarly, we have
\[
	\PP (\cE_2 \cap \cA_2 \cap \cA_3) 
	= \EE \mathbbm{1}_{\cE_2 \cap \cA_2 \cap \cA_3}
	= \EE [\EE \{\mathbbm{1}_{\cE_2} \mathbbm{1}_{\cA_2 \cap \cA_3} | (x_{i,1})_{i \in [n]}\}]
	\leq 2 e ^ {- \xi}.
\]
By display (3.6) in \cite{Adamczak2009RestrictedIP}, we have $\PP (\cA_2 ^ c) \le 2 e ^ {- \xi}$ and $\PP (\cA_3 ^ c) \le 2 e ^ {- \xi}$. Therefore,
\begin{equation}
\label{eq:pr_hessian_term3}
	\begin{aligned}
		\PP (\cE_2) \leq \PP (\cE_2 \cap \cA_2 \cap \cA_3) + \PP (\cA_2 ^ c) + \PP (\cA_3 ^ c) \le 6 e ^ {- \xi}.
	\end{aligned}
\end{equation}
Combining \eqref{eq:pr_hessian_term1}, \eqref{eq:pr_hessian_term2} and \eqref{eq:pr_hessian_term3}, for some constant $C_8$, we have with probability at least $1 - 14 e ^ {-\xi}$ that
\[
	\bigg\|\frac{1}{n} \sum_{i=1} ^ n (\bx_i ^ \top \be_1) ^ 2 \bx_i\bx_i ^ \top - (\bI_p + 2\be_1 \be_1^\top)\bigg\|_2 
	\le C_8 \bigg\{\bigg(\frac{p \vee \xi}{n}\bigg) ^ {1/2} + \frac{(p \vee \xi) \xi}{n}\bigg\},
\]
provided that $p \le C_9 n / \log ^ 2n$ for some positive constant $C_9$. Let 
\[
	A := C_8\bigg(\frac{p}{n}\bigg) ^ {1/2} 
	\quad \text{and} \quad
	Z := \max \bigg(\bigg\|\frac{1}{n} \sum_{i=1} ^ n (\bx_i ^ \top \be_1) ^ 2 \bx_i\bx_i ^ \top - (\bI_p + 2\be_1 \be_1^\top)\bigg\|_2 - A, 0\bigg).
\]
Next, we derive the rate of moment bound of $Z$. We consider two situation: (a) $p ^ 3 < n$; (b) $p ^ 3 \ge n \ge p ^ 2$. Firstly, for (a), we have
\[
\begin{aligned}
	\EE (Z ^ 2) 
	&= \int_0 ^ \infty 2t \PP(Z > t) dt
	= \int_0 ^ \infty 2t \PP(Z > t) dt \\
	&\lesssim \int_0 ^ \infty \bigg\{\bigg(\frac{\xi}{n}\bigg) ^ {1/2} + \frac{(p \vee \xi) \xi}{n}\bigg\}
		\bigg\{\bigg(\frac{1}{\xi n}\bigg) ^ {1/2} + \frac{p \vee \xi}{n}\bigg\} e ^ {-\xi} d\xi \\
	&\lesssim \int_0 ^ {n ^ {1/3}} \frac{1}{n} e ^ {-\xi} d\xi
		+ \int_ {n ^ {1/3}} ^ \infty \frac{\xi ^ 3}{n ^ 2} e ^ {-\xi} d\xi 
	\lesssim \frac{1}{n}.
\end{aligned}
\]
where in the second step we substitute $t$ with $\big\{\big(\frac{\xi}{n}\big) ^ {1/2} + \frac{(p \vee \xi) \xi}{n}\big\}$. Similarly, for (b), we have
\[
\begin{aligned}
	\EE (Z ^ 2) 
	&= \int_0 ^ \infty 2t \PP(Z > t) dt
	= \int_0 ^ \infty 2t \PP(Z > t) dt \\
	&\lesssim \int_0 ^ \infty \bigg\{\bigg(\frac{\xi}{n}\bigg) ^ {1/2} + \frac{(p \vee \xi) \xi}{n}\bigg\}
		\bigg\{\bigg(\frac{1}{\xi n}\bigg) ^ {1/2} + \frac{p \vee \xi}{n}\bigg\} e ^ {-\xi} d\xi \\
	&\lesssim \int_0 ^ {n ^ {1/3}} \frac{p}{n} e ^ {-\xi} d\xi 
		+ \int ^ p _ {n ^ {1/3}} \frac{p ^ 2 \xi}{n ^ 2} e ^ {-\xi} d\xi
		+ \int_ p ^ \infty \frac{\xi ^ 3}{n ^ 2} e ^ {-\xi} d\xi 
	\lesssim \frac{p}{n}.
\end{aligned}
\]
Applying Stirling's approximation yields that $\{\EE(Z^2)\}^{1 / 2} \lesssim  (p/n) ^ {1 / 2}$. Therefore,
\[
\begin{aligned}
	\bigg\{\EE \bigg(\bigg\|\frac{1}{n} \sum_{i=1} ^ n (\bx_i ^ \top \be_1) ^ 2 \bx_i\bx_i ^ \top - (\bI_p + 2\be_1 \be_1^\top)&\bigg\|_2 ^ 2\bigg) \bigg\} ^ {1/2} 
	= [\EE\{(A+Z) ^ 2\}] ^ {1 / 2} \\
	&\leq A + \{\EE(Z^2)\}^{1 / 2} 
	\lesssim  (p/n) ^ {1 / 2}.
\end{aligned}
\]
\end{proof}

\begin{lem}
\label{lem:pr_center_veps}
Suppose that $\varepsilon_1, \ldots, \varepsilon_n$ are independent random variables and follow standard normal distribution. Suppose that $\bx_1, \ldots, \bx_n$ are $i.i.d.$ observations of  sub-Gaussian random vector $\bx$ valued in $\RR ^ p$ satisfying that $\EE \bx = \bzero$ and $\|\bx\|_{\psi_2} \le K$. Then, we have with probability at least $1 - 4e ^ {-\xi}$ that
\[
	\bigg\|\frac{1}{n} \sum_{i=1}^{n} \varepsilon_i \bx_i \bx_i ^ \top \bigg\|_2 
	\lesssim K ^ 2 \bigg\{\bigg(\frac{p \vee \xi}{n}\bigg) ^ {1/2} + \frac{(p \vee \xi) ^ {5/2}}{n ^ {3/2}}\bigg\}.
\]
\end{lem}

\begin{proof}
Note that 
\begin{equation}
\label{eq:pr_net_bound2}
	\bigg\|\frac{1}{n} \sum_{i=1}^{n} \varepsilon_i \bx_i \bx_i ^ \top \bigg\|_2
	= \sup_{\|\bv\|_2 = 1}\bigg\{\frac{1}{n} \sum_{i=1}^{n} \varepsilon_i (\bx_i ^ \top \bv) ^ 2\bigg\} 
	\leq 2 \sup_{\bv \in \cN(1/4)}\bigg\{\frac{1}{n} \sum_{i=1}^{n} \varepsilon_i (\bx_i ^ \top \bv) ^ 2\bigg\}.
\end{equation}
Conditional on $\bx_1, \ldots, \bx_n$ and applying Hoeffding's inequality, we have with probability at least $1 - 2e ^ {-t}$ that
\begin{equation}
\label{eq:pr_lem_bound1}
	\bigg|\frac{1}{n} \sum_{i=1}^{n} \varepsilon_i (\bx_i ^ \top \bv) ^ 2\bigg| 
	\leq C_1 \bigg(\frac{t}{n}\bigg) ^ {1/2} \times \bigg\{\frac{1}{n} \sum_{i} ^ n (\bx_i ^ \top \bv) ^ 4\bigg\} ^ {1/2},
\end{equation}
where $C_1$ is a constant. Given any $\bv$, by sub-Gaussianity, we have $\EE \exp [\{(\bx_i ^ \top \bv) ^ 4 / K ^ 4\} ^ {1 / 2}]  \leq 2$ from which we deduce that $\|(\bx_i ^ \top \bv) ^ 4\| _ {\psi_{1 / 2}} \leq K ^ 4$. By display (3.6) in \cite{Adamczak2009RestrictedIP}, we have
\[
	\PP \bigg[\bigg|\frac{1}{n}\sum_{i=1}^{n} (\bx_i ^ \top \bv) ^ 4 - \EE (\bx_i ^ \top \bv) ^ 4\bigg| \geq C_2 K ^ 4 \bigg\{\bigg(\frac{t}{n} \bigg) ^ {1 / 2} + \frac{t ^ 2}{n}\bigg\}\bigg] 
	\leq 2e ^ {-(t - 3)}, 
\]
where $C_2$ is a universal constant. Given that $\EE\{(\bx ^ \top \bv) ^ 4\} \lesssim K ^ 4$, for a constant $C_3$, we have with probability at least $1 - 2e ^ {-(t-3)}$
\[
    \frac{1}{n} \sum_{i=1}^{n} (\bx_i ^ \top \bv) ^ 4 \leq C_3 K ^ 4 \bigg\{1 + \bigg(\frac{t}{n} \bigg) ^ {1 / 2} + \frac{t ^ 2}{n} \bigg\}.
\]
Combining with \eqref{eq:pr_lem_bound1} and applying a union bound over $\bv \in \cN(1 / 4)$ in \eqref{eq:pr_net_bound2} , we have with probability at least $1 - 4e ^ {-(t - 3 -p \log5)}$
\begin{equation}
	\bigg\|\frac{1}{n} \sum_{i=1}^{n} \varepsilon_i \bx_i \bx_i ^ \top \bigg\|_2 
	\lesssim K ^ 2 \bigg\{\bigg(\frac{t}{n}\bigg) ^ {1/2} + \frac{t ^ {5/2}}{n ^ {3/2}}\bigg\}.
\end{equation}
Substituting $\xi = t - 3 - p \log 5$ into the probability with positive $\xi$ yields the claimed results.
\end{proof}

\begin{lem}
\label{lem:pr_nable3_f}
Suppose that $\bx_1, \ldots, \bx_n$ are i.i.d. observations of  sub-Gaussian random vector $\bx$ valued in $\RR ^ p$ satisfying that $\EE \bx = \bzero$ and $\|\bx\|_{\psi_2} \le K$. Following the same notation as in the proof of Theorem \ref{thm:reboot_phase_retrieval}, we have with probability at least $1 - 2 e ^ {-\xi}$ that 
\[
		\|\nabla ^ 3 \ell(\bbeta)\|_2 \lesssim K ^ 4 \|\bbeta\|_2 \bigg\{1 + \bigg(\frac{p \vee \xi}{n}\bigg) ^ {1/2} + \frac{(p \vee \xi) ^ 2}{n}\bigg\},
\]
for any $\bbeta \in \RR ^ p$.
\end{lem}

\begin{proof}
Recall that $\nabla ^ 3 \ell(\bbeta) = \frac{1}{n} \sum_{i=1}^n 6 (\bx_i ^ \top \bbeta) (\bx_i \otimes \bx_i \otimes \bx_i)$ and
\[
	\|\nabla ^ 3 \ell(\bbeta)\|_2 
	\leq 2 \sup_{\bu \in \cN(1/16)} \bigg|\frac{1}{n} \sum_{i=1}^n 6 (\bx_i ^ \top \bbeta) (\bx_i ^ \top \bu) ^ 3\bigg|.
\]
By sub-Gaussianity, we have $\|(\bx_i ^ \top \bbeta) (\bx_i ^ \top \bu) ^ 3\|_{\psi_{1/2}} \leq K ^ 4 \|\bbeta\|_2$. Applying display (3.6) in \cite{Adamczak2009RestrictedIP} and the fact that $\EE \{(\bx_i ^ \top \bbeta) (\bx_i ^ \top \bu) ^ 3\} \leq K ^ 4 \|\bbeta\|_2$, we have 
\[
\begin{aligned}
	\PP \bigg[\bigg|\frac{1}{n} \sum_{i=1}^n 6 (\bx_i ^ \top \bbeta) (\bx_i ^ \top \bu) ^ 3\bigg| \geq C_1 K ^ 4 \|\bbeta\|_2 \bigg\{1 + \bigg(\frac{t}{n}\bigg) ^ {1/2} + \frac{t ^ 2}{n}\bigg\}\bigg] \leq2 e ^ {-(t-3)},
\end{aligned}
\]
where $C_1$ is a universal constant. Applying a union bound over $\cN(1/16)$, we derive the desired probability bound. 
\end{proof}



\end{document}